%% file: main_new.tex
\documentclass[10pt,usenames,dvipsnames]{article} 
\usepackage{amsmath,amssymb}
\usepackage{natbib}
\usepackage{caption}
\usepackage[usenames]{color}
\usepackage{bm}
\usepackage{ying}
\usepackage{comment}
\usepackage{wrapfig}
\usepackage{multirow}
\usepackage{rotating}
\usepackage{fullpage}
\usepackage{authblk}
\usepackage{enumerate}
\usepackage{geometry}
\usepackage{parskip}
\usepackage{float}
\usepackage{subcaption}
\usepackage{tikz,tabularx}
\usetikzlibrary{patterns}
\usetikzlibrary{arrows}
\usetikzlibrary{tikzmark, positioning, fit, shapes.misc}

\usetikzlibrary{decorations.pathreplacing, calc}
\tikzset{brace/.style={decorate, decoration={brace}},
  brace mirrored/.style={decorate, decoration={brace,mirror}},
}

\newcolumntype{g}{>{\columncolor{red}}c}

\usepackage{graphicx,amssymb}
\usepackage{xcolor}

\usepackage[colorlinks,
            linkcolor=red,
            anchorcolor=blue,
            citecolor=blue
            ]{hyperref}
\usepackage{algorithm}
\usepackage{algorithmic}

\newcommand\Dominik[1]{{\color{orange}Dominik: ``#1''}}

\let\hat\widehat
\let\tilde\widetilde

\def \iid {\stackrel{\text{i.i.d.}}{\sim}}

\def \mgn {{\textnormal{mgn}}}
\def \relat {{\textnormal{rel}}}
 
\def \std {{\textnormal{std}}}
\def \stab {{\textnormal{stab}}}

\theoremstyle{plain}

\providecommand{\keywords}[1]
{
  \fontsize{9}{12}\selectfont
  \textbf{\textit{Keywords:}} #1
}

\usepackage{hyperref}
\def\@#1\@{\begin{align}#1\end{align}}
\def\$#1\${\begin{align*}#1\end{align*}}

\newcommand\Ying[1]{{\color{magenta}[Ying: #1]}}  
\newcommand\Naoki[1]{{\color{magenta}[Naoki: #1]}}

\usepackage{color-edits}
\addauthor[Ying]{ying}{blue}

\usepackage{xcolor}

\title{
Beyond Reweighting: \\ On the Predictive Role of Covariate Shift in Effect Generalization  
} 
\author[1]{Ying Jin\thanks{Reproduction code for data processing and analysis is available at \href{https://github.com/ying531/predictive-shift}{https://github.com/ying531/predictive-shift}.}}
\affil[1]{\normalsize Data Science Initiative \& Department of Health Care Policy, Harvard University}
\author[2]{Naoki Egami}
\affil[2]{Department of Political Science, Columbia University}
\author[3]{Dominik Rothenh\"ausler}
\affil[3]{Department of Statistics, Stanford University}
\date{\today}

\begin{document}

\maketitle

\begin{abstract}
Many existing approaches to generalizing statistical inference amidst distribution shift operate under the covariate shift assumption, which posits that the conditional distribution of unobserved variables given observable ones is invariant across populations. However, recent empirical investigations have demonstrated that adjusting for shift in observed variables (covariate shift) is often insufficient for generalization. In other words, covariate shift does not typically ``explain away'' the distribution shift between settings. As such, addressing the unknown yet non-negligible shift in the unobserved variables given observed ones (conditional shift) is crucial for generalizable inference. 

In this paper, we present a series of empirical evidence from two large-scale multi-site replication studies to support a new role of covariate shift in ``predicting'' the strength of the unknown conditional shift. 
Analyzing 680 studies across 65 sites, we find that even though the conditional shift is non-negligible, its strength can often be bounded by that of the observable covariate shift. However, this pattern only emerges when the two sources of shifts are quantified by our proposed standardized, ``pivotal'' measures. We then interpret this phenomenon by connecting it to similar patterns that can be theoretically derived from a random distribution shift model. Finally, we demonstrate that exploiting the predictive role of covariate shift leads to reliable and efficient uncertainty quantification for target estimates in generalization tasks with partially observed data. Overall, our empirical and theoretical analyses suggest a new way to approach the problem of distributional shift, generalizability, and external validity.
\end{abstract}

\fontsize{10}{12}\selectfont


\input{sec1_intro}

\input{sec2_motivation}

\input{sec3_bounding}

\input{sec4_generalization}

\section{Discussion}

In this work, we offer new insights on distribution shifts when inferring parameter estimates in a new site  based on data from one site and covariate data from the other one. By empirical benchmarking in large-scale replication projects, we find significant distributional shifts between sites. Moreover, approaches that only account for distribution shifts of observed covariates---thereby relying on the \textit{explanatory} role of covariate shift---are often insufficient for explaining discrepancies between sites.

Instead of using covariates in an explanatory fashion, we propose to use covariates in a \emph{predictive} fashion. More precisely, we suggest predicting the strength of the shift of unobserved conditional distribution based on the strength of the shift of observed covariates. We provide empirical evidence based on two large-scale replication studies and offer a theoretical justification under a random distribution shift model.

In our empirical applications, we show that our proposed prediction intervals maintain the desired coverage even in the presence of (unobservable) distributional shifts. While these intervals can sometimes be over-conservative, they offer a significant improvement over existing approaches. Our method compares favorably to worst-case approaches, which tend to be overly pessimistic and lead to excessively wide intervals.

Our empirical and theoretical findings open up  several exciting future avenues for research. First, real-world scenarios may involve more complex forms of distributional change than the one studied in this work. For instance, in settings with emulated target populations, it might be reasonable to consider  hybrid models where there is a combination of controlled (and potentially large) covariate shift   and random shifts that arise due to inevitable deviations. Developing optimal estimation procedures for such hybrid models would be a valuable contribution. Second, the non-negligible conditional shift suggests the importance of collecting data from diverse sources to properly address the ``distributional uncertainty'' component in estimation. Towards this goal, the investigation of distribution shift patterns can provide insights for an important methodological challenge:  determining how to prioritize data collection. For example, if there is a partial covariate shift and partial random shift, it may be beneficial to prioritize the collection of covariates most affected by the shift, rather than gathering more data across all variables equally.






\section*{Acknowledgments}
We appreciate excellent research assistance by Diana Da In Lee. Egami acknowledges financial support from the National Science Foundation (SES–2318659). Rothenh\"ausler acknowledges financial support from the Dieter Schwarz Foundation, the Dudley Chamber fund, and the David Huntington Foundation.

\newpage
\bibliographystyle{apalike}
\bibliography{reference}

\newpage
\appendix 
\input{appendix}

\end{document}

%% file: sec1_intro.tex

\section{Introduction}

Distribution shift is a central issue in generalizing statistical evidence from an observed (source) population to a new, at most partially observed (target) population, with significant implications in many domains. 
For instance, in the medical and social sciences, researchers/policymakers seek to leverage existing randomized control trials (RCTs) to estimate the treatment effect on a new cohort to guide clinical decisions or policy making~\citep{shadish2002, hotz2005predicting, imai2008misunderstandings, cole2010generalizing, Tipton:2013ew, bareinboim2016fusion, Deaton:2017kg}. However, the challenge lies in whether statistical methods can capture the changes between populations to produce credible predictions of target effects. 
 
To address the generalizability question, many statistical methods operate under assumptions positing that observed variables capture all distributional differences between populations. These assumptions can often be described as \emph{covariate shift}, that is, the distribution of covariates observed in both populations can change, while the conditional distribution of the outcomes (unobserved in the target population) given the observed covariates remains invariant. For example, the distribution of age, gender, and education can differ across  populations (e.g., due to convenience sampling), but the conditional treatment effect is the same for individuals with the same covariate profiles. Under this common assumption, adjusting for shift in the observed covariates, either by reweighting based on density ratios or estimating the heterogeneous covariate-outcome relationship~\citep{stuart2011use, tipton2014sample, miratrix2018worth, dahabreh2019generalizing, egami2023elements}, is sufficient for unbiased estimation of the target parameters. This common approach highlights the role of covariate shift in \emph{explaining away} the distribution shift. 

Given its popularity, a series of recent papers \citep{cai2023diagnosing,jin2023diagnosing,lu2023you} have empirically evaluated the performance of generalization estimators based on the covariate shift assumption by comparing them against experimental benchmark estimates. Although each paper focuses on different domains, a common yet somewhat surprising finding is that observed covariate shift often can only explain a small proportion of the distributional shift in real-world applications. This implies two pessimistic messages: (1) adjusting for observed covariate shift may be insufficient for generalization, and (2) the remaining, unobserved conditional shift (i.e., shift in the conditional distribution of the outcomes given the observed covariates) is ``larger'' than the observed covariate shift. 
As such, it remains unclear how the conditional shift may be addressed for effect generalization in practice even in well-controlled settings. 

\subsection{This work: the predictive role of covariate shift}

In this paper, we introduce a different role of covariate shift in \emph{predicting} the unknown shift in the conditional distribution for generalization (Figure~\ref{fig:overview}). The distribution shift between the source and target populations consists of the observed covariate shift and unobserved conditional shift, the latter being a key challenge in a generalization task. In contrast to existing approaches that either (i) assume there is no conditional shift, or (ii) establish worst-case bounds based on adversarial shift in the conditional distribution, we argue that the strength of covariate shift can \emph{bound} that of the unknown conditional shift. Exploiting this bounding relationship is useful in effect generalization with improved validity and efficiency. 

\begin{figure}[!t]
    \centering
    \includegraphics[width=0.9\linewidth]{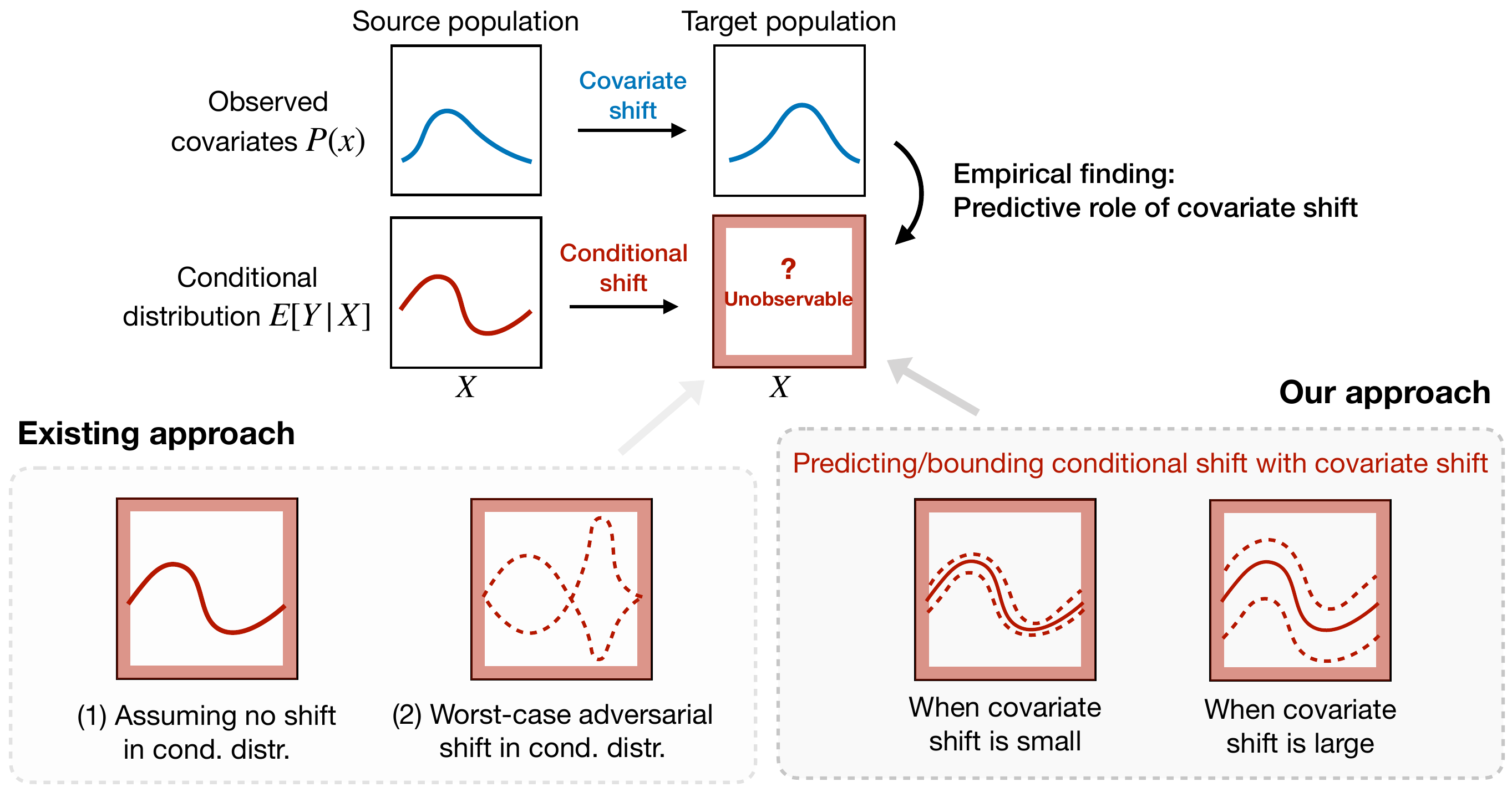}
    \caption{\textbf{Overview of the problem and our approach:}~\emph{Effect generalization from source and target populations needs to address the distribution shift consisting of the observed covariate shift and unobserved conditional shift. We argue a novel predictive role of covariate shift in bounding the strength of unknown conditional shift, which is supported by our empirical findings and leads to reliable and efficient generalization.}}
    \label{fig:overview}
\end{figure}




Our proposal is supported by empirical evidence from two well-known, large-scale multi-site replication projects---the Pipeline project \citep{schweinsberg2016pipeline} and the Many Labs 1 project~\citep{klein2014investigating}---from the social sciences, analyzing a total of 680 studies across 65 sites examining 25 hypotheses.\footnote{Note that not all sites examine all hypotheses.} To ensure faithful evaluation, since we have no access to the underlying population parameters, we build prediction intervals---based on various distribution shift assumptions---for \emph{estimators} in target populations (including our proposed ones built upon empirical findings), and use their empirical coverage to examine the plausibility of the assumptions they are based upon. Figure~\ref{fig:main} previews our empirical results.

\begin{figure}[!t]
    \centering
    \includegraphics[width=0.95\linewidth]{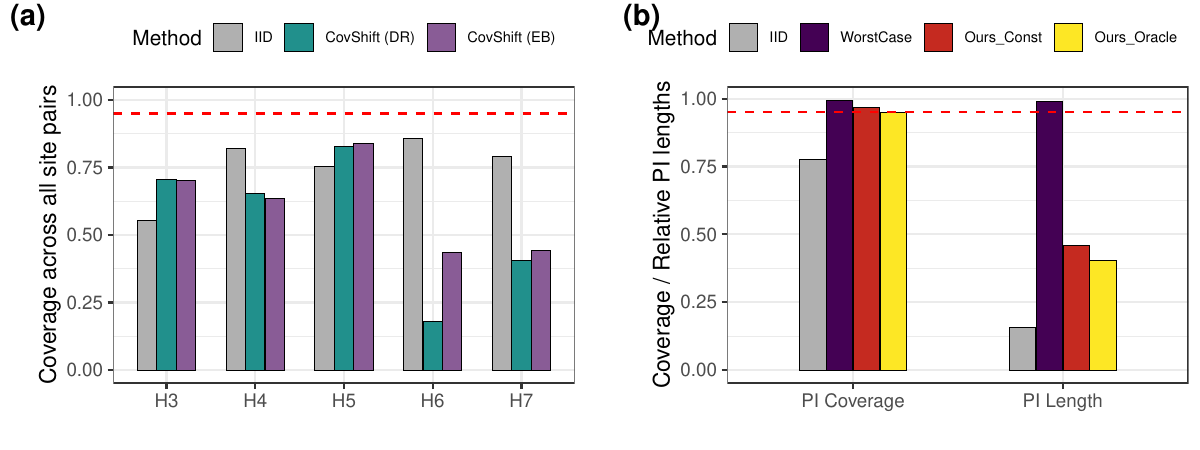}
    \caption{Preview of results. \textbf{[Left]~Insufficient explanatory role of covariate shift:}~\emph{Empirical coverage of prediction intervals based on i.i.d.~assumption (grey) and covariate shift assumption (green and purple), showing covariate shift cannot explain away distribution shift across sites.} \textbf{[Right]~Reliable and efficient effect generalization based on the predictive role of covariate shift:}~\emph{Empirical coverage of prediction intervals based on i.i.d.~assumption (grey), worst-case bounds (dark blue), and our method with the belief that conditional shift is bounded by covariate shift (red) or with knowledge of their relative strength (yellow).}}
    \label{fig:main}
\end{figure}

We begin by 
examining common approaches that either ignore distribution shift or assume covariate shift (Section~\ref{subsec:explain}). 
In the two replication projects, the \emph{explanatory} role of covariate shift is limited as evident from the low coverage of prediction intervals, complementing existing work that either examine pairs of studies~\citep{jin2023diagnosing} or mean squared errors~\citep{lu2023you,kern2016assessing}. 
As shown in Panel (a) of Figure~\ref{fig:main}, 
even for controlled multi-site replication studies, distribution shifts across sites are not negligible (methods that assume no distributional shift (\texttt{IID}) do not achieve valid coverage). 
Furthermore, observed covariate shift cannot explain away the total distributional shift, as methods that only adjust for observed covariate shift (\texttt{CovShift}) do not achieve valid coverage, either.  


We then proceed to compare the strengths of the observed covariate shift and the conditional shift (Section~\ref{subsec:bound}).
In stark contrast with the pessimistic conjectures in previous works, we find that conditional shift is often \textit{smaller} than covariate shift across different applications and comparisons. However, this empirical pattern became clear only after we measured covariate and conditional shifts with proper standardization. 
%

We interpret our empirical findings by connecting them to similar patterns that can be theoretically derived under a recently proposed random distribution shift model~\citep{jeong2022calibrated,jeong2024out,bansak2024learning} (Section~\ref{sec:interpret}). Under this model, one expects to observe smaller conditional shift than covariate shift when the probability space is randomly perturbed in a way that does not favor any direction yet some component of the observed data, which is the treatment assignment here, is kept invariant. This model describes scenarios where the difference between the source and target distributions is not adversarial but is contributed by many small and random factors.  Such scenarios are common in collaborative replication studies and potentially other carefully controlled studies where replicators try their best to mimic the original study design and population, but they have to deviate due to logistical and other constraints. 

Finally, we demonstrate the effectiveness of exploiting this predictive role in effect generalization, again (for evaluation purposes) by examining the empirical coverage of prediction intervals that aim to address the unknown conditional shift (Section~\ref{sec:generalize}).  
Panel (b) of Figure~\ref{fig:main} previews  key takeaway messages. Prediction intervals based on the novel predictive role maintain valid coverage while significantly shortening the intervals. This reveals that the predictive role is  stable across contexts and permits effective  empirical calibration. In contrast, existing methods assuming worst-case conditional shift (\texttt{WorstCase}) achieve valid coverage when the worst-case shift strength is (unrealistically) calibrated by data, but at the expense of too wide intervals. 


Overall, our empirical and theoretical analyses suggest a new way to approach the problem of distributional shift, generalizability, and external validity. Most existing methods either (i)  assume no shift in the unobserved conditional shift or (ii) assume shift in the unobserved conditional shift is bounded, and search for the worst-case scenarios that tend to be extremely adversarial. Instead, we offer a data-adaptive middle ground---shift in the unobserved conditional shift is non-negligible but is predictable from the observed covariate shift. 
Our results shall serve as the empirical and conceptual basis for developing new
methods and models beyond the covariate shift assumption.

\subsection{Scope of the paper}
\label{subsec:eval}

We note with caution that the main goal of this paper is to provide empirical and theoretical evidence for a new way of understanding real-world distribution shifts. The random distribution shift modeling assumption offers a perspective to justify our empirical findings, yet we do not anticipate it to be universally grounded. In particular, we limit the interpretation of our results to contexts similar to multi-site replication studies where data are collected in a ``natural'' manner, meaning that experimenters try to maintain consistency without adversarial patterns. In other words, the two projects provide a testbed for distribution shifts that emerge due to inevitable deviations despite well-controlled experimental settings~\citep{stroebe2014alleged,hudson2023explicating}.  Counter-examples include studies where the recruitment strategy changes. As an example, one study may be conducted on university students, whereas the second study may recruit only middle-aged participants. In this case, the random shift assumption may not be appropriate. 


We also note that our evaluation mainly focuses on uncertainty quantification, that is, whether statistical methods can produce reliable prediction intervals for the actual estimates from data in the target population. 
Focusing on prediction intervals is inevitable since the underlying super-population parameter is not accessible for evaluation purposes. 
In addition, uncertainty quantification offers a more comprehensive assessment than evaluating the consistency or unbiasedness of point estimates (see Section~\ref{subsec:related_work} for more discussion).

\subsection{Related work}
\label{subsec:related_work}

\paragraph{Re-weighting in causal inference.}

Using re-weighting to generalize from one population to another population has a long history in causal inference. Early examples include
Horvitz-Thompson \citep{horvitz1952generalization} and H\'ajek's estimator. Inverse probability weights are often unstable in practice. This has spurred the development of procedures that use outcome models to reduce variance \citep{robins1994estimation} and balancing weight procedures that penalize the weights \citep{deville1992, hainmueller2012entropy}. Modern re-weighting procedures were used to generalize the results of experiments from one site to another \citep[e.g.,][]{cole2010generalizing, stuart2011use, Tipton:2013ew, hartman2015sate, buchanan2018ipsw, dahabreh2019generalizing, dahabreh2020extending, egami2021covariate}. See \citet{degtiar2023review} 
 and \citet{colnet2024causal} for recent reviews.



\vspace{-0.5em}
\paragraph{Empirical evaluation of generalization.} 

This work adds to several recent works empirically evaluating generalization procedures that use unit-level data to generalize from one site to another. \citet{cai2023diagnosing} diagnose how much of the drop of prediction performance can be attributed to covariate shift vs concept $Y|X$ shift. \citet{jin2023diagnosing} and \citet{lu2023you} investigate how much of the discrepancy between causal effect estimates in different sites is due to unit-level covariates, among other factors. In welfare-to-work experiments, \cite{lu2023you} found that less than 10\% of discrepancies between sites is explained by changes in covariate distributions. This work echoes these works on the insufficient explanatory role of covariate shift. An important distinction is that our evaluation leverages the coverage of prediction intervals over many replication studies, which offers more comprehensive and faithful evaluation than methods that evaluate one pair of studies for a hypothesis~\citep{jin2023diagnosing,cai2023diagnosing} or examine the mean squared errors~\citep{lu2023you,kern2016assessing}. For example, while \cite{kern2016assessing} find in another multi-site replication dataset that covariate adjustment leads to \emph{unbiased} estimators (with bias averaged over multiple sites) for target estimates, it may still underestimate the variability if the conditional shift leads to discrepancies that are mean zero when averaged over studies but have non-negligible magnitude. 
More importantly, we also investigate a novel predictive role of covariate shift that can inform reliable generalization in practice. 

\vspace{-0.5em}
\paragraph{Heterogeneity and meta-analysis in replicability.}
Multi-site replication projects have  been used to examine the heterogeneity in effect estimates across sites~\citep{klein2018many,coppock2018generalizability,mcshane2022modeling,delios2022examining,holzmeister2024heterogeneity}. A prominent distinction is that these works often measure certain global notions of heterogeneity via meta-analysis~\citep{mcshane2022modeling}, while we focus on generalization from one site to another. Methodologically, our generalization methods are applicable when data from only the source and target sites are available, whereas meta-analysis needs data from many sites. In addition, these works provide echoing messages for weak explanatory roles of observed factors~\citep{klein2018many,delios2022examining} or complementary messages for design and estimation uncertainty~\citep{krefeld2024exposing,holzmeister2024heterogeneity}; the latter may be interpreted as ``random'' shifts if not documented.

\vspace{-0.5em}
\paragraph{Covariate and conditional shift in machine learning.}

The term covariate shift was first introduced by \cite{shimodaira2000improving}, and has become one of the standard domain adaptation models, see \citet{quinonero2008dataset} and \citet{pan2009survey}. Most commonly, covariate shift is addressed via importance weighting with the density ratio, which can be estimated directly, e.g.,~via a classifier \citep{bickel2007discriminative}. Similarly, density ratio reweighting is a standard approach to addressing covariate shift for statistical estimation and inference.
The conditional shift we study is related to the notion of concept drift in machine learning~\citep{gama2014survey,lu2018learning}. The techniques for addressing these shifts in prediction problems serve distinct goals than our estimation and inference problems. 



%% file: sec2_motivation.tex

\section{Motivating Applications and Methodological Problem}

We introduce our motivating applications and illustrate the core methodological challenges in generalization. 

\subsection{Motivating Applications: Multi-Site Replication Projects}
\label{sec:dataset}
In this paper, we use two large-scale multi-site replication projects from the social sciences to empirically investigate the role of covariate shifts in generalization. The Many Labs 1 project~\citep{klein2014investigating} evaluates the replicability of 13 classic and contemporary experimental findings in the social sciences, ranging from gain versus loss framing~\citep{tversky1981framing} to sex differences in implicit attitudes toward math~\citep{nosek2002math}, across 36 independent data collection sites. Similarly, in the Pipeline project~\citep{schweinsberg2016pipeline}, 25 laboratories across the world (contributing $29$ populations) independently replicate experiments for 10 scientific hypotheses concerning moral judgment, which is a well-known theory in psychology. Combining the two replication projects, we analyze 680 studies across 65 sites, examining 25 research hypotheses. This scale and diversity allow us to assess the proposed new role of covariate shifts across diverse empirical settings.

Several features of these multi-site replication projects make them suitable for evaluating distribution shifts in generalization. First, we can mimic the real-world generalization task by generalizing an effect estimate from one source site to another target site. Unlike the real generalization task, we have access to the effect estimate from the target site, and therefore, we can empirically evaluate the performance of common generalization estimators based on the covariate shift assumption and our proposed estimator, without simulating data from the artificial data-generating process. Second, in these replication projects, multiple laboratories follow the same experimental process as much as they can, known as direct replications. As a result, the measurement of the outcome variable and treatment variable is consistent across sites, and the interpretation of the covariate shift and the unobserved conditional shift becomes clearer. Finally, the two replication projects differ in how laboratories are recruited. In the Pipeline project~\citep{schweinsberg2016pipeline}, laboratories are invited by the project lead because they had ``access to a subject population in which the original finding was theoretically expected to replicate using the original materials'' (p 57). Therefore, sites were selected such that distributional shifts between them are expected to be small or negligible. On the other hand, in the Many Labs 1 project~\citep{klein2014investigating}, laboratories voluntarily participated in the project without specific eligibility criteria related to whether each site was expected to replicate the original finding. Here, sites were selected conveniently but ``naturally'' without explicit intention. This variation in site selection enables us to empirically evaluate distributional shifts in diverse scenarios. 

The datasets are processed based on the raw data and scripts published by the original authors. In both projects, the covariates include demographic variables  such as political ideology, gender, age, education and income. See Appendix~\ref{app:sec_data} for details about the datasets and data pre-processing.

\subsection{Notation and Setup}
To formally discuss the generalization problem, we introduce some notation. While we tailor our notation to the two projects above for concrete presentation, the same general framework can be applied to any generalization setting across sites. 

We first index the hypotheses by $k\in \{1,\dots,K\}$ and the sites by $j\in \{1,\dots,N\}$. Each hypothesis $k$ is tested by a randomized experiment in a subset of sites $\cJ_k\subseteq \{1,\dots, N\}$, following the same experimental protocol. Each site $j \in \cJ_k$ independently collects $n_{j}^{(k)}\in \NN$ participants and collects data $\cD_j^{(k)} = \{X_i^{(j,k)},T_i^{(j,k)},Y_i^{(j,k)}\}_{i=1}^{n_j^{(k)}}$, where $X_i$ is the covariates, $T_i\in \{0,1\}$ is the binary treatment, and $Y_i$ is the outcome(s). Then, within each site, we can define the parameter of interest $\theta_j^{(k)}$ and its consistent and asymptotically normal estimator $\hat\theta_j^{(k)}$, which is a function of $\cD_j^{(k)}$. In our applications, most of them consider the average treatment effect (ATE) as $\theta_j^{(k)}$ and use a $t$-test that compares the sample mean of treated and control groups as $\widehat{\theta}_j^{(k)}$. Some hypotheses are tested with $\theta_j^{(k)}$ being  the mean of outcomes and $\widehat{\theta}_j^{(k)}$ being a paired $t$-test comparing two outcomes. The specific hypotheses and tests are summarized in Tables~\ref{tab:summary_of_studies} and~\ref{tab:summary_of_studies_ml1}. 

We assume $\cD_j^{(k)}$ are drawn i.i.d.~from an underlying (hypothetical) super-population $\cS_j^{(k)}$, and datasets are independent across sites $j\in\cJ_k$ for each hypothesis $k$. Importantly, the underlying data generating process $\cS_j^{(k)}$ may vary across sites $j\in \cJ_k$ since there might exist distribution shifts. 

We consider the generalization of estimates from site $j_1$ to $j_2$ for all pairs $j_1, j_2\in \{1,\dots, N\}$, $j_1\neq j_2$, in each application. In general, we call the population in site $j_1$ as the source population $P$ and the population in site $j_2$ as the target population $Q$. As typically the case in practice, for a \emph{generalization task}, we assume all data from $P$ are observed while only covariates $X$ are observed from $Q$. When we \emph{evaluate} the performance of various generalization estimators, we will use the full data in the target population $Q$ to empirically evaluate how well the generalization estimators approximate the benchmark estimates in $Q$.

\subsection{Challenge: Covariate Shift Cannot Explain Away Distributional Shift}
\label{subsec:explain}

The vast majority of existing methods for generalization assume that accounting for distributional shifts in observed covariates is sufficient, known as the covariate shift assumption. For example, when researchers want to generalize causal effects in one site to another site in the Pipeline project, they may assume that adjusting for observed characteristics of respondents, such as political ideology, gender, age, and education, is sufficient for generalization (consistent estimation and valid inference for the parameter in the target site). 



However, in line with recent empirical evaluations \citep{cai2023diagnosing,jin2023diagnosing, lu2023you}, we find that this common assumption of covariate shift is often insufficient to explain away distributional shifts in the real-world applications. Figure~\ref{fig:explanatory} examines existing procedures that adjust for shift in observed covariates. We consider generalizing treatment effects from one site to another, using two commonly used estimators---the doubly robust (DR) estimator~\citep{robins1994estimation, dahabreh2019generalizing} and the entropy balancing (EB) estimator~\citep{sarndal2003model, hainmueller2012entropy}---to construct  point estimates that are consistent for the target \emph{parameter} under the covariate shift assumption. Then, we follow~\cite{jin2024tailored} to construct prediction intervals that would cover the target \emph{estimator} with probability $1-\alpha$ under covariate shift, and evaluate their empirical coverage.\footnote{We use prediction intervals rather than the conventional confidence intervals because we only have access to target population estimates (instead of the underlying parameters) for rigorous evaluation purposes.} As a simple baseline, we also compute prediction intervals based on the i.i.d.~assumption that assumes no distribution shift between sites. Detailed estimation procedures are deferred to Appendix~\ref{app:subsec_est_explain}. Figure~\ref{fig:explanatory} highlights two key findings: 
\vspace{0.5em}
\begin{enumerate}[(i)]
    \item \emph{Adjusting for distribution shift is necessary}, as prediction intervals based on the assumption of no distribution shift (denoted as \texttt{IID}) do not deliver valid coverage (grey bars in panel (a)). 
    \item \emph{The explanatory role of covariate shift is insufficient.} This is evident from the under-coverage in panel (a) of both of the two \texttt{CovShift} methods. The coverage is sometimes even lower than \texttt{IID}; this is because the uncertainty that remains after adjustment is under-estimated. When comparing the estimates in the source population and generalization estimates in panel (b), we see that adjusting for covariate shift does not necessarily bring the estimators closer to the target estimate. 
\end{enumerate}

\begin{figure} 
    \hspace{1ex}
    \includegraphics[width=0.95\linewidth]{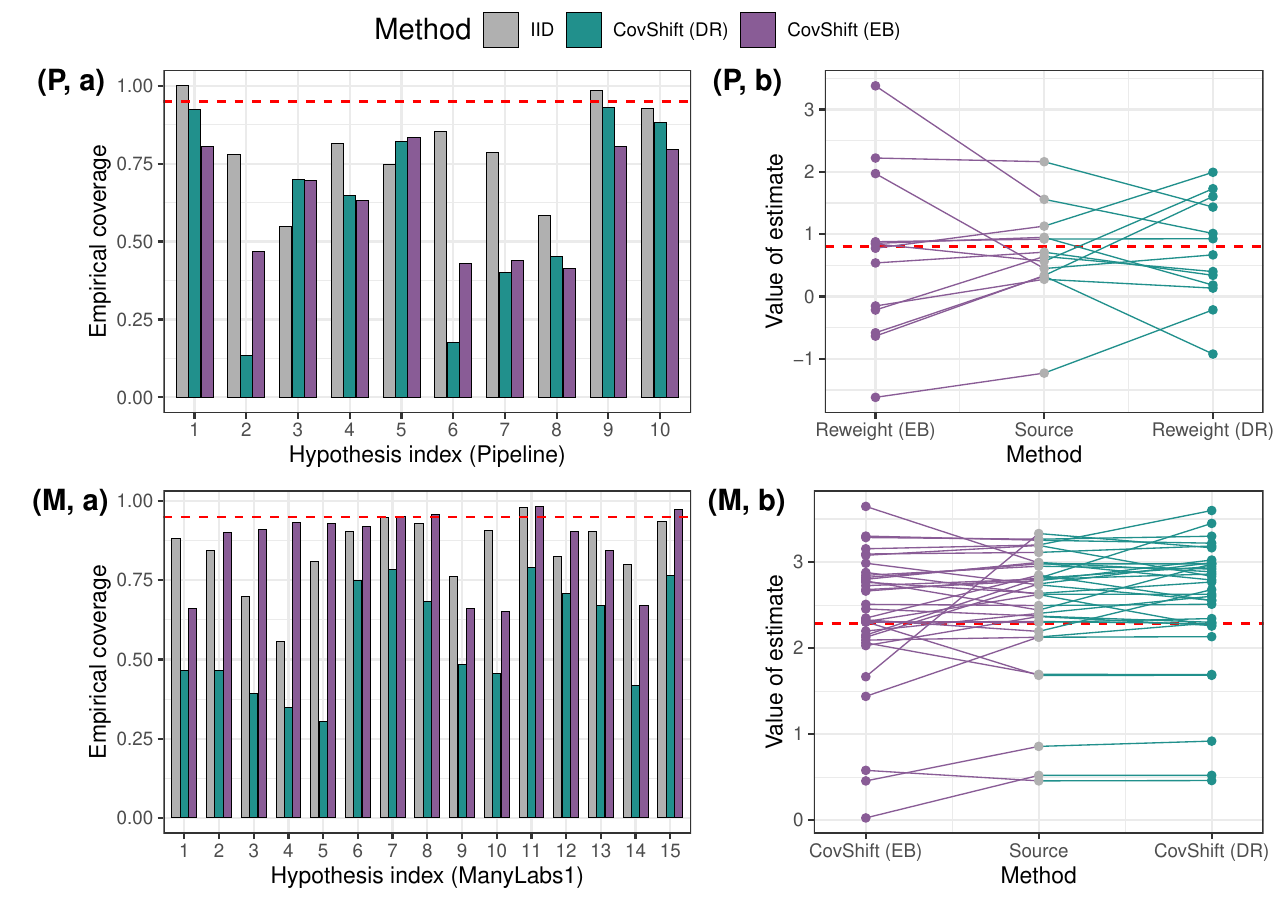}
    \caption{\textbf{Insufficient explanatory role of covariate shift}. \textbf{[Left]}: \textit{Under-coverage of $95\%$ prediction intervals based on the i.i.d.~assumption (grey) and covariate shift assumption  adjusted via doubly robust estimator (green) and entropy balancing (purple), averaged over all pairs of sites within each hypothesis for the Pipeline project} (P,\,a) \emph{and the ManyLabs 1 data} (M,\,a)\emph{, respectively. The red dashed line is the nominal level.} 
    \textbf{[Right]}: \textit{Estimates based on existing approaches (via doubly robust estimator (green) and entropy balancing (purple)) do not bring the source estimates (grey) closer to the target estimate (red dashed line). As illustrative examples, we show results when generalizing from all other sites to site 5 (raw ID) in hypothesis 5 in the Pipeline data} (P,\,b) \emph{and when generalizing from all other sites to site 4 in hypothesis 4 in ManyLabs 1 data} (M,\,b). \emph{The segments connect estimates for the same pairs of sites.}}
    \label{fig:explanatory}
\end{figure}

%% file: sec3_bounding.tex

\section{The Predictive Role of Covariate Shift}
\label{sec:bound}

In this paper, we highlight a new role of covariate shifts: observed covariate shifts can be used to \textit{predict} unobserved shifts in the conditional distribution of $Y$ given $X$, even though covariate shifts cannot fully explain the total distributional shift. We first propose standardized measures of distributional shifts, and then provide empirical and theoretical evidence for the predictive role of covariate shift.

\subsection{Comparing the Strength of Covariate Shift and Conditional Shift}
\label{subsec:def_shift_measure}


We begin by defining our measures of the two sources of distribution shifts: (i) the covariate shift in $X$ (the part commonly addressed in existing methods) and (ii) the conditional shift---the shift in the conditional distribution of $Y$ given $X$ (the part assumed away under the covariate shift assumption). 
Our construction is based on two simple principles:
\vspace{0.5em}
\begin{itemize}
    \item \emph{Scale invariance.} We would like our measures to reflect the strength of perturbations to the probability space, hence they should be invariant under linear scalings of the variables. 
    \item \emph{Numerical stability.} We would like our measures to be useful in guiding real generalization tasks, hence they should permit stable estimation. 
\end{itemize}
\vspace{0.5em}

Throughout the paper, we suppose the goal is to understand how causal effects change across sites, and we have two randomized experiments with treatment assignment probability $\pi$ (most studies in our datasets are of this form). We can write the the difference of the causal effects across sites as
\begin{align*}
    \theta_Q - \theta_P &= \EE_Q[\phi(T, Y)] - \EE_P[\phi(T, Y)] \\
    \text{ where } \phi & = \frac{T}{\pi}Y - \frac{1-T}{1-\pi}Y,
\end{align*}
where $\E_{P}$ and $\E_{Q}$ are expectations over the source and target distribution. While we focus our discussion on causal effects in this paper for the sake of clear presentation, our proposed approach is applicable to any parameter of interest by redefining $\phi$. For example,  some studies in the Pipeline project use a one-sample $t$-test, in which case the parameter of interest is the mean of the outcome and $\phi = Y$. 

We begin by conceptually decomposing the impact of overall distribution shift on the parameter of interest ($\theta_Q-\theta_P$) to measure the shifts in $X$ and $\phi$ given $X$ separately: 
\begin{equation}\label{eq:decomp}
    \mathbb{E}_Q[\phi] - \mathbb{E}_{P}[\phi] =  \underbrace{\big\{  \mathbb{E}_Q[\phi_P(X)]  - \mathbb{E}_P[\phi_P(X)]\big\}}_{=:\text{ \small Covariate shift}} \ + \ \underbrace{\big\{ \mathbb{E}_Q[\phi - \phi_P(X)] \big\}}_{=:\text{ \small Conditional shift}}, 
\end{equation}
where $\phi_P(X) := \EE_P[\phi|X]$ is the conditional expectation of the influence function in the source distribution. When the parameter of interest is the average treatment effect (ATE), we have $\phi_P(X) = \E_P[Y(1) - Y(0) \mid X]$, the conditional ATE function. In \cite{jin2023diagnosing}, the decomposition~\eqref{eq:decomp} is used to diagnose the roles of different  distribution shifts on the discrepancy of effect estimates between a pair of studies. 

The first ``Covariate shift'' term in the decomposition~\eqref{eq:decomp} captures the shift in the observed covariates $X$. Intuitively, it measures how much the estimate can be brought closer to the target by adjusting for the shift in $X$. This term becomes larger when the strength of shift between $P(X)$ and $Q(X)$ is larger. Importantly, it also depend on the heterogeneity in $\phi_P(X)$, that is, how much the parameter of interest varies with the covariates. Our proposed distribution shift measures will remove the impact of such heterogeneity (sensitivity) on our measure of the strength of distribution shift to ensure interpretability and scale invariance.

The second term in~\eqref{eq:decomp} is equal to $\mathbb{E}_Q[\phi_Q(X) - \phi_P(X)]$, which captures the shift in the conditional expectation $\EE[\phi|X]$ between the source and target distribution. For example, when the parameter of interest is the average treatment effect (ATE), this part captures how much the conditional ATE changes between the source and target distribution. Similarly, it not only depends on the strength of conditional shift but also the heterogenity in $\phi-\phi_P(X)$; again, the latter will be removed in our meausures.

The common assumption of covariate shift  essentially assumes away the second shift in the conditional distribution and only accounts for the first term. We formalize the covariate shift assumption  as follows. 
\begin{assumption}[Covariate Shift]\label{def:cov_shift}
$P(\phi\given X) = Q(\phi\given X)$ holds $P_X$-almost surely. 
\end{assumption}

If $\phi=Y$, this assumption is the classical covariate shift assumption in machine learning. For experiments, Assumption~\ref{def:cov_shift} is satisfied if the treatment probabilities do not change and the conditional distribution of the potential outcomes is invariant, i.e., if $P(Y(1),Y(0)|X) = Q(Y(1),Y(0)|X)$.

Assumption~\ref{def:cov_shift} implies the second term in~\eqref{eq:decomp} is zero, and thus it suffices to adjust for the shift in observed covariates (the first term). While this is a commonly imposed assumption for the identifiability of target parameters, as discussed in Section~\ref{subsec:explain}, it is often violated in practice, which implies that the conditional shift (the second term) is often nonzero in real-world applications. Therefore, instead of assuming away the conditional shift, we are to carefully investigate the relationship between these two shifts to offer new insights for moving beyond the covariate shift assumption in practice.  


We define our distribution shift measures by rescaling the two terms in~\eqref{eq:decomp} by their standard deviation to ensure scale invariance: 
\@
    \text{Relative conditional shift} &= \frac{|\EE_Q[\phi - \phi_P(X)]|}{\text{sd}_P(\phi - \phi_P(X))}, \label{def:rel_cond_shift} \\
    \text{Relative covariate shift} &= \frac{|\EE_Q[\phi_P(X)] - \EE_P[\phi_P(X)]|}{\text{sd}_P(\phi_P(X))}.\label{def:rel_cov_shift}
\@
We will measure the strength of the conditional shift by the ``relative conditional shift''~\eqref{def:rel_cond_shift}. However, an issue with the ``relative covariate shift'' measure~\eqref{def:rel_cov_shift} is numerical instability whenever $\text{sd}_P(\phi_P(X))$ is close to zero. This might be problematic in social science applications  where the explanatory power of covariates $X$ can be low. 
To address this issue, we will use a Mahalanobis-type, ``stabilized'' measure instead:
\@\label{def:stab_cov_shift}
    \text{Stabilized covariate shift measure} = \sqrt{\frac{1}{L} \sum_{\ell=1}^L \frac{(\EE_Q[X_\ell] - \EE_P[X_\ell])^2}{\text{Var}_P(X_\ell)}},
\@
where $L$ is the number of covariates. We justify this covariate shift measure from a theoretical perspective in Section~\ref{sec:interpret}. Importantly, this measure is also invariant under the scaling of features. 

\begin{remark}\normalfont
    We note that both (i) rescaling by standard deviation for scale invariance and (ii) adopting stabilized measure of covariate shift are crucial for interpretable and robust empirical insights. We illustrate the importance of these considerations through an  ablation study  in Appendix~\ref{app:alt_shift_measures} which explores alternative distribution shift measures without these elements.  These alternative distribution shift measures either fail to induce the predictive role or lead to much wider intervals in generalization due to numerical instablity.
\end{remark}

In our evaluations, the two population measures will be replaced by their estimators. The estimation details are deferred to Appendix~\ref{app:sec_estimation} with specific references in the corresponding parts of the paper.

\subsection{Empirical Evidence: Covariate Shift Can Bound Conditional Shift}
\label{subsec:bound}

Using data from both the Pipeline project and the ManyLabs1 project, we establish empirical evidence that with our distribution shift measures, the covariate shift can \emph{bound} the conditional shift, even though the strength of both may change across hypotheses and sites. Because the covariate shift is  estimable in common generalization tasks, researchers can use this bounding relationship to \textit{predict} the conditional shift, which is usually unobserved. We provide theoretical justification for the empirical findings in the next subsection.

We estimate the two distribution shift measures for any pair of sites for each hypothesis in the Pipeline project and the ManyLabs1 project. For any given hypothesis $k$, we define $\phi$ following the original analysis (c.f.~Tables~\ref{tab:summary_of_studies} and~\ref{tab:summary_of_studies_ml1} for details), and $P=\cS_{j_1}^{(k)}$, $Q=\cS_{j_2}^{(k)}$ for all site pairs $(j_1, j_2)$ and hypothesis index $k$. Then, we compute an estimate for the relative conditional shift (denoted by $\hat{t}_{Y|X}^{j_1 \to j_2, (k)}$), and an estimate for the relative covariate shift (denoted by $\hat{t}_{X}^{j_1 \to j_2,(k)}$). The estimation details are in Appendix~\ref{app:subsec_est_bound}. 

\begin{figure}[!h]
  \centering
    \includegraphics[width=\textwidth]{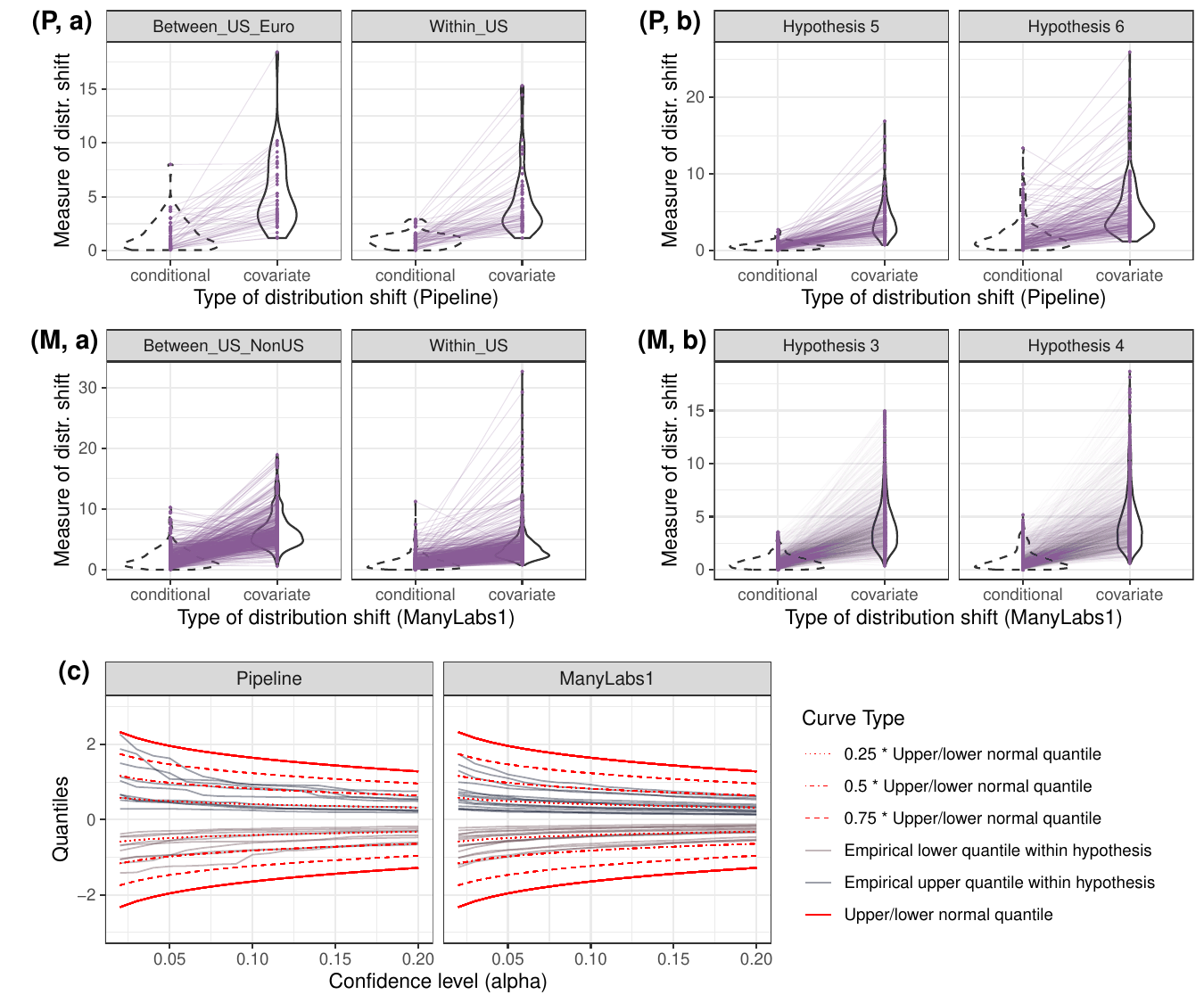} 
  \caption{\textbf{Our covariate shift measures bound conditional shift measures in various contexts (pivotality)}. \textbf{[Left]}:  \textit{Conditional and covariate shift measures for site pairs between US and Europe/Non-US and site pairs within US in the Pipeline data} (P,\,a) \emph{and the ManyLabs 1 data} (M,\,a).
  \textbf{[Right]}: \textit{Conditional and covariate shift measures for all site pairs in hypotheses 5 and 6 in the Pipeline data} (P,\,b), \emph{and those in hypotheses 3 and 4 in the ManyLabs 1 data} (M,\,b). \emph{A few ($\leq 5$) largest values are removed for visualization.}
  \textbf{[Bottom]}: \textit{Empirical quantiles of the ratios between conditional and covariate shift measures within each hypothesis in the Pipeline and ManyLabs1 datasets (grey and brown curves). The red curves are multiples of the quantiles of standard normal distribution plotted for reference.}}
  \label{fig:context_measure_PPML}
\end{figure}

Figure~\ref{fig:context_measure_PPML} compares the conditional shift measure $\hat{t}_{Y|X}^{j_1 \to j_2, (k)}$ and the covariate shift measure $\hat{t}_{X}^{j_1 \to j_2, (k)}$ in various contexts. The left two panels (P,\,a) and (M,\,a) show site pairs $(j_1, j_2)$ where one is in the United States (US) and the other is not in the US, as well as pairs where both sites are in the US. The right two panels (P,\,b) and (M,\,b) show site pairs within two hypotheses for each project. 

In (P,\,a) and (M,\,a), we observe that the distribution shift between US-NonUS pairs tends to be larger than within-US pairs. In (P,\,b) and (M,\,b), the magnitude of distribution shifts also vary across hypotheses. 
Despite the variation across contexts, however, the covariate shift measure upper bounds the conditional shift measure most of the time. 
In addition, when the conditional shift is larger (which is typically unobservable in a generalization task), the observable covariate shift also tends to be larger, justifying the ``predictive role'' of the covariate shift for the conditional shift.

Finally, panel (c) of Figure~\ref{fig:context_measure_PPML} provides a more quantitative illustration of the predictive role. In the figure, each curve is the $\alpha/2$ or $(1-\alpha/2)$-th empirical quantiles of the ratios $\{\hat{t}_{Y|X}^{j_1 \to j_2,(k)} / \hat{t}_X^{j_1\to j_2,(k)}\}_{j_1 \neq j_2}$ for a hypothesis $k$ across a series of confidence levels $\alpha$ on the $x$-axis. For reference, we compare them with multiples of standard normal distribution quantiles. A few comments are in order: 
\vspace{0.5em}
\begin{itemize}
    \item First, the absolute ``bounding'' relationship $| \hat{t}_{Y|X}^{j_1\to j_2,(k)} / \hat{t}_X^{j_1\to j_2,(k)}| \leq 1$ holds with high probability. Thus, in practice, the belief that $|\hat{t}_{Y|X}| \leq \hat{t}_X$ is a plausible option to establish a plausible range of the conditional shift strength. We will see reliable effect generalization based on this idea in Section~\ref{sec:generalize}. 
    \item Second, if one wants to adjust the upper bound of $| \hat{t}_{Y|X}^{j_1\to j_2,(k)} / \hat{t}_X^{j_1\to j_2,(k)}|$ based on a desired confidence level, it is reasonable to use some multiplicative of standard normal, e.g., $0.75$. Indeed, the empirical quantiles are smooth and similar to normal quantiles in general. This suggests a ``smooth'' and ``random'' nature of distribution shift, instead of being adversarial. 
\end{itemize}
\vspace{0.5em}


\subsection{Theoretical Analysis: Random Distribution Shift Model}
\label{sec:interpret}

We here offer a theoretical framework to motivate the predictive role of covariate shift which justifies the empirical evidence in the last section. 


We begin by modeling the data collection procedure as a two-stage sampling process. In the first stage, the underlying distribution is randomly ``perturbed''. With this perturbation, we aim to model unintended changes in the study population or random deviations from the experimental protocols despite efforts to keep them, etc.  In the second stage, data is drawn i.i.d.\ from the perturbed distributions. Thus, we have three sources of uncertainty.
\begin{equation*}
    \underbrace{\hat \theta_Q}_{\text{estimator on target dataset}} - \underbrace{\hat \theta_P}_{\text{estimator on source dataset}} = \underbrace{\hat \theta_Q - \theta(Q)}_{\text{sampling uncertainty}} + \underbrace{\theta(Q) - \theta(P)}_{\text{random shift}} + \underbrace{ \theta(P)- \hat \theta_P}_{\text{sampling uncertainty}}
\end{equation*}
Here, ``sampling uncertainty'' refers to the usual statistical uncertainty arising from randomly drawing observations from an underlying population $P$ or $Q$, and ``random shift'' refers to the discrepancy between two underlying distributions $P$ and $Q$ due to natural ``perturbations'' to them. In the following, we will construct  $Q$ by randomly perturbing $P$. Constructing $P$ by randomly perturbing $Q$ or constructing both $P$ and $Q$ by randomly purturbing a third distribution $P^0$ would lead to the same asymptotics. 

Our model for distribution shift includes three elements:
\vspace{0.5em}
\begin{itemize}
    \item We assume that the treatment distribution is invariant, since the treatment probability is fixed and chosen by the scientists for the datasets we consider here.
    \item There is distribution shift in observed covariates $X$, which we will model as random.  Potentially there is shift in some unobserved effect modifiers $U$, which we will model as random, too. 
    \item The outcome $Y$ is a function of treatment indicator $T$, covariates $X$, and unobserved modifiers $U$. Thus, the shift in $U$ is the driving factor for the conditional shift.
\end{itemize}
\vspace{0.5em}



Let $Y = g(T,X,U)$, where the treatment $T$ is independent of the modifiers $(X,U)$ under $P$ due to randomization. Recall that $X$ is observed, while $U$ is not.

\vspace{-0.5em}
\paragraph{Random distribution shift.}
The key idea of our random distribution shift model is that the original probability measure is randomly brought up and down in small pieces which, put together, leads to CLT-like behavior of the estimates with inflated variance.

To be precise, we let events $\{ C_m^{(M)} \}_{m=1,\ldots,M}$ be a disjoint covering of the sample space of $(X,U)$. We assume that these ``pieces'' have the same probability mass, i.e., $\PP(C_m^{(M)}) = 1/M$ for $m=1,\ldots,M$ and that step functions on these pieces approximate square-integrable functions.\footnote{That is, for any function $\phi(X,U) \in L^2(P)$, it holds that $\EE_P[(\phi(X,U) - \phi_M(X,U))^2] \rightarrow 0$ as $M\to \infty$, where $\phi_M = \sum_{m=1}^M 1_{C_m^{(M)}} \EE_P[\phi(X,U)|C_m^{(M)}]$. This can be achieved relatively easily, e.g.\ for a continuous random variable $X \in \mathbb{R}$ one can choose $C_m$ as intervals whose endpoints correspond to $(m-1)/M$-th quantile and the $m/M$-th quantile of $X$ under ${P}$.} Later, we will take $M\to \infty$ to describe a scenario where many random factors change the probability masses of $C_m^{(M)}$ independently. 

Our model describes random perturbations of $P$ in these small event pieces. Specifically, we define the randomly re-weighted distribution $Q$ for any event $E \subseteq \cup_{m=1}^M C_m^{(M)}$ via 
\begin{equation*}
 Q( E) = \sum_{m=1}^M   P(E \given C_m^{(M)}) \cdot \frac{W_{m}}{\frac{1}{M} \sum_{m'=1}^M W_{m'}},
\end{equation*}
where $(W_{m})_{m=1}^M$ are i.i.d.\ positive random variables that are bounded away from zero and have finite variance. As written above, the treatment indicator $T$ is assumed to be independent of the modifiers $(X,U)$ under both $P$ and $Q$, and its distribution is invariant.

Figure~\ref{fig:shift_idea_flat} visualizes this idea, where probability masses of small events $\{C_m\}$ in the $(X,U)$ sample space are independently perturbed by ``nature''. Such small, random perturbations are suitable to describe unintended but inevitable distribution shifts in such multi-site replication studies, such as unintended changes in the study population or random deviations from the experimental protocols despite efforts to keep them, etc.


\begin{figure}
    \centering
    \includegraphics[width=\linewidth]{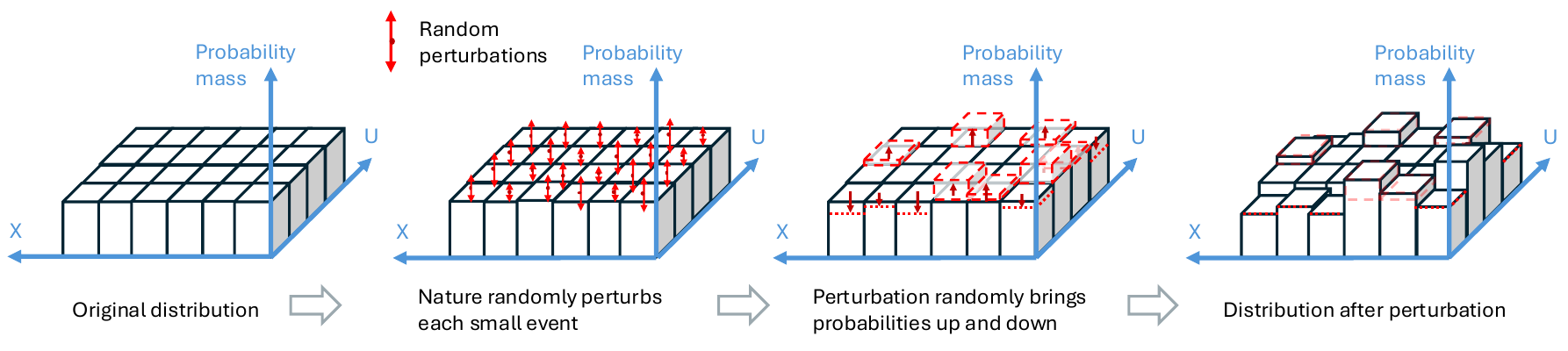}
    \caption{\textbf{Visualization of the random distribution shift model.} The original distribution is randomly perturbed to produce the distribution from which data are i.i.d.~drawn. Our model assumes independent perturbation/reweighting of equal-probability small events and takes the number of small events to infinity.}
    \label{fig:shift_idea_flat}
\end{figure}


Making the grid more fine-grained and taking limits ($n_Q,n_P,M \rightarrow \infty$) we obtain a distributional CLT that describes the shift of empirical means under this two-stage sampling procedure. There are various asymptotic regimes that one could consider. Considering the asymptotic regime where $n_Q/M \rightarrow \rho \in (0,\infty)$ means sampling uncertainty and distributional uncertainty are of the same order~\citep{jeong2022calibrated}. Taking $n_Q /M \rightarrow 0$ means distributional uncertainty is of larger order than sampling uncertainty~\citep{jeong2024out}. In the following, we focus on scenarios where sampling uncertainty and distributional uncertainty are of the same order, that is, we assume that $n_Q/M$ and $n_P/M$ converge to positive real numbers as we let $M \rightarrow \infty$.

\begin{theorem}[Distributional CLT]
\label{thm:clt}
Let $\hat{\mathbb{E}}_Q[\psi]$ denote the sample mean of a function $\psi(T,X,U)$ over $n_Q$ i.i.d.\ draws from $Q$ and $\hat{\mathbb{E}}_P[\psi]$ denote the sample mean of $\psi$ over $n_P$ i.i.d.\ draws from $P$. Under the random distributional shift model described above, for any function $\psi(T,X,U) \in L^2(P)$, we have
\begin{equation*}
  s_n^{-1} \left(  \hat{\mathbb{E}}_Q[\psi] -\hat{\mathbb{E}}_P[\psi] \right) \stackrel{d}{\rightarrow} \mathcal{N}(0,1),
\end{equation*}
where $s_n^2 =  \left( \frac{1}{n_P} + \frac{1}{n_Q} \right) \text{Var}_{P}(\psi)  + \delta_\text{M}^2 \text{Var}_{P}(\EE_P[\psi|X,U]) $, and $\delta_\text{M}^2 = \frac{1}{M} \frac{\EE[W^2]}{\EE[W]^2} $ measures the strength of perturbation. 
If $\psi$ is a vector of functions, then $\text{Var}_{P}(\psi)$ and $\text{Var}_P(\EE_P[\psi|X,U])$ are covariance matrices. 
\end{theorem}

In Theorem~\ref{thm:clt}, the variance term $(\frac{1}{n_P} + \frac{1}{n_Q}) \text{Var}_{P}(\psi)$ is the usual asymptotic variance one would obtain under the i.i.d.~assumption that $P=Q$. In addition,  random perturbations to the distributions contributes a factor of $\delta_\text{M}^2 \text{Var}_{P}(\EE_P[\psi|X,U])$, where only the variance of $\EE_P[\psi|X,U]$ counts because only the distribution of $(X,U)$ is perturbed, while that of $T$ remains invariant. 

\paragraph{Why covariate shift often upper bounds conditional shift.}
In the following, we further discuss how this distributional CLT implies that covariate shift often upper bounds conditional shift. 

For simplicity, we focus on deriving the generalization error for the estimators $\hat\theta_P=\hat\EE_P[\phi]$ and $\hat\theta_Q = \hat\EE_Q[\phi]$. 
A formal justification of this influence function approximation for general $M$-estimators can be found in \cite{jeong2022calibrated}. 
The numerator of our relative conditional shift measure~\eqref{def:rel_cond_shift} equals the difference-in-means estimator with $\psi = \phi - \phi_P(X)$ (ignoring the estimation of $\phi_P(X)$ for simplicity), where $ \phi = \frac{T}{\pi} Y - \frac{1-T}{1-\pi} Y$ or $\phi = Y$ depending on the hypothesis. Applying the distributional CLT, for the squared relative conditional shift measure $\eqref{def:rel_cond_shift}^2$, we get the estimate
\begin{equation}\label{eq:stochsmaller}
  \frac{(\hat{\EE}_Q[\psi] - \hat{\EE}_P[\psi])^2}{\hat{\text{Var}}_P(\psi)} \stackrel{d}{=} \left(\frac{1}{n_P} + \frac{1}{n_Q} + \delta_M^2 \frac{\text{Var}_{P}(\EE_P[\psi|X,U])}{\text{Var}_{P}(\psi)} \right) Z_1 + o_P(\delta_M),
\end{equation}
where $Z_1 \sim \chi^2(1)$.
Using the distributional CLT for the covariates (taking $\psi = X_\ell$ where $X_\ell$ is the $\ell$-th observed covariate), we obtain that standardized squared differences follows a scaled chi-square distribution:
\begin{equation}\label{eq:stochlarger}
   \frac{(\hat{\mathbb{E}}_Q[X_\ell] -\hat{\mathbb{E}}_P[X_\ell])^2}{\hat{\text{Var}}_{P}(X_\ell)} \stackrel{d}{=} 
 \left(\frac{1}{n_P} + \frac{1}{n_Q} + \delta_M^2 \right) Z_1 + o_P(\delta_M).
\end{equation}
Here, $\hat{\text{Var}}_{P}(X_\ell)$ is the sample variance of $X_\ell$ in the source data from $P$. Thus, up to lower order terms, equation~\eqref{eq:stochsmaller} is stochastically smaller than equation~\eqref{eq:stochlarger} because $\text{Var}_{P}(\EE_P[\psi|X,U])/\text{Var}_{P}(\psi) \leq 1$. In other words, the standardized conditional shift is stochastically smaller than the standardized covariate shift. This is in line with the empirical phenomenon observed in Figure~\ref{fig:context_measure_PPML}. This also justifies replacing~\eqref{def:rel_cov_shift} by the stabilized version~\eqref{def:stab_cov_shift}: this is roughly because the perturbations are homogeneous in different directions. 

 
 If we average over multiple covariates $X_\ell$ that are uncorrelated under $P$, by the distributional CLT, we can estimate the squared covariate shift measure $\eqref{def:stab_cov_shift}^2$ by
\begin{equation}\label{eq:stabilized_estimate}     \frac{1}{L}  \sum_{\ell=1}^L   \frac{(\hat{\mathbb{E}}_Q[X_\ell] -\hat{\mathbb{E}}_P[X_\ell])^2}{\hat{\text{Var}}_{P}(X_\ell)} \stackrel{d}{=} 
 \left(\frac{1}{n_P} + \frac{1}{n_Q} + \delta_\text{M}^2 \right) \frac{Z_L}{L} + o_P(\delta_M),
\end{equation}
where $Z_L \sim \chi^2(L)$.
As $\frac{Z_L}{L} \rightarrow 1$ for $L \rightarrow \infty$, equation~\eqref{eq:stabilized_estimate} will be close to $\frac{1}{n_P} + \frac{1}{n_Q} + \delta_\text{M}^2$. 
When the covariates are correlated, one may standardize them with their empirical covariance matrix to restore~\eqref{eq:stabilized_estimate}. In our empirical studies, the covariates exhibit low correlation, hence we directly employ the formula~\eqref{def:stab_cov_shift}. 

These results motivate using a ratio of the estimated conditional shift and estimated covariate shift as a pivot to create prediction intervals. In the next section, we propose such prediction intervals and evaluate the empirical performance.

%% file: sec4_generalization.tex
\section{Effect Generalization by Exploiting the Predictive Role}
\label{sec:generalize}

In this section, we demonstrate that leveraging the predictive role of covariate shift leads to reliable generalization for target distributions. To this end, we build prediction intervals\footnote{We again create prediction intervals for easier evaluation based on target estimates (instead of the underlying parameters).} for the target population estimate $\widehat{\theta}_Q$ based on our distribution shift measure and evaluate their empirical coverage. 
 
\subsection{Constructing Prediction Intervals} 
Before presenting the results, we begin with a high-level overview of our construction of the prediction intervals using the relationship between conditional and covariate shift measures, while we defer technical details on the estimation procedures to Appendix~\ref{app:subsec_est_generalize}. 

We consider generalization tasks where a scientist has access to full observations from the source distribution $P$ but only the covariates $X$ from the target distribution $Q$. To construct our prediction interval for the target \emph{estimate} $\hat\theta_Q$, we leverage the ratio between the covariate and conditional shift measures:
\$
\hat{r} := \hat{t}_{Y|X}  / \, \hat{t}_{X},
\$
where $\hat{t}_{Y|X}$ is the estimated conditional shift measure~\eqref{def:rel_cond_shift}, and $\hat{t}_{X}$ is the corresponding covariate shift measure~\eqref{def:stab_cov_shift}. 
Note that one can estimate $\hat{t}_X$ but not $\hat{t}_{Y|X}$ in a generalization task. 
Suppose the distribution of $\hat{r}$ can be characterized (e.g., using calibration approaches we will discuss below) so that one can find upper and lower bounds ${L}$ and ${U}$ obeying approximately
\@\label{eq:LU_high_level}
\PP\Big( L\leq \hat{r} \leq U   \Big)  \geq 1-\alpha.
\@
By definition, inverting the above event leads to a general form of our prediction interval for $\hat\theta_Q$: 
\@\label{eq:PI_high_level}
\hat{C}  =\Big[  \hat\theta_w + L\cdot \hat{t}_{X}  \cdot \hat{s}_{Y|X}  ,~  \hat\theta_{w} + U\cdot \hat{t}_{X} \cdot \hat{s}_{Y|X}  \Big],
\@
where $\hat{s}_{Y|X}$ is an estimate for $\textrm{sd}_P(\phi-\phi_P)$ in~\eqref{def:rel_cond_shift}, and $\hat{\theta}_w$ is an estimator for $\EE_Q[\phi_P(X)]$ in~\eqref{def:rel_cond_shift} which adjusts for the covariate shift. 
Above, all quantities in~\eqref{eq:PI_high_level} except $L$ and $U$   can be estimated with full observations from the source distribution $P$ and the covariate data from the target distribution $Q$. 

\begin{figure} 
    \centering
    \includegraphics[width=0.6\linewidth]{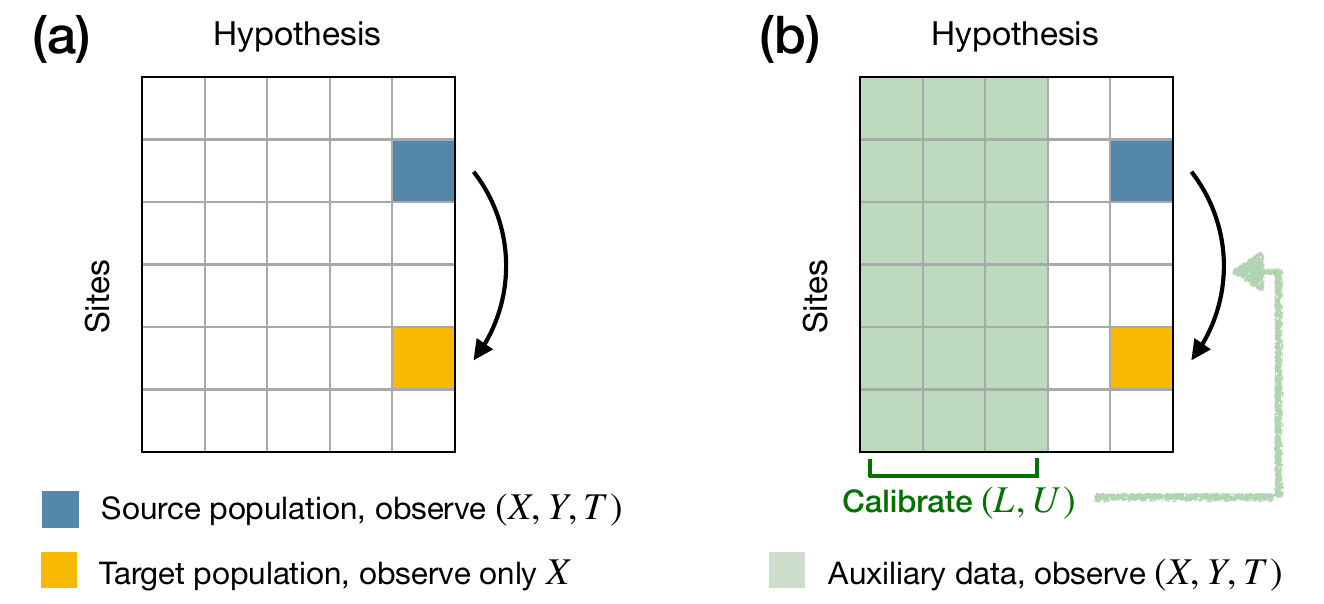}
    \caption{Generalization in two scenarios for the availability of data. \textbf{Left}: \emph{Generalization without auxiliary data from source (blue) to target (yellow)}. \textbf{Right}: \emph{Generalization with auxiliary data (green) from the same sites for other hypotheses. The auxiliary data are used to calibrate $L$ and $U$ for a new generalization task}.}
    \label{fig:generalize_idea}
\end{figure}

We consider two ways to calibrate $(L,U)$ under two scenarios of data availability (visualized in Figure~\ref{fig:generalize_idea}):
\begin{enumerate}
    \item \textbf{Constant calibration.} We construct prediction intervals assuming that the conditional shift measure is bounded by the covariate shift measure (i.e., using constant bounds $L= -1$ and $U=1$). This is theoretically justified under the random distribution shift model (Section~\ref{sec:interpret}).     
    This approach is applicable to a generalization task with no information other than covariate data from the target site. 
    \item \textbf{Data-adaptive calibration.} We construct prediction intervals by calibrating the relative strengths of conditional and covariate shift measures using some separate, existing data. This is applicable when some relevant auxiliary data are available (but not full observations in the target site) and we believe they inform the (relative) strengths of distribution shifts in the current generalization task. 
\end{enumerate}
\vspace{0.25em}

Of course, the set of available data in the second approach can be more general; we explore other scenarios in Appendix~\ref{sec:appendix-calibration}.
These proposed prediction intervals are compared with three baselines:
\vspace{0.5em}
\begin{enumerate}
    \item \texttt{IID.} Prediction intervals under the i.i.d.~assumption, i.e., $P=Q$, ignoring distribution shift.
\item \texttt{WorstCase.} Prediction intervals based on upper and lower worst-case bounds under restrictions on the distributional distance between the target distribution and the reweighted distribution, i.e., $\textrm{KL}(Q_X\otimes P_{Y|X} \| Q)\leq \rho$, where $\rho$ is calibrated with data (see Appendix~\ref{app:subsec_est_generalize} for details in both scenarios). 
\item \texttt{Oracle.} Prediction intervals calibrated with true knowledge of the relative strength of covariate shift and conditional shift measures. This is the ``ideal'' but unrealistic version of our method.
\end{enumerate}
\vspace{0.5em}

We evaluate the generalization performance of different methods by the empirical coverage and average length of prediction intervals across all site pairs for each hypothesis.

\subsection{Empirical Evaluation}

\subsubsection{Without Any Auxiliary Data}
In the first scenario, the scientists have data from the source distribution but they do not have any information other than covariates  $X$ from the target distribution. In this setting, researchers can use our proposed approach with constant calibration. 

More specifically, we consider the generalization of site $j_1$ to $j_2$ for all pairs $j_1, j_2\in \{1,\dots, N\}$, $j_1\neq j_2$, for each hypothesis $k \in \{1, \ldots, K\}$ in each application. When we construct prediction intervals, we assume all data from $j_1$ are observed while only covariates $X$ are observed from $j_2$. When we then \textit{evaluate} the statistical performance of various generalization methods, we use the full data in site $j_2$ to empirically evaluate how well each estimator approximates the benchmark estimate in $j_2$. 

\begin{figure}[!h]
    \centering
    \includegraphics[width=\linewidth]{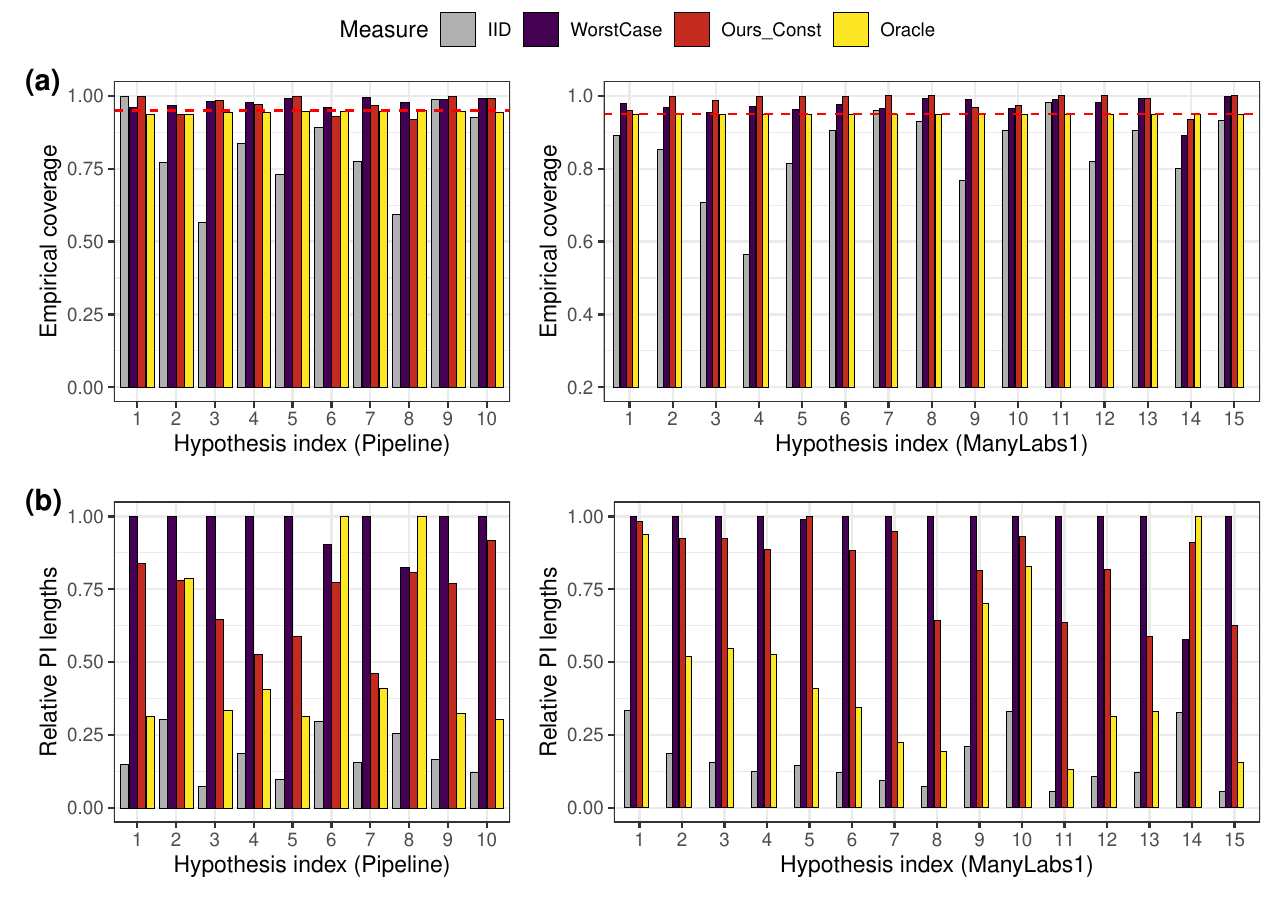}
    \caption{\textbf{Effect Generalization Without Auxiliary Data.} \textbf{Row (a)}: \textit{Empirical coverage of prediction intervals via constant calibration at nominal level $1-\alpha=0.95$ and three baseline methods using the Pipeline data (left) and ManyLabs 1 data (right).}
    \textbf{Row (b)}: \textit{Average length of prediction intervals for constant calibrated prediction intervals at nominal level $1-\alpha=0.95$ for the four methods, normalized by the largest average length in each study, using the Pipeline data (left) and ManyLabs 1 data (right).}}
    \label{fig:const_eb_KL_PPML}
\end{figure}

In Figure~\ref{fig:const_eb_KL_PPML}, we report the empirical coverage and relative lengths of prediction intervals averaged over all pairs within each hypothesis. Across two distinct applications, our procedure (denoted as ``Ours\_Const'' in red) achieves the target of $95\%$ coverage in most cases (see panel (a)). \texttt{WorstCase} prediction intervals achieve the target coverage as well but are much wider than the proposed intervals (see panel (b)). Not surprisingly, intervals based on the i.i.d.\ assumption exhibit undercoverage.

\subsubsection{With Auxiliary Data}
Next, we examine how we can improve the performance of our estimator when researchers have some auxiliary data to use data-adaptive calibration for our method. We specifically consider a scenario where data from all sites exist for other hypotheses to build prediction intervals for a new hypothesis. In practice, this setting arises when there is existing data from the same set of sites on other research questions or hypotheses. 

We calibrate $L$ and $U$ by the quantiles of site pairs from existing hypotheses, and then build a prediction interval for a new hypothesis that is only observed in one single site. To ensure stable evaluation, the ordering of the sites is randomly permuted for $10$ times. Additional calibration scenarios (generalizing to new sites for existing hypotheses, and new sites for new hypotheses)  in Appendix~\ref{sec:appendix-calibration} deliver similar messages.

\begin{figure}[!h]
    \centering
    \includegraphics[width=\linewidth]{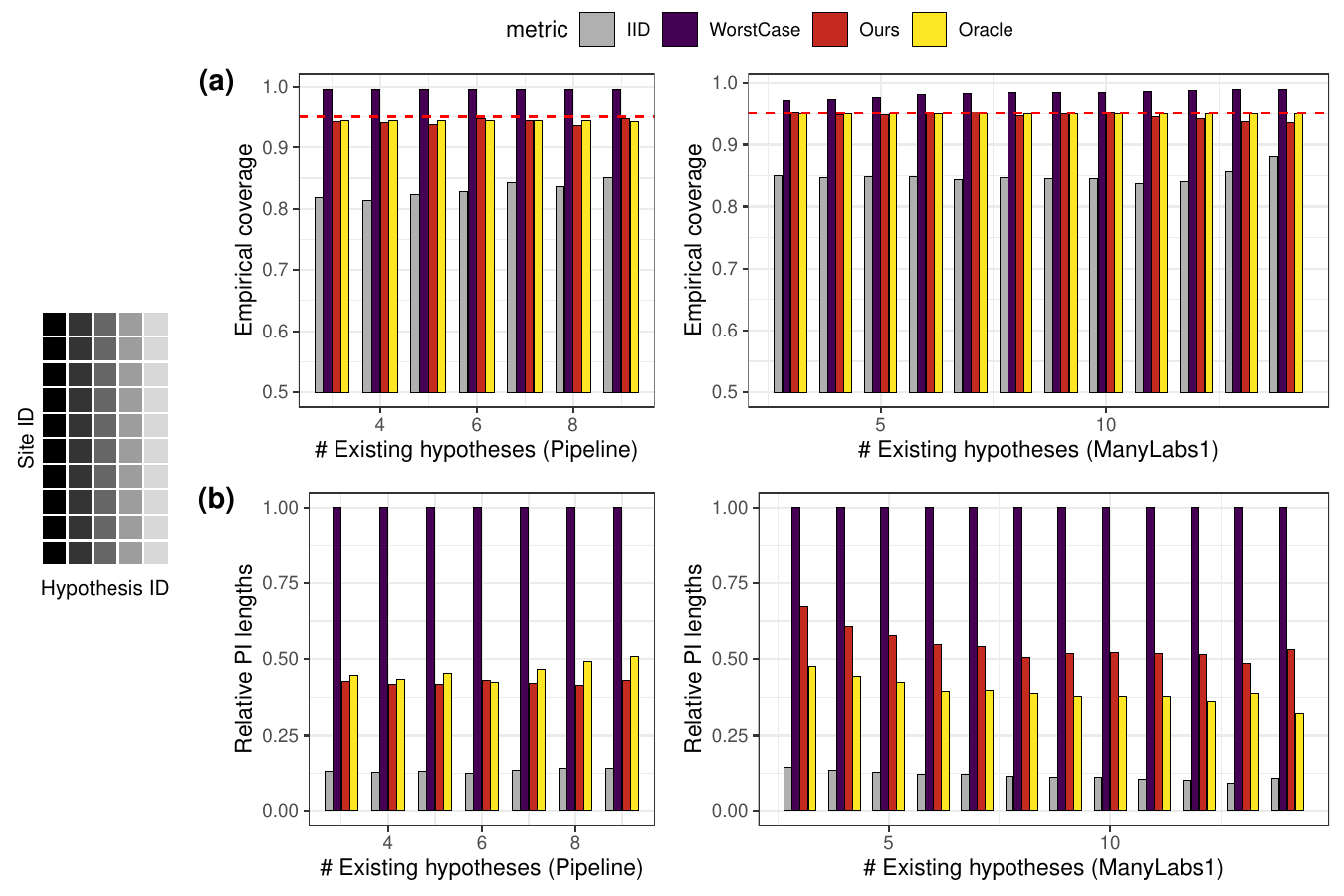}
    \caption{\textbf{Effect Generalization With Auxiliary Data:} We generalize to new studies based on distribution shift measures calibrated from the same sites in other hypotheses. \textbf{Left}: \textit{Illustration of data collection order, where dark color means earlier.} 
    \textbf{Row (a)}: \textit{Average coverage of prediction intervals built with four methods over 10 random draws of study ordering, using the Pipeline data (left) and ManyLabs 1 data (right). The red dashed line is the nominal level $0.95$}. 
    \textbf{Row (b)}: \textit{Average length of prediction intervals over 10 random draws of study ordering, normalized by the largest average length in each study, using the Pipeline data (left) and ManyLabs 1 data (right)}.}
    \label{fig:over_study_eb_KL_PPML}
\end{figure}

In Figure~\ref{fig:over_study_eb_KL_PPML}, we report the coverage and lengths of prediction intervals. For both projects, our procedure achieves coverage close to the nominal level, with prediction intervals that are much smaller than the intervals based on worst-case bounds, and quite close to the oracle method. As before, prediction intervals based on the i.i.d.\ assumption exhibit undercoverage.

%% file: appendix.tex

\section{Details of  datasets and data pre-processing}
\label{app:sec_data}

\subsection{Pre-processing for Pipeline project}

The raw datasets for the Pipeline project can be found in the OSF repository \url{https://osf.io/q25xa/}. The detailed data pre-processing script can be found in the folder  \texttt{Pipeline} in the GitHub respository \url{https://github.com/ying531/awesome-replicability-data}. 

We follow the data processing scripts (in the folder ``SPSS Syntax files'') provided in the OSF repository to compute the response variables, encode the treatment indicators, and extract the covariates including age, gender, country of birth, language, ethnicity, parent education, and family incomes. 
When running the analysis, we additionally process the data for each site as follows: covariates with all N/A values are excluded; otherwise, the missing observations are imputed by the site median. Since entropy balancing enforces positive weights, when running the EB-based methods, we also exclude covariates whose sample average in the target dataset falls outside the support in the source dataset. 

\subsection{Pre-processing for ManyLabs1 project}

The raw datasets for the ManyLabs1 project an be found in the OSF repository \url{https://osf.io/wx7ck/}. The detailed data processing script can be found in the folder \texttt{ManyLabs1} in the GitHub respository \url{https://github.com/ying531/awesome-replicability-data}. 

We  follow the data processing scripts \texttt{Syntax.Manylabs.sps} in the OSF repository to encode the responses (\texttt{dv}) and treatment indicators (\texttt{iv}), and extract the covariates including gender, age, race, ethnicity, nationality, native language, religion, and ideology.

\subsection{Reproduction code}

The code for reproducing the analysis is available at~\url{https://github.com/ying531/predictive-shift}. For easier reproduction, we also include analyses results (such as computed distribution shift measures and constructed KL-based bounds which can be costly to run) ready for producing the figures in the main text.

\subsection{Dataset information}

Table~\ref{tab:study_institutions} lists the data indices and data collection sites for the Pipeline project from the Open Science Framework (OSF) repository. Table~\ref{tab:summary_of_studies} summarizes the information for each of the $10$ hypothesis studied in the Pipeline project, including the name, test statistic and formula, number of sites conducting experiments for testing this hypothesis, and total sample sizes $N$ recruited in these sites.

Table~\ref{tab:study_institutions_ml1} lists the data collection sites in the Manylabs1 project. 
Table~\ref{tab:summary_of_studies_ml1} summarizes the information for each of the $15$ hypotheses studied in the ManyLabs1 dataset, including the hypothesis, estimator, formula (for processed data), number of sites conducting experiments for the hypothesis, and total sample sizes $N$.

\begin{table}[H]
\centering
\begin{tabular}{|c|c|p{10.5cm}|}
\hline
\textbf{New Index} & \textbf{Raw ID} & \textbf{PI, Institution} \\ \hline
1   & 0   & Original Study data collection \\ \hline
2   & 1   & Aaron Sackett, University of St. Thomas \\ \hline
3   & 2   & Alexandra Mislin, American University \\ \hline
4   & 4   & David Tannenbaum, University of Chicago \\ \hline
5   & 5   & Daniel Storage, University of Illinois at Urbana-Champaign \\ \hline
6   & 6   & Adam Hahn, University of Cologne \\ \hline
7   & 7   & Nicole Legate, Illinois Institute of Technology \\ \hline
8   & 8   & INSEAD Sorbonne Lab \\ \hline
9   & 9   & Victoria Brescoll, Yale University \\ \hline
10  & 10  & Felix Cheung, Michigan State University/University of Hong Kong \\ \hline
11  & 11  & Fiery Cushman, Harvard University \\ \hline
12  & 12  & Jay Van Bavel, New York University \\ \hline
13  & 13  & Tatiana Sokolova, HEC Paris and University of Michigan \\ \hline
14  & 15  & Jesse Graham, University of Southern California \\ \hline
15  & 16  & Anne-Laure Sellier, HEC Paris \\ \hline
16  & 17  & Eli Awtrey, University of Washington \\ \hline
17  & 18  & Jennifer Jordan, University of Groningen \\ \hline
18  & 19  & Sapna Cheryan, University of Washington \\ \hline
19  & 20  & Xiaomin Sun, Beijing Normal University \\ \hline
20  & 21  & Yoel Inbar, University of Toronto \\ \hline
21  & 22  & Wendy Bedwell, University of South Florida \\ \hline
22  & 24  & Deanna Kennedy, University of Washington Bothell \\ \hline
23  & 25  & Matt Motyl, University of Illinois at Chicago \\ \hline
24  & 26  & Erik Cheries, University of Massachusetts Amherst \\ \hline
25  & 27  & Additional INSEAD-Sorbonne lab data for Study 1\\ \hline
26  & 141 & Dan Molden, Packet 1 for Study 7 \\ \hline
27  & 142 & Dan Molden, Packet 2 for Study 4 and Study 8 \\ \hline
28  & 311 & UCI Psychology Students \\ \hline
29  & 312 & UCI Business Students \\ \hline
\end{tabular}
\caption{List of new index, raw site ID in the dataset, and contributing PI and site institutions in the Pipeline project dataset, taken from the Open Science Framework project repository~\citep{Madan_Uhlmann_Schweinsberg_Tierney_2016}.}
\label{tab:study_institutions}
\end{table}

\begin{table}[H]
\centering
\begin{tabular}{|c|l|c|c|c|l|}
\hline
\textbf{ID} & \textbf{Hypothesis} & \textbf{Estimator} & \textbf{Formula} & \textbf{\#Sites} & \textbf{$N$}\\ \hline
1  & Bigot–misanthrope   & $t$-test & \texttt{bigot\_personjudge $\sim$ condition} & 12 & 2861 \\ \hline
2  & Cold-hearted prosociality & Paired $t$-test  & \texttt{tdiff $\sim$ 1} & 12   & 2806\\ \hline
3  & Bad tipper & $t$-test & \texttt{tipper\_personjudg $\sim$ condition}  & 16  & 3658 \\ \hline
4  & Belief–act inconsistency   & $t$-test & \texttt{beliefact\_mrlblmw\_rec $\sim$ condition13}  & 13  & 3006 \\ \hline
5  & Moral inversion  & $t$-test & \texttt{moralgood $\sim$ condition} & 14  & 3076 \\ \hline
6  & Moral cliff & Paired $t$-test & \texttt{diff $\sim$ 1} & 15  & 3300 \\ \hline
7  & Intuitive economics  & $t$-test & \texttt{yz $\sim$ condition} & 15  & 3164 \\ \hline
8  & Burn-in-hell & Paired $t$-test & \texttt{tdiff $\sim$ 1} & 15  & 3176 \\ \hline
9  & Presumption of guilt  & $t$-test & \texttt{companyevaluation $\sim$ condition} & 17  & 3806 \\ \hline
10 & Higher standard & $t$-test & \texttt{standard\_evalu\_7items $\sim$ condition} & 11  & 2692\\ \hline 
\end{tabular}
\caption{Estimator, number of sites  and total sample size $N$ for each study  in the Pipeline project.}
\label{tab:summary_of_studies}
\end{table}

\begin{table}[!h] 
\renewcommand{\arraystretch}{1}
\centering
\begin{tabular}{|c|c|p{12cm}|}
\hline
\textbf{New Index} & \textbf{Raw Site ID} & \textbf{Institution, Location} \\ \hline
1  & Abington & Penn State Abington, Abington, PA \\ \hline
2  & Brasilia & University of Brasilia, Brasilia, Brazil \\ \hline
3  & Charles & Charles University, Prague, Czech Republic \\ \hline
4  & Conncoll & Connecticut College, New London, CT \\ \hline
5  & CSUN & California State University, Northridge, LA, CA \\ \hline
6  & Help & HELP University, Malaysia \\ \hline
7  & Ithaca & Ithaca College, Ithaca, NY \\ \hline
8  & JMU & James Madison University, Harrisonburg, VA \\ \hline
9  & KU & Koç University, Istanbul, Turkey \\ \hline
10 & Laurier & Wilfrid Laurier University, Waterloo, Ontario, Canada \\ \hline
11 & LSE & London School of Economics and Political Science, London, UK \\ \hline
12 & Luc & Loyola University Chicago, Chicago, IL \\ \hline
13 & McDaniel & McDaniel College, Westminster, MD \\ \hline
14 & MSVU & Mount Saint Vincent University, Halifax, Nova Scotia, Canada \\ \hline
15 & MTURK & Amazon Mechanical Turk (US workers only) \\ \hline
16 & OSU & Ohio State University, Columbus, OH \\ \hline
17 & Oxy & Occidental College, LA, CA \\ \hline
18 & PI & Project Implicit Volunteers (US citizens/residents only) \\ \hline
19 & PSU & Penn State University, University Park, PA \\ \hline
20 & QCCUNY & Queens College, City University of New York, NY \\ \hline
21 & QCCUNY2 & Queens College, City University of New York, NY \\ \hline
22 & SDSU & SDSU, San Diego, CA \\ \hline
23 & SWPS & University of Social Sciences and Humanities Campus Sopot, Sopot, Poland \\ \hline
24 & SWPSON & Volunteers visiting www.badania.net \\ \hline
25 & TAMU & Texas A\&M University, College Station, TX \\ \hline
26 & TAMUC & Texas A\&M University-Commerce, Commerce, TX \\ \hline
27 & TAMUON & Texas A\&M University, College Station, TX (Online participants) \\ \hline
28 & Tilburg & Tilburg University, Tilburg, Netherlands \\ \hline
29 & UFL & University of Florida, Gainesville, FL \\ \hline
30 & UNIPD & University of Padua, Padua, Italy \\ \hline
31 & UVA & University of Virginia, Charlottesville, VA \\ \hline
32 & VCU & VCU, Richmond, VA \\ \hline
33 & Wisc & University of Wisconsin-Madison, Madison, WI \\ \hline
34 & WKU & Western Kentucky University, Bowling Green, KY \\ \hline
35 & WL & Washington \& Lee University, Lexington, VA \\ \hline
36 & WPI & Worcester Polytechnic Institute, Worcester, MA \\ \hline
\end{tabular} 
\caption{List of new index, raw site ID in the dataset, and contributing PI and site institutions in the ManyLabs 1 dataset, taken from the original paper~\citep{klein2014investigating}.}
\label{tab:study_institutions_ml1}
\end{table}

\begin{table}[!h]
\centering
\begin{tabular}{|c|c|c|c|c|c|}
\hline
\textbf{ID} & \textbf{Hypothesis} & \textbf{Estimator} & \textbf{Formula} & \textbf{\#Sites} & \textbf{$N$} \\ \hline
1  & Allowedforbidden & $t$-test & \texttt{dv $\sim$ iv} & 36 & 6292 \\ \hline
2  & Anchoring1 & $t$-test & \texttt{dv $\sim$ iv} & 36 & 5362 \\ \hline
3  & Anchoring2 & $t$-test & \texttt{dv $\sim$ iv} & 36 & 5284 \\ \hline
4  & Anchoring3 & $t$-test & \texttt{dv $\sim$ iv} & 36 & 5627 \\ \hline
5  & Anchoring4 & $t$-test & \texttt{dv $\sim$ iv} & 36 & 5609 \\ \hline
6  & Contact & $t$-test & \texttt{dv $\sim$ iv} & 36 & 6336 \\ \hline
7  & Flag & $t$-test & \texttt{dv $\sim$ iv} & 36 & 6251 \\ \hline
8  & Gainloss & $t$-test & \texttt{dv $\sim$ iv} & 36 & 6271 \\ \hline
9  & Gambfal & $t$-test & \texttt{dv $\sim$ iv} & 36 & 5942 \\ \hline
10 & Iat & $t$-test & \texttt{dv $\sim$ iv} & 36 & 5851 \\ \hline
11 & Money & $t$-test & \texttt{dv $\sim$ iv} & 36 & 6333 \\ \hline
12 & Quote & $t$-test & \texttt{dv $\sim$ iv} & 36 & 6325 \\ \hline
13 & Reciprocity & $t$-test & \texttt{dv $\sim$ iv} & 36 & 6276 \\ \hline
14 & Scales & $t$-test & \texttt{dv $\sim$ iv} & 36 & 5899 \\ \hline
15 & Sunk & $t$-test & \texttt{dv $\sim$ iv} & 36 & 6330 \\ \hline
\end{tabular}
\caption{Estimator, number of sites  and total sample size $N$ for each study (indices and variable names in cleaned data and this paper) for ManyLabs 1 data.}
\label{tab:summary_of_studies_ml1}
\end{table}

\section{Estimation details}
\label{app:sec_estimation}

In this section, we detail the estimation procedures for all the analyses in this paper. Appendix~\ref{app:subsec_est_notation} recalls important notations. Appendix~\ref{app:subsec_est_explain} describes the analysis for the explanatory role in Section~\ref{subsec:explain}. Appendix~\ref{app:subsec_est_bound} details the estimation for our distribution shift measures in Section~\ref{sec:bound}. 
Finally, Appendix~\ref{app:subsec_est_generalize} details our estimation and evaluation procedures for effect generalization in Section~\ref{sec:generalize}.

\subsection{Notations}
\label{app:subsec_est_notation}

We begin by revisiting some notations. A hypothesis $k$ is replicated by sites $j\in \{1,\dots,N_k\}$, each observing a dataset $\cD_j^{(k)} = \{X_i^{(j,k)},T_i^{(j,k)},Y_i^{(j,k)}\}_{i=1}^{n_j^{(k)}}$, where $X_i$ is the covariates, $T_i\in \{0,1\}$ is the binary treatment, and $Y_i$ is the outcome(s). For each hypothesis $k$, the estimate for site $j$ is $\hat\theta_j^{(k)} = \theta^{(k)}(\cD_j^{(k)})$, where $\theta^{(k)}$ is the same functional that represents the analysis procedure applied to all sites (as listed in Tables~\ref{tab:summary_of_studies} and~\ref{tab:summary_of_studies_ml1}). 
Here, $\hat\theta_j^{(k)}$ estimates the population parameter $\theta_j^{(k)} = \theta^{(k)}(P_{j}^{(k)})$, where $P_j^{(k)}$ is the underlying distribution from which $\cD_j^{(k)}$ is drawn. We assume access to a function $\phi^{(k)}(\cdot)$ such that 
\$
\hat\theta_j^{(k)} = \frac{1}{n_j^{(k)}} \sum_{i=1}^{n_j^{(k)}} \phi^{(k)}\big(X_i^{(j,k)},Y_i^{(j,k)},T_i^{(j,k)}\big).
\$

\subsection{Estimation for the explanatory role}
\label{app:subsec_est_explain}

In this part, we detail how the prediction intervals for IID, CovShift (DR) and CovShift (EB) are constructed and evaluated in Section~\ref{subsec:explain}. 
The sites are denoted as $i,j\in \{1,\dots,N\}$, where site $i$ is the ``original'' site with full observations, and site $j$ is the ``target'' site we want to generalize the effects to.  

\paragraph{Estimation for IID.}
For any site pair $(i,j)$ for a hypothesis $k$, we assume access to a consistent variance estimator $(\hat\sigma_i^{(k)})^2$ for $\hat\theta_{i}^{(k)}$, such that 
\$
\sqrt{n_{i}^{(k)}} \cdot \frac{\hat\theta_i^{(k)} - \theta_i^{(k)}}{\hat\sigma_i^{(k)}} \stackrel{d}{\to}\cN(0,1).
\$
Note that $\hat\theta_i$ and $\hat\sigma_i$ can be computed using  full observations $\cD_i^{(k)}$ from the ``original'' site $i$. It is straightforward to construct these estimators for the $t$-tests and paired $t$-tests considered in this work, and we note that $\hat\sigma_i^{(k)} = \hat\sigma_j^{(k)}+o_P(1)$ for any $i\neq j$ if the i.i.d.~assumption holds.
For the \texttt{IID} method, we construct a prediction interval based on site $i$ for $\hat\theta_j^{(k)}$ via 
\@\label{eq:PI_iid}
\hat{C}_{i\to j}^{\text{IID},(k)} = \hat\theta_i^{(k)} \pm  q_{1-\alpha/2}\cdot \hat\sigma_i^{(k)} \cdot \Big( \sqrt{1/n_{i}^{(k)} + 1/ n_{j}^{(k)} } \Big),
\@
where $q_{1-\alpha/2}$ is the $(1-\alpha/2)$-th quantile of a standard normal distribution.
Under the i.i.d.~assumption that $P_{i}^{(k)}=P_j^{(k)}$, we know that 
\$
\PP \Big( \hat\theta_j^{(k)} \in \hat{C}_{i\to j}^{\text{IID},(k)}  \Big) \to 1-\alpha.
\$
For evaluation, we will use full observations from the target site. Each grey bar in (P,a) and (M,a) of Figure~\ref{fig:explanatory} is computed via 
\$
\hat{\text{Cov}}_k^{\text{IID}} := \frac{1}{N_k(N_k-1)} \sum_{i=1}^{N_k}\sum_{j\neq i}  \ind \Big\{ \hat\theta_j^{(k)} \in \hat{C}_{i\to j}^{\text{IID},(k)}  \Big\}.
\$
 Thus, if the i.i.d.~assumption holds, we will expect $\hat{\text{Cov}}_k^{\text{IID}} \approx 1-\alpha$.

\paragraph{Estimation for CovShift (DR).} 
For any site pair $(i,j)$ in a hypothesis $k$, we first describe how to construct a point estimate for generalization via reweighting. We denote the estimator as  $\hat\theta_{i\to j}^{(k)}$ when generalizing from site $i$ with full observations to site $j$ with only covariate information. We will employ cross-fitting~\citep{chernozhukov2018double} to allow the use of flexible machine learning algorithms such as random forests in estimating the covariate shift weights and conditional mean functions. 

First, we randomly split the data $\cD_i^{(k)}$ and covariates in $\cD_j^{(k)}$ into two equally-sized halves each. We use one half of data to estimate the covariate shift function $\ud P_{j,X}^{(k)}/\ud P_{i,X}^{(k)}(x)$ via $\hat{w}(x)$, and the conditional mean function $\varphi(x) := \EE[\phi^{(k)}(X,Y,T)\given X=x]$ via $\hat\varphi(x)$. These functions will be applied to the other fold of data, and construct the reweighted estimator 
\@\label{eq:est_dr}
\hat\theta_{i\to j}^{(k)} = \frac{1}{n_{i}^{(k)}}\sum_{\ell} \hat{w}(X_\ell^{(i,k)}) \cdot\Big\{ \phi^{(k)}(X_\ell^{(i,k)},Y_\ell^{(i,k)},T_\ell^{(i,k)}) - \hat\varphi(X_\ell^{(i,k)} ) \Big\} +  \frac{1}{n_{j}^{(k)}}\sum_{\ell=1}^{n_j^{(k)}}\hat\varphi(X_\ell^{(j,k)} ).
\@
Following~\cite{jin2024tailored}, if the covariate shift condition holds, one can show that for the $t$-test and paired $t$-test considered in this work, as long as $\hat{w}$ and $\hat\varphi$ converge to the true covariate shift weight function and the true conditional mean function with a rate of $o_P((n_{i}^{(k)})^{-1/4})$, it holds that 
\$
\frac{ \hat\theta_j^{(k)} - \hat\theta_{i\to j}^{(k)} }{\hat\sigma_{i\to j}^{(k),\textrm{CovShift}}} \stackrel{d}{\to} \cN(0,1),
\$
where $\hat\sigma_{i\to j}^{(k),\textrm{CovShift}}$ is any consistent estimator for $\sigma_{i\to j}^{(k),\textrm{CovShift}}$, and 
\$
(\sigma_{i\to j}^{(k),\textrm{CovShift}})^2 = \frac{\EE_i^{(k)}[w(X)^2\cdot(\phi^{(k)}(X,Y,T) - \varphi^{(k)}(X))^2]}{n_{i}^{(k)}} + \frac{\EE_i^{(k)}[w(X)\cdot(\phi^{(k)}(X,Y,T) - \varphi^{(k)}(X))^2]}{n_{j}^{(k)}}.
\$
As such, we construct the prediction interval for CovShift (DR) via 
\$
\hat{C}_{i\to j}^{\text{CovShift},(k)} =\hat\theta_{i\to j}^{(k)} \pm  q_{1-\alpha/2}\cdot \hat\sigma_{i\to j}^{(k),\textrm{CovShift}} ,
\$
where $\hat\sigma_{i\to j}^{(k),\textrm{CovShift}} $ is constructed by plugging in $\hat{w}$ and $\hat\varphi$ into the definition of $\sigma_{i\to j}^{(k),\textrm{CovShift}}$.
Based on the arguments above, assuming covariate shift, under standard assumptions above, we would have 
\$
\PP\Big( \hat\theta_j^{(k)} \in \hat{C}_{i\to j}^{\text{CovShift},(k)}  \Big) \to 1-\alpha.
\$
For evaluation, we will use full observations from the target site. Each green bar in (P,a) and (M,a) of Figure~\ref{fig:explanatory} is computed via 
\$
\hat{\text{Cov}}_k^{\text{CovShift}} := \frac{1}{N_k(N_k-1)}\sum_{i=1}^{N_k}\sum_{j\neq i}  \ind\Big\{\hat\theta_j^{(k)} \in \hat{C}_{i\to j}^{\text{CovShift},(k)}  \Big\}.
\$
If the covariate shift assumption holds, we expect $\hat{\text{Cov}}_k^{\text{CovShift}} \approx 1-\alpha$ under standard regularity conditions.

\paragraph{Estimation for CovShift (EB).} 
The idea for constructing the point estimate for CovShift (EB) is similar to CovShift (DR), with the only exception that we obtain the weights $\hat{w}_\ell^{(i,k)}$ are obtained by entropy balancing~\citep{hainmueller2012entropy} following the procedure in~\cite{jin2023diagnosing}, while $\hat\sigma_{i\to j}^{(k),\textrm{CovShift}}$ is obtained in the same way as in CovShift (DR). We then construct the point estimate 
\@\label{eq:est_eb}
\hat\theta_{i\to j}^{(k)} = \frac{1}{n_{i}^{(k)}}\sum_{\ell} \hat{w}_\ell^{(i,k)} \cdot \phi^{(k)}(X_\ell^{(i,k)},Y_\ell^{(i,k)},T_\ell^{(i,k)})
\@
and prediction interval 
\$
\hat{C}_{i\to j}^{\text{CovShift},(k)} =\hat\theta_{i\to j}^{(k)} \pm  q_{1-\alpha/2}\cdot \hat\sigma_{i\to j}^{(k),\textrm{CovShift}} .
\$

Following~\cite{jin2023diagnosing},  assuming covariate shift, if the weight is a logistic function of the covariates or if $\varphi(x)$ is a linear function of the covariates, we would have 
\$
\PP\Big( \hat\theta_j^{(k)} \in \hat{C}_{i\to j}^{\text{CovShift},(k)}  \Big) \to 1-\alpha.
\$
For evaluation, we will use full observations from the target site. 
For evaluation, each purple bar in (P,a) and (M,a) of Figure~\ref{fig:explanatory} is computed via 
\$
\hat{\text{Cov}}_k^{\text{CovShift}} := \frac{1}{N_k(N_k-1)}\sum_{i=1}^{N_k}\sum_{j\neq i}  \ind\Big\{\hat\theta_j^{(k)} \in \hat{C}_{i\to j}^{\text{CovShift},(k)}  \Big\}
\$
using the prediction intervals for CovShift (EB). 
Thus, if the covariate shift assumption holds, we will expect $\hat{\text{Cov}}_k^{\text{CovShift}} \approx 1-\alpha$ under the stated linear assumptions which are standard in the balancing literature. We note that CovShift (EB) is more stable than CovShift (DR) for small-to-moderate sample sizes, as is the case for the datasets analyzed in this work.

\subsection{Estimation for distribution shift measures}
\label{app:subsec_est_bound}

We then proceed to detail the estimation procedure for our new distribution shift measures. To begin with, we note the following decomposition by~\cite{jin2023diagnosing}, which measures the contributions of distribution shifts (on the super-population level) to effect discrepancy:
\@\label{eq:mgn_decomp}
\theta(Q) - \theta(P) = \underbrace{\theta(Q) - \theta(Q_X\times P_{Y|X})}_{\textrm{Contribution of conditional shift}} + \underbrace{\theta(Q_X\times P_{Y|X}) - \theta(P)}_{\textrm{Contribution of covariate shift}}
\@
where $\theta(\cdot)$ is the functional for the parameter of interest, $P$ is the source distribution,  $Q$ is the target distribution, and $Q_X\times P_{Y|X}$ is the reweighted distribution. Note that the contribution of conditional shift will be zero under the covariate shift assumption (Definition~\ref{def:cov_shift}). 
In multi-site replication studies, for generalizing estimates for a hypothesis $k$ from site $i$ to site $j$, we will take $\theta=\theta^{(k)}$, $P = P_i^{(k)}$, and $Q = P_j^{(k)}$.

\paragraph{Computing the conditional shift measure.}
Following~\eqref{eq:mgn_decomp}, we recall our definitions of the population-level conditional shift measure (for generalizing from $P$ to $Q$) in Section~\ref{subsec:def_shift_measure}, denoted as 
\$
t_{Y|X} := \frac{\Delta_{Y|X}}{s_{Y|X}},\quad 
\Delta_{Y|X} = \theta(Q) - \theta(Q_X\times P_{Y|X}),\quad 
s_{Y|X}^2 = \Var_P \big( \phi(X,Y,T) - \EE_P[\phi(X,Y,T)\given X]\big),
\$
where the contributions of the conditional shift is rescaled by the standard deviation of its influence function to ensure scale invariance. 

Following the notations in the preceding subsection, 
we compute the conditional shift measure from site $i$ to site $j$ in hypothesis $k$ via the following formula:
\@\label{eq:est_cond_shift}
\hat{t}_{Y|X}^{i\to j,(k)} = \frac{\hat\Delta_{Y|X}^{i\to j,(k)}}{\hat{s}_{Y|X}^{i\to j,(k)}} := \frac{\hat\theta_j^{(k)} - \hat\theta_{i\to j}^{(k)} }{\hat{s}_{Y|X}^{i\to j,(k)}}
\@
where $\hat\theta_j^{(k)}$ is the target estimator for $\theta(Q)$, $\hat\theta_{i\to j}^{(k)}$ is the doubly robust estimator~\eqref{eq:est_dr} or the entropy balancing estimator~\eqref{eq:est_eb} in the previous part, so that $\hat\Delta_{Y|X}$ is an estimator for the contribution of conditional shift. In addition, $\hat{s}_{Y|X}^{i\to j,(k)}$ is a consistent estimator for $\Var_P \big( \phi(X,Y,T) - \EE_P[\phi(X,Y,T)\given X]\big)^{1/2}$, which we detail in Appendix~\ref{app:subsec_est_variance} and introduce its fast convergence properties.

\paragraph{Computing the covariate shift measure.}
Finally, we compute the ``stabilizes'' covariate shift measure as mentioned in the main text. Namely, supposing there are $L$ covariates $\{X_\ell\}_{\ell=1}^L$, we compute 
\@\label{eq:est_stab_X}
\hat{t}_{X}^{i \to j, (k)} := \sqrt{\frac{1}{L} \sum_{\ell=1}^L \Big( \frac{\hat\EE_Q[X_\ell] - \hat\EE_P[X_\ell] }{\hat\sigma_P(X_{\ell})}\Big)^2},
\@
where $\hat\sigma_P(X_\ell)$ is the empirical standard deviation of $X_\ell$ in the source dataset. 
Note that  $\hat{t}_{X}^{i \to j, (k)}$ is pivotal as $n_{i}^{(k)}, n_j^{(k)}\to \infty$ under the i.i.d.~assumption. 

\paragraph{Computing the ratios.}
After computing the two measures $\hat{t}_{Y|X}^{i\to j,(k)} $ and $\hat{t}_{X}^{i\to j,(k)} $, we simply measure their relative strengths by the ratio 
\$
\hat{r}_{i\to j}^{(k)} = \hat{t}_{Y|X}^{i\to j,(k)} / \hat{t}_{X}^{i\to j,(k)} .
\$
Alternative definitions of distribution shift measures will be explored in Appendix~\ref{app:alt_shift_measures}, yet we find they  either (i) are scale-dependent (hence interpretation is sensitive the definition of the parameter functional $\theta(\cdot)$), or (ii) lead to unstable performance in estimation and effect generalization. 

\paragraph{Idea for effect generalization based on distribution shift measures.}

Finally, we recall the high-level idea of effect generalization based on our distribution shift measures. 
If the distribution of the ratio $\hat{r}^{i\to j,(k)}$ (which depends on both sampling uncertainty and distribution shifts) can be characterized, so that one can find upper and lower bounds ${L}$ and ${U}$ (either by asymptotic distribution or data-adaptive calibration) such that (approximately) 
\@\label{eq:idea_LU}
\PP\Big( L\leq \hat{r}_{i\to j}^{(k)} \leq U   \Big)  \geq 1-\alpha,
\@
then, inverting this fact would give a prediction interval for $\hat\theta_j^{(k)}$, which is 
\@\label{eq:idea_PI}
\hat{C}_{i\to j}^{(k)} =\Big[  \hat\theta_{i\to j}^{(k)} + L\cdot \hat{t}_{X}^{i \to j, (k)} \cdot \hat{s}_{X}^{i \to j, (k)} ,~  \hat\theta_{i\to j}^{(k)} + U\cdot \hat{t}_{X}^{i \to j, (k)} \cdot \hat{s}_{X}^{i \to j, (k)} \Big].
\@
Above, except for $L$ and $U$, all quantities can be estimated with full observations from site $i$ and covariates from site $j$. Next, we will detail how $L$ and $U$ are calibrated in Section~\ref{sec:generalize}. 

\subsection{Estimation for effect generalization}
\label{app:subsec_est_generalize}

In this part, we detail our estimation and evaluation procedures for effect generalization in Section~\ref{sec:generalize}. We first introduce the IID method and the Oracle method evaluated in both Figure~\ref{fig:const_eb_KL_PPML} and Figure~\ref{fig:over_study_eb_KL_PPML}. 
Then, we introduce WorstCase and Ours methods for constant calibration and adaptive calibration in the two figures, respectively. 

\paragraph{IID method.} With the i.i.d.~assumption, we construct prediction intervals as~\eqref{eq:PI_iid} for generalizing from site $i$ to site $j$ for hypothesis $k$. That is, we use no covariate information in the sites, and the empirical coverage of the IID method is mainly plotted for reference. For coverage and lengths, we average over all site pairs for a given hypothesis in all scenarios. 

\paragraph{Oracle method.} This method uses all site pairs to calibrate the range of $\hat{r}^{i\to j,(k)}$, namely, we compute 
\$
L^{\text{Orc},(k)} := \text{Quantile}\Big( \alpha/2;\, \big\{ \hat{r}_{i\to j}^{(k)} \big\}_{i\neq j} \Big), \quad 
U^{\text{Orc},(k)} := \text{Quantile}\Big( 1-\alpha/2;\, \big\{ \hat{r}_{i\to j}^{(k)} \big\}_{i\neq j} \Big)
\$
for the bounds $L$ and $U$ in~\eqref{eq:idea_LU}. 
As its name suggests, it is the ideal prediction interval when we have perfect knowledge of how distribution shifts between all sites for a hypothesis. 
Note that this approach uses much more information than available in a real generalization task, and is hence evaluated just for reference. For coverage and lengths, we average over all site pairs for a given hypothesis in all scenarios.

\subsubsection{Constant calibration}
\label{app:subsubsec_const_cal}

\paragraph{Our method (constant calibration).} 
We take constants $L=-1$ and $U=1$ in~\eqref{eq:idea_LU}, i.e., we believe that the conditional shift is upper bounded by the covariate shift. This leads to the prediction interval 
\$
\hat{C}_{i\to j}^{\textrm{Ours}, (k)} =\Big[  \hat\theta_{i\to j}^{(k)} -   \hat{t}_{X}^{i \to j, (k)} \cdot \hat{s}_{X}^{i \to j, (k)} ,~  \hat\theta_{i\to j}^{(k)} + \hat{t}_{X}^{i \to j, (k)} \cdot \hat{s}_{X}^{i \to j, (k)} \Big],
\$
which is computable in a real generalization task with $\cD_i^{(k)}$ and covariates in $\cD_j^{(k)}$. The barplots in Figure~\ref{fig:over_study_eb_KL_PPML} show the empirical coverage 
\$
\frac{1}{N_k(N_k-1)} \sum_{i\neq j}  \ind\Big\{  \hat\theta_j^{(k)}\in \hat{C}_{i\to j}^{\textrm{Ours}, (k)}  \Big\}
\$
and  average lengths 
\$
\frac{1}{N_k(N_k-1)} \sum_{i\neq j} \Big| \hat{C}_{i\to j}^{\textrm{Ours}, (k)}  \Big|
\$
after normalization by the largest average length for each hypothesis $k$.

\paragraph{Worst-case method.} 
We also evaluate the performance of worst-case bounds on the conditional shift, calibrated with data at hand. 
These worst-case bounds estimate the range of target parameters under the constraint that the unknown conditional shift is bounded in a KL-divergence ball.

Before we introduce our approach, we first remark two aspects about this approach:
\vspace{0.5em}
\begin{enumerate}
    \item Rigorously speaking, this is not a feasible generalization approach since we need full observations from all sites (especially the outcomes from the target site) to calibrate the KL bound $\hat{\textrm{KL}}_{\textrm{upp}}^{(k)}$, which is typically not available in a real generalization task. As such, we mainly use it for reference. 
    \item There are several approximations in this approach, since the estimation uncertainty in $\hat{\textrm{KL}}_{\textrm{upp}}^{(k)}$ is not accounted for, and it usually needs to account for larger uncertainty to cover the actual estimator than the underlying parameter. Thus, the intervals we obtain here can be viewed as underestimating the actual uncertainty, and a rigorous approach would construct even wider intervals. 
\end{enumerate}
\vspace{0.5em}

Specifically, let $P$ be the source distribution and $Q$ be the target distribution. The strength of conditional shift can be characterized by  the KL divergence between the reweighted distribution $Q_X\times P_{Y|X}$ and the target distribution $Q$, i.e., 
\$
\text{KL} \big(Q \,\|\, Q_X\times P_{Y|X} \big) 
&= \EE_{Q_X\times P_{Y|X}} \bigg[ \frac{\ud Q}{\ud (Q_X\times P_{Y|X})}(X,Y)  \cdot  \log \frac{\ud Q}{\ud (Q_X\times P_{Y|X})}(X,Y) \bigg] \\ 
&= \EE_{Q_X\times P_{Y|X}} \bigg[ \frac{\ud Q_{Y|X}}{\ud   P_{Y|X} }(X,Y)  \cdot  \log \frac{\ud Q_{Y|X}}{\ud   P_{Y|X} }(X,Y)  \bigg].
\$
To estimate this quantity, we first use a classification model to estimate the joint density ratio $\ud Q_{X,Y}/\ud P_{X,Y}(x,y)$ via $\hat{w}_{X,Y}(\cdot)$, and then the covariate density ratio $\ud Q_{X }/\ud P_{X }(x )$ via $\hat{w}_X(\cdot)$. Then, we estimate the conditional density ratio $\ud Q_{Y|X}/\ud   P_{Y|X}(x,y) $ via $\hat{w}_{X,Y}(x,y) / \hat{w}_X(x)$, and plug in the definition to obtain an estimator for the KL-divergence, denoted as 
$
\hat{\textrm{KL}}_{i\to j}^{(k)}
$
when taking $P = P_i^{(k)}$ and $Q = P_j^{(k)}$. 

After obtaining $\hat{\textrm{KL}}_{i\to j}^{(k)}$ for all pairs of studies, we calibrate an upper bound for the conditional KL-divergence for any given hypothesis $k$ via 
\$
\hat{\textrm{KL}}_{\textrm{upp}}^{(k)} := \text{Quantile}\Big( 0.99; \, \{\hat{\textrm{KL}}_{i\to j}^{(k)}\}_{i\neq j}   \Big),
\$
where we take the $0.99$ quantile to avoid outliers. Then, we compute upper and lower bounds for the parameters $\theta(P_j^{(k)})$ 
by solving the following optimization program:
\$
\textrm{Maximize/minimize} \quad & \theta(\bar{Q}) \\ 
\textrm{Subject to} \quad  & \text{KL} \big(\bar{Q} \,\|\, Q_X\times P_{Y|X} \big) \leq \hat{\textrm{KL}}_{\textrm{upp}}^{(k)}.
\$
Algorithms for solving the above program with data are standard in the literature; see, e.g.,~\cite{hu2013kullback}.
We then use the maximized and minimized objective as upper and lower bounds for the target estimator, giving rise to the prediction interval 
\$
\hat{C}_{i\to j}^{\textrm{KL}, (k)} 
:= \Big[  \hat{U}_{i\to j}^{\textrm{KL}, (k)} ,\, \hat{U}_{i\to j}^{\textrm{KL}, (k)} \Big].
\$
The barplots in Figure~\ref{fig:over_study_eb_KL_PPML} show the empirical coverage 
\$
\frac{1}{N_k(N_k-1)} \sum_{i\neq j}  \ind\Big\{  \hat\theta_j^{(k)}\in \hat{C}_{i\to j}^{\textrm{KL}, (k)}  \Big\}
\$
and  average lengths 
\$
\frac{1}{N_k(N_k-1)} \sum_{i\neq j} \Big| \hat{C}_{i\to j}^{\textrm{KL}, (k)}  \Big|
\$
after normalization for each hypothesis $k$.

\subsubsection{Data-adaptive calibration}
\label{app:subsubsec_adapt_cal}

Data-adaptive calibration uses separate datasets, which we assume to be available at a generalization task, to calibrate the strength of distribution shift. We will follow the notations in the preceding part. 

\paragraph{Ours (data-adaptive calibration).} 
In Figure~\ref{fig:over_study_eb_KL_PPML}, we assume that data for hypothesis $k_1,\dots,k_t$ are available when we want to generalize between sites for a new hypothesis $k_{t+1}$. Thus, we calibrate the lower and upper bounds in~\eqref{eq:idea_LU} at step $t$ by 
\@\label{eq:adapt_LU}
L^{\text{Ours},(t)} := \text{Quantile}\Big( \alpha/2;\, \big\{ \hat{r}_{i\to j}^{(k_s)} \big\}_{i\neq j, s\leq t} \Big), \quad 
U^{\text{Ours},(t)} := \text{Quantile}\Big( 1-\alpha/2;\, \big\{ \hat{r}_{i\to j}^{(k_s)} \big\}_{i\neq j, s\leq t} \Big).
\@
The idea is that if the distribution of $\hat{r}_{i\to j}^{(k)}$ is ``pivotal'' across hypothesis, using data for other hypotheses (other outcomes) to calibrate new hypotheses will lead to reliable coverage. 
We then construct prediction intervals for sites in new hypotheses $k_{s}$, $s>t$ by 
\$
\hat{C}_{i\to j}^{(k_s)} =\Big[  \hat\theta_{i\to j}^{(k_s)} + L^{\text{Ours},(t)} \cdot \hat{t}_{X}^{i \to j, (k_s)} \cdot \hat{s}_{X}^{i \to j, (k_s)} ,~  \hat\theta_{i\to j}^{(k_s)} + U^{\text{Ours},(t)}\cdot \hat{t}_{X}^{i \to j, (k_s)} \cdot \hat{s}_{X}^{i \to j, (k_s)} \Big].
\$
The coverage and lengths of them are similarly evaluated. We also randomly order the hypotheses $(k_1,\dots,k_t)$ and evaluate for ten times.

\paragraph{Worst-case method.} 
Similar to the previous worst-case method, we will calibrate an upper bound of conditional shift and compute prediction intervals. 
Here, the upper bound will be calibrated with the observed data. Specifically, given all sites for hypotheses $\{k_1,\dots,k_t\}$, we compute individual KL divergences following Section~\ref{app:subsubsec_const_cal}, and then compute 
\@\label{eq:bound_worst_case}
\hat{\textrm{KL}}_{\textrm{upp}}^{(t)} := \text{Quantile}\Big( 0.99; \, \{\hat{\textrm{KL}}_{i\to j}^{(k_s)}\}_{i\neq j, s\leq t}   \Big).
\@
For each future site pair $(i,j)$ for hypothesis $k_s$, $s>t$, we solve an empirical version of 
\$
\textrm{Maximize/minimize} \quad & \theta(\bar{Q}) \\ 
\textrm{Subject to} \quad  & \text{KL} \big(\bar{Q} \,\|\, P_{j,X}^{(k_s)}\times P_{i,Y|X}^{(k_s)} \big) \leq \hat{\textrm{KL}}_{\textrm{upp}}^{(t)},
\$
and use the obtained maximized/minimized objectives as the upper/lower bounds. 
Note that this time, all quantities are computable in a real generalization task.

\section{Additional empirical results}
\label{app:sec_plot}

\subsection{Other calibration scenarios}\label{sec:appendix-calibration}

In this part, we present additional calibration scenarios omitted in Section~\ref{sec:generalize} in the main text, where we use certain observed data to calibrate the relative strength of covariate and conditional shifts, and construct prediction intervals in future generalization tasks. We omit the detailed procedures  as they follow exactly the same ideas as Appendix~\ref{app:subsubsec_adapt_cal}, except for the construction of the bounds $L$ and $U$ in~\eqref{eq:idea_LU}.

\paragraph{Calibration with other sites.}
The second scenario is to calibrate the measures with existing sites involving all hypotheses for new sites. 
We randomly order the sites with $(j_1,\dots,j_{N})$ as a permutation of $(1,\dots,N)$. Then, at each step $t\in \{1,\dots,N-1\}$, we assume data from sites $\{j_1,\dots,j_t\}$ for all the hypotheses are observed, and use the empirical quantiles of $\{\hat{r}_{i\to j}^{(k)}\}_{i,j\in \{j_1,\dots,j_t\},k=1,\dots,K}$ as $L$ and $U$ in the construction of prediction intervals~\eqref{eq:PI_high_level}. Finally, for each pair of sites $j_1, j_2\in \{j_{t+1},\dots,j_{29}\}$, we consider the task of generalization from fully observed data in site $j_1$ for hypothesis $k$ to partially observed site $j_2$ for all hypotheses $j\in \{1,\dots,10\}$, using the aforementioned quantiles to construct prediction intervals following~\eqref{eq:PI_high_level}. 
On the other hand, the KL-divergence bound for \texttt{WorstCase} is also calibrated with these existing pairs in a way that is similar to~\eqref{eq:bound_worst_case}. 

The empirical coverage and PI lengths calibrated with other sites are reported in Figure~\ref{fig:over_site_eb_KL_PPML}. Again, the \texttt{WorstCase} method exhibits overcoverage and very wide intervals, while our method achieves valid coverage while being close to the \texttt{Oracle} method. 

\begin{figure}[H]
    \centering
    \includegraphics[width=\linewidth]{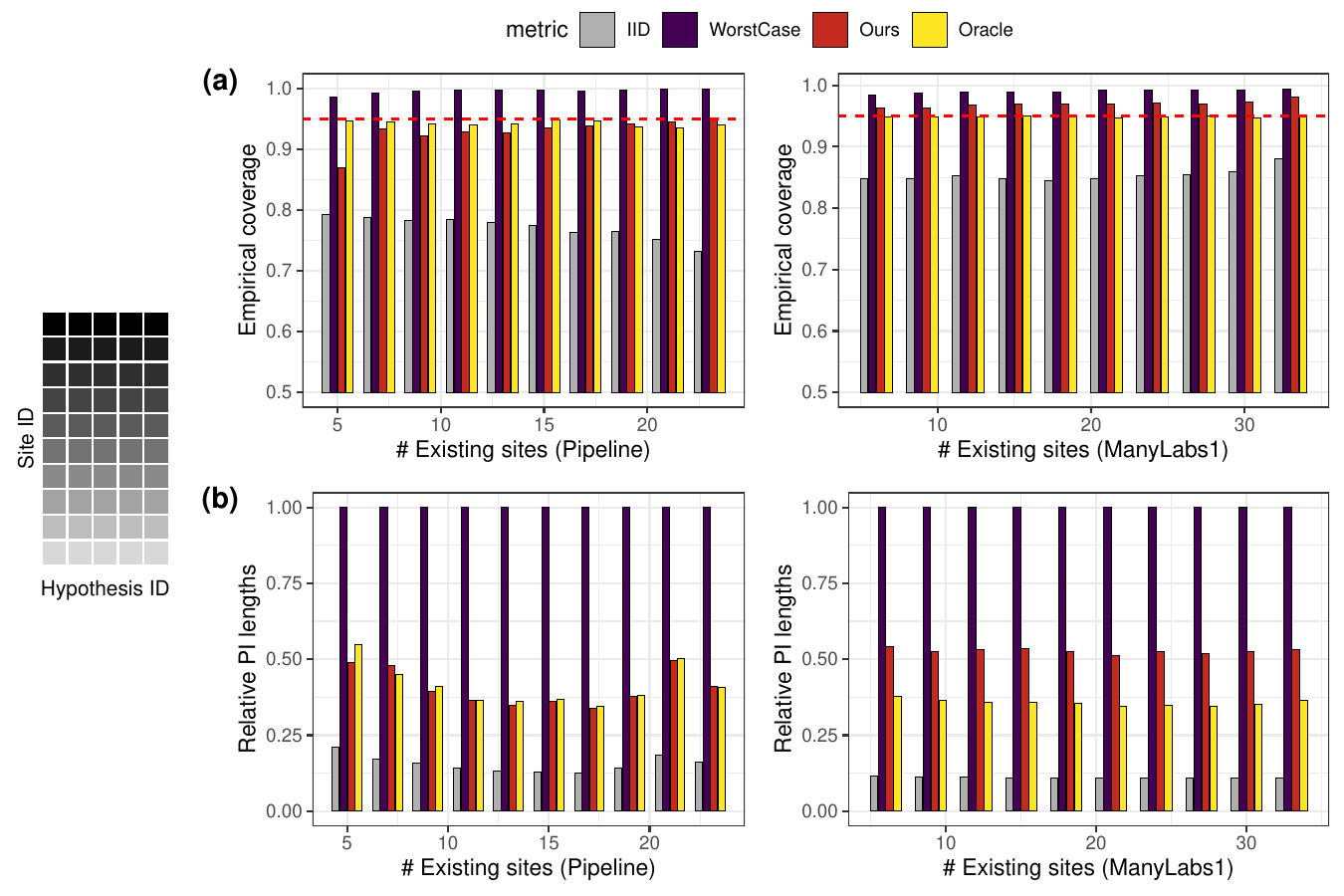}
    \caption{\textit{Generalization in new studies based on distribution shift measures from other sites.} \textbf{Left}: \textit{Illustration of data collection order, where dark color means earlier.} 
    \textbf{Row (a)}: \textit{Average coverage of prediction intervals using the Pipeline data (left) and ManyLabs 1 data (right).}  
    \textbf{Row (b)}: \textit{Average length of prediction intervals using the Pipeline data (left) and ManyLabs 1 data (right)}. \textit{Details are otherwise as Figure~\ref{fig:over_study_eb_KL_PPML}.}}
    \label{fig:over_site_eb_KL_PPML}
\end{figure}

\paragraph{Calibration with other sites and other hypotheses.} 
The final scenario is the most challenging, where for a new generalization task, only data from other sites for other hypotheses are available. 
Specifically, we randomly order the sites by $(j_1,\dots, j_{N})$ and hypotheses by $(k_1,\dots, k_{10})$. 
Then, at each step $t\in \{1,\dots,9\}$, data for studies $\{k_1,\dots, k_K\}$ are available in sites $\{j_1,\dots, j_{3t}\}$, we use the empirical quantiles of $\{\hat{r}_{i\to j}^{(k)}\}_{i,j\in \{j_1,\dots, j_{3t}\},k\in \{k_1,\dots, k_K\}}$ as $L$ and $U$ in the construction of prediction intervals~\eqref{eq:PI_high_level}.
Finally, for each pair of sites $j_1, j_2 \in \{j_{3t+1},\dots,j_{N}\}$, we consider generalization from site $j_1$ to site $j_2$ for each hypothesis $k \in \{k_{t+1},\dots, k_{K}\}$, using the aforementioned quantiles to construct prediction intervals following~\eqref{eq:PI_high_level}. 
The KL-divergence bound for \texttt{WorstCase} is also calibrated with these existing pairs similar to~\eqref{eq:bound_worst_case}. 

The empirical coverage and length of PIs are reported in Figure~\ref{fig:over_both_eb_KL}. Similar to the observations in other scenarios, \texttt{WorstCase} is much more conservative, while our method achieves valid coverage with prediction interval lengths close to \texttt{Oracle}. This scenario is the most challenging among all, since the sites and hypotheses are entirely disjoint between existing data and new generalization tasks. 

\begin{figure}[h]
\centering
\includegraphics[width=\linewidth]{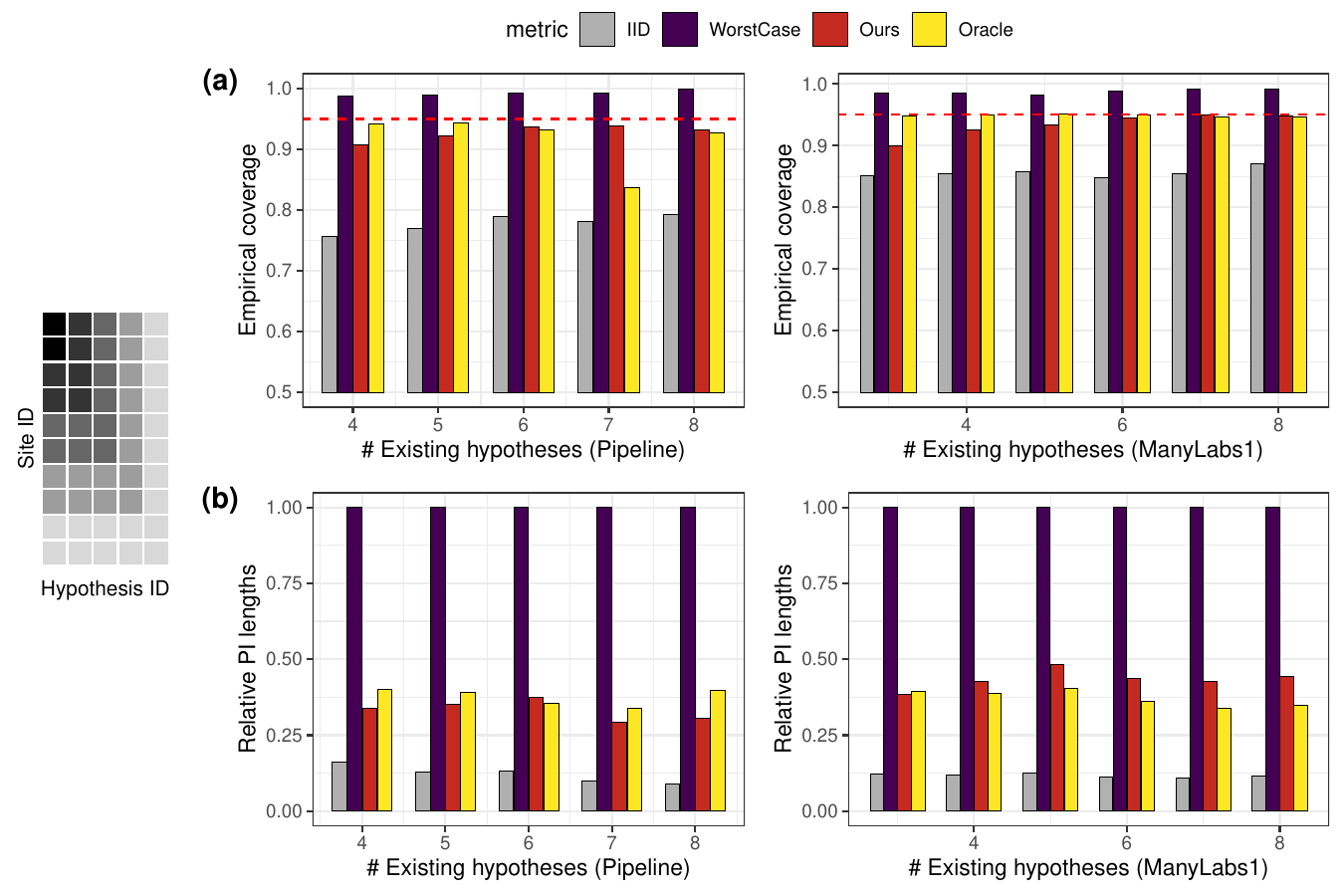}
\caption{\textit{Generalization in new studies based on distribution shift measures from other sites and other hypotheses for new sites and new hypotheses.} \textbf{Left}: \textit{Illustration of data collection order, where dark color means earlier.} 
    \textbf{Row (a)}: \textit{Average coverage of prediction intervals using the Pipeline data (left) and ManyLabs 1 data (right).}  
    \textbf{Row (b)}: \textit{Average length of prediction intervals using the Pipeline data (left) and ManyLabs 1 data (right)}. \textit{Details are otherwise as Figure~\ref{fig:over_study_eb_KL_PPML}.}}
\label{fig:over_both_eb_KL}
\end{figure}

\subsection{Relative strengths of distribution shift measures}

In this part, we report additional results for the relative strengths of distribution shift measures in both projects, which complement Figure~\ref{fig:context_measure_PPML} in the main text. In particular, Figures~\ref{fig:all_PP_eb_measure} and~\ref{fig:all_PP_dr_measure} plots distribution shift measures in each hypothesis of the Pipeline project, computed with entropy balancing and the doubly robust estimator, respectively. Figures~\ref{fig:all_ML_eb_measure} and~\ref{fig:all_ML_dr_measure} plot those for the ManyLabs1 project. 

Consistent with Figure~\ref{fig:context_measure_PPML}, we see that the covariate shift upper bounds the conditional shift most of the time, but the balancing method tends to produce more stable estimates with small-to-moderate sample sizes.

\begin{figure}[!h]
    \centering
    \includegraphics[width=\linewidth]{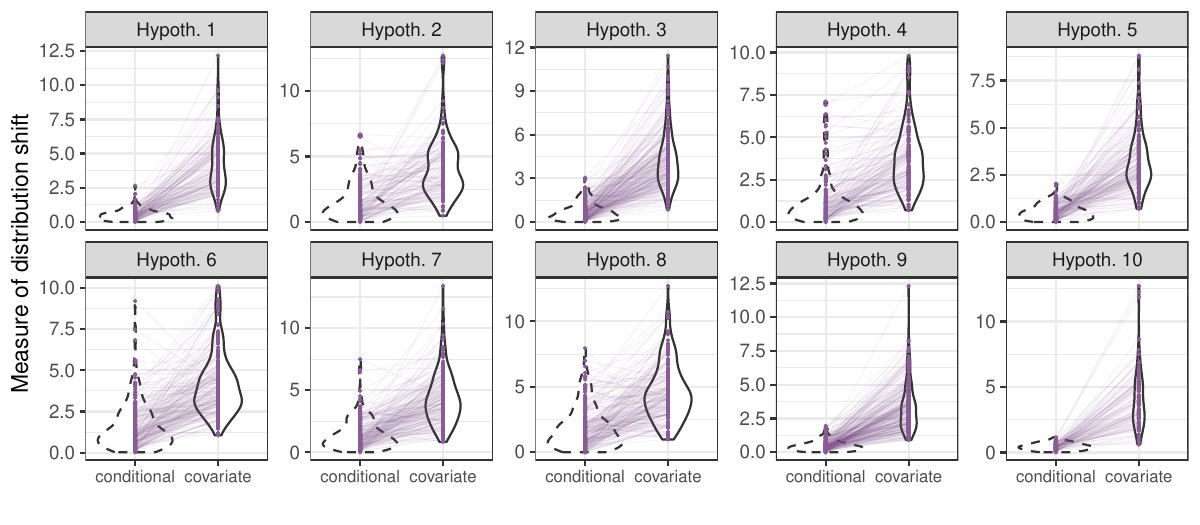}
    \caption{Distribution shift measures between all site pairs in each hypothesis in the Pipeline project, where the covariate shift adjustment uses entropy balancing.}
    \label{fig:all_PP_eb_measure}
\end{figure}

\begin{figure}[H]
    \centering
    \includegraphics[width=\linewidth]{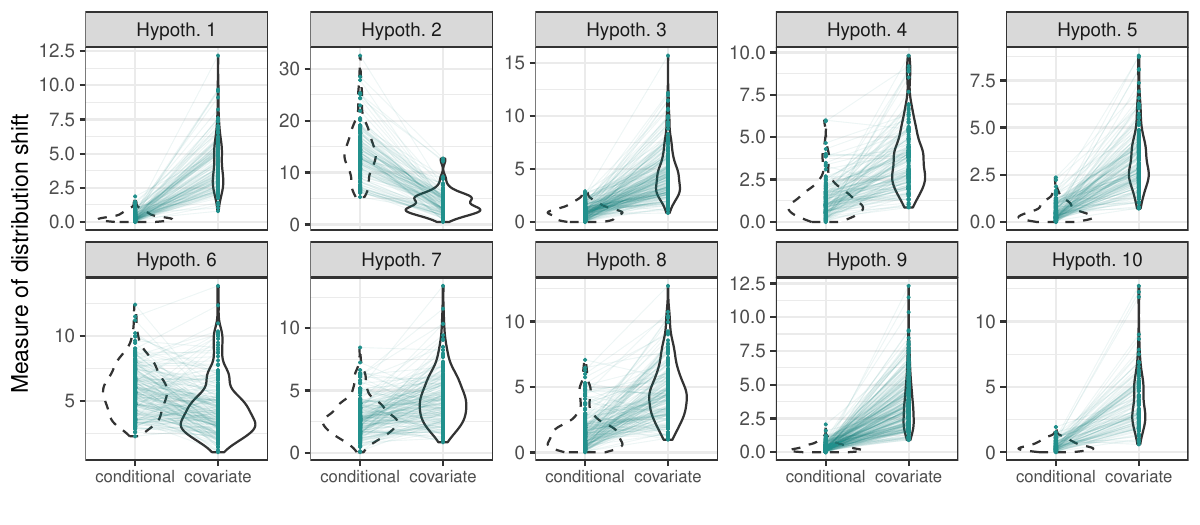}
    \caption{Distribution shift measures between all site pairs in each hypothesis in the Pipeline project, where the covariate shift adjustment uses the doubly robust estimator.}
    \label{fig:all_PP_dr_measure}
\end{figure}

\begin{figure}[H]
    \centering
    \includegraphics[width=\linewidth]{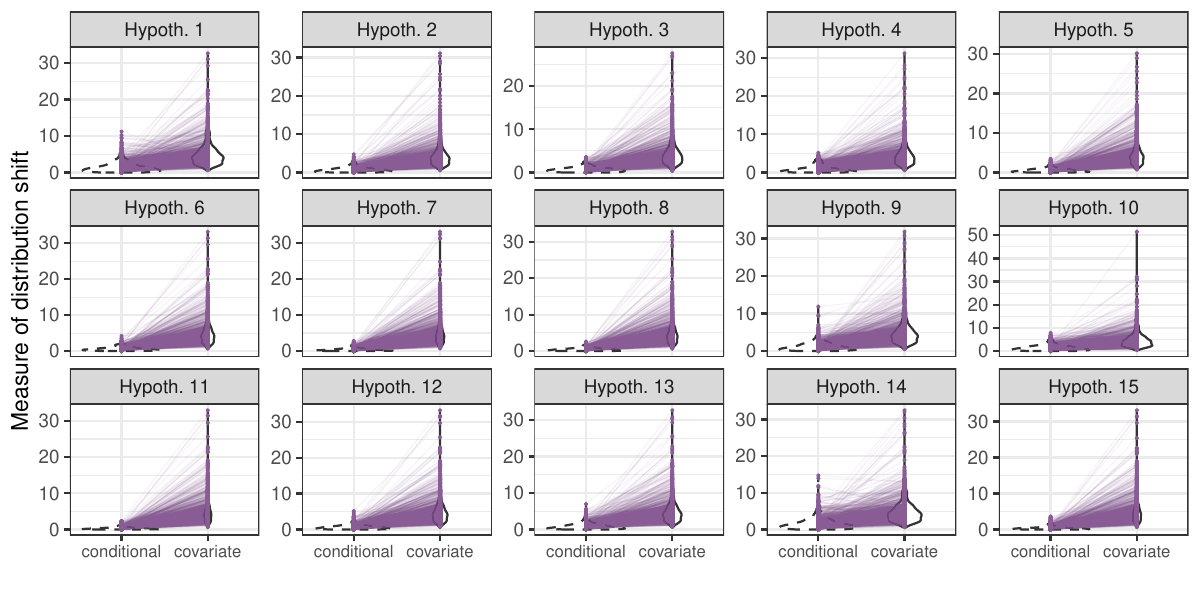}
    \caption{Distribution shift measures between all site pairs in each hypothesis in the ManyLabs1 project, where the covariate shift adjustment uses entropy balancing.}
    \label{fig:all_ML_eb_measure}
\end{figure}

\begin{figure}[H]
    \centering
    \includegraphics[width=\linewidth]{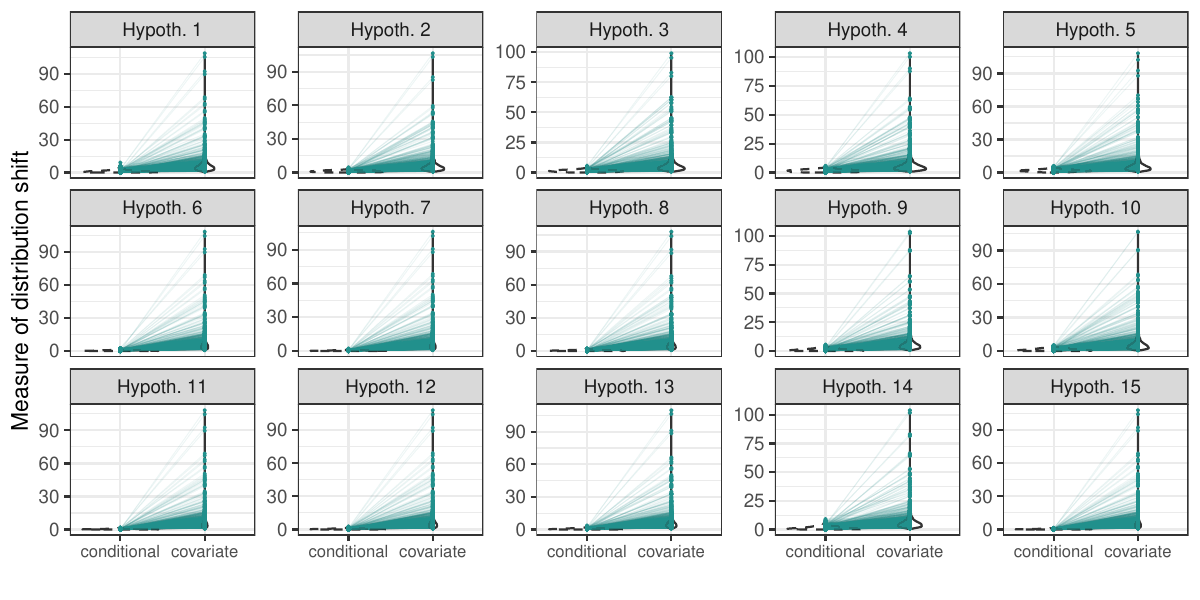}
    \caption{Distribution shift measures between all site pairs in each hypothesis in the ManyLabs1 project, where the covariate shift adjustment uses the doubly robust estimator.}
    \label{fig:all_ML_dr_measure}
\end{figure}

\section{Exploring alternative distribution shift measures}
\label{app:alt_shift_measures}

This selection collects our results on exploring alternative distribution shift measures. Here we exclusively focus on entropy-balancing-based estimation for stability and conciseness, while doubly-robust estimation exhibits similar patterns. Also, we only present results for the Pipeline project for conciseness. 

We introduce two sets of alternative distribution shift measures. By comparing our measures in the main text with them, we demonstrate the importance of (i) re-scaling by standard deviation to ensure scale invariance, and (ii) stabilizing the covariate shift measure in our definitions in Section~\ref{subsec:def_shift_measure}. 

\subsection{Alternative distribution shift measures}

Following the notations in Section~\ref{subsec:def_shift_measure}, we consider generalizing from a site with distribution $P$ to a site with distribution $Q$, and the parameter of interest has an influence function $\phi$. Recall that $\phi_P(x):=\EE_P[\phi\given X=x]$.

\paragraph{Marginal shift measures.} The first set of distribution shift measures follow \cite{jin2023diagnosing}, which are not rescaled by the standard deviation. Namely,  
\$
\textrm{absolute of conditional shift} ~= ~\EE_Q[\phi(T,Y) - \phi_P(X)] , \\ 
\textrm{absolute covariate shift} ~= ~\EE_Q[\phi_P(X)] - \EE_P[\phi_P(X)].
\$ 
These quantities describe the contributions of various types of distribution shifts to the discrepancy between effect estimates from two studies (sites) in \cite{jin2023diagnosing}. 
Their estimation is already included in Appendix~\ref{app:subsec_est_bound}, following which we denote the estimators as $\hat\Delta_{Y|X}$ and $\hat\Delta_X$, respectively. Namely,
\@
\hat\Delta_{Y|X}^{i\to j,(k)} :=  \hat\theta_j^{(k)} - \hat\theta_{i\to j}^{(k)} , \label{eq:est_cond_mgn} \\ 
\hat\Delta_{X}^{i\to j,(k)} :=  \hat\theta_{i\to j}^{(k)}- \hat\theta_i^{(k)}  , \label{eq:est_cov_mgn}  
\@
where $\hat\theta_{i\to j}^{(k)}$ is the reweighted estimator using full observations from site $i$ and covariates from site $j$. 

These unscaled measures may lack interpretability in certain cases. For one thing, the magnitude of these quantities depends on how sensitive the function $\phi$ is to shifts in the probability space: for instance, if $\phi(X)$ is highly heterogeneous, then even small changes in the distribution of $X$ would lead to large values of absolute covariate shift. While this is meaningful for diagnosing how the effect discrepancy relies on the distribution shifts and guiding future data collection efforts as in \cite{jin2023diagnosing}, this might be undesirable when we are interested in \emph{understanding the distribution shift itself}. 

We will see later that with marginal shift measures, the conditional shift is usually much larger than the covariate shift measure, which is consistent with the (somewhat pessimistic) findings in \cite{jin2023diagnosing} and a similar work of \cite{cai2023diagnosing}. This is mainly due to the fact that $\textrm{sd}(\phi-\phi_P(X))$ is much larger than $\textrm{sd}(\phi_P(X))$, i.e., the explanatory power of $X$ for the parameter is low. However, this hides the fact that the strength of perturbation to the probability space is indeed the other way.

\paragraph{Relative shift measures.} The second set of distribution shift measures follow Section~\ref{subsec:def_shift_measure}, but we adopt the relative conditional shift instead of the stabilized one. Thus, we call them relative shift measures. 
The estimation of the relative conditional shift is straightforward; following  Appendix~\ref{app:subsec_est_bound}, we use 
\@\label{eq:est_rel_X}
\hat{\Delta}_{\text{rel},X}^{i\to j, (k)} = \frac{\hat\Delta_{X}^{i\to j,(k)}}{\hat{s}_{i,X}} = \frac{\hat\theta_{i\to j}^{(k)} - \hat\theta_i^{(k)}  }{\hat{s}_{X}^{i\to j,(k)}},
\@
where $\hat\theta_{i\to j}^{(k)}$ is the reweighted estimator using entropy balancing or doubly robust estimator, and $\hat{s}_{X}^{i\to j,(k)}$ is a consistent estimator for $\textrm{sd}_P(\phi_P(X))$ which can be obtained following the estimation of $\textrm{sd}_P(\phi - \phi_P(X))$.

The issue with $\hat\Delta_{\textrm{rel},X}^{i\to j,(k)}$ is that the quantity $\hat{s}_X^{i\to j,(k)}$ can be extremely small in some cases when the explanatory power of $X$ for $\phi$ is low, as typical in the datasets we study here. Thus, even if we also observe a bounded role of relative covariate shift for relative conditional shift, the estimation is so unstable that it is not appropriate to be used in generalization tasks. 

\paragraph{Summary.} We summarize the three sets of distribution shift measures in Table~\ref{tab:measures} for the ease of reference. We also include the notations for their ratios to be used in the next two subsections. 

\begin{table}[H]
\renewcommand{\arraystretch}{1.2}
\centering
\begin{tabular}{|c|c|c|c|}
\hline
\textbf{Name} & \textbf{Conditional shift measure} & \textbf{Covariate Shift measure} & \textbf{Shift ratio} \\ \hline
Stabilized & $\hat{t}_{Y|X}$,~\eqref{eq:est_cond_shift} & $\hat{t}_{Y|X}$,~\eqref{eq:est_stab_X} & $\hat{r}^{\textrm{stab}}$ \\ \hline
Relative & $\hat{t}_{Y|X}$,~\eqref{eq:est_cond_shift} &  $\hat{\Delta}_{\textrm{rel},X}$,~\eqref{eq:est_rel_X} & $\hat{r}^{\textrm{rel}}$ \\ \hline
Marginal & $\hat\Delta_{Y|X}$,~\eqref{eq:est_cond_mgn} & $\hat\Delta_{X}$,~\eqref{eq:est_cov_mgn} & $\hat{r}^{\textrm{mgn}}$ \\ \hline
\end{tabular}
\caption{Summary of notations and estimations of distribution shift measures.}
\label{tab:measures}
\end{table}

\subsection{The importance of rescaling for the predictive role}

In this part, we demonstrate that rescaling is important for revealing the predictive role of covariate shift for the unknown conditional shift. 

Figure~\ref{fig:dist_shift_measure} plots the distribution (violin plots) and pairwise relations (connected segments) of each pair of distribution shift measures in Table~\ref{tab:measures} across all pairs of sites for Hypothesis 5 in the Pipeline project: 

\begin{itemize}
    \item First, in the left panel, the relationship between the marginal measures $\hat\Delta_{Y|X}$ and $\hat\Delta_{X}$ in each pair is somewhat arbitrary. This means knowledge of $\hat\Delta_X$ does not necessarily allow to control $\hat\Delta_{Y|X}$.
    \item The middle panel of Figure~\ref{fig:dist_shift_measure} shows that the relative measure of covariate shift $\hat\Delta_{\textrm{rel},X}$  bounds the conditional shift measure $\hat{t}_{Y|X}$ most of the time. 
This reveals the importance of normalization with standard deviation for interpretability. Without being scale-invariant, the marginal measures fail  to reveal the predictive role since the conditional ``sensitivity'', quantified by sd$(\phi-\phi_P(X))$, is larger than sd$(\phi_P(X))$. 
However,  estimated values of $\hat\Delta_{\textrm{rel},X}$ can be extremely large, since  sd$(\phi_P(X))$ and its estimated value can be tiny when the explanatory power of the covariates is low. This is also not desirable in practice as it will cause instability in downstream tasks such as effect generalization; we will explore this in the next part.
    \item  
Finally, the right panel of  shows the stabilized measures introduced in the main text. They reveal the predictive role of covariate shift due to scale invariance; in addition, they are  more stable than the relative measures, and the bounding role is tighter due to fewer extreme estimated values. 
\end{itemize}

\begin{figure}[h]
    \centering
    \includegraphics[width=0.9\linewidth]{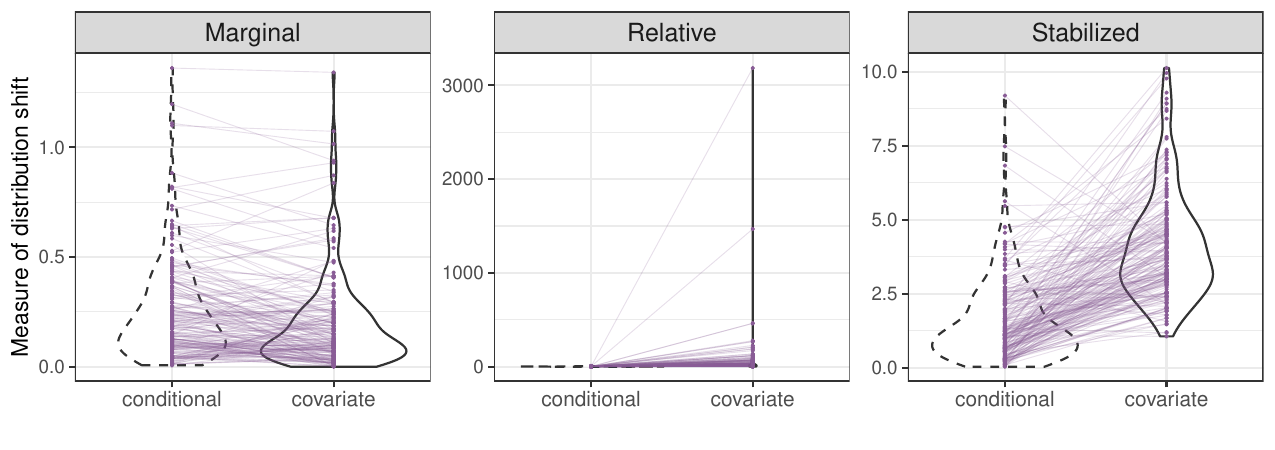}
    \caption{\textit{(Relative) magnitude of measures of conditional shift (dashed) and covariate shift (solid) across all site pairs in hypothesis 5 of the Pipeline project, analyzed with entropy balancing.} 
    \textbf{Left}: \textit{Marginal measures $\hat\Delta_{Y|X}$ and $\hat\Delta_X$.}
    \textbf{Middle}: \textit{Relative measures $\hat{t}_{Y|X}$ and $\hat\Delta_{\textnormal{rel},X}$.}
    \textbf{Right}: \textit{Stabilized measures $\hat{t}_{Y|X}$ and $\hat{t}_{X}$.}
    }
    \label{fig:dist_shift_measure}
\end{figure}
  
With a similar goal as panel (c) of Figure~\ref{fig:context_measure_PPML} in the main text, we explore the stability of the three sets of distribution shift ratios by their within-hypothesis quantiles. If quantiles of the ratios are stable across hypotheses, then the ratio is ``pivotal'' and generalizable, meaning that external knowledge of the magnitude of distribution shift ratios from other data sources may be useful for the data at hand. 
From Figure~\ref{fig:context_measure_PPML}, we see that the quantiles of the marginal ratios and ralative ratios are quite variable. In contrast, the within-hypothesis quantiles of the stabilized ratio are more ``pivotal''; they are stable across hypotheses and also close to the global quantile. We will see next the implications of such stability for generalization.

\begin{figure}
    \centering
    \includegraphics[width=\linewidth]{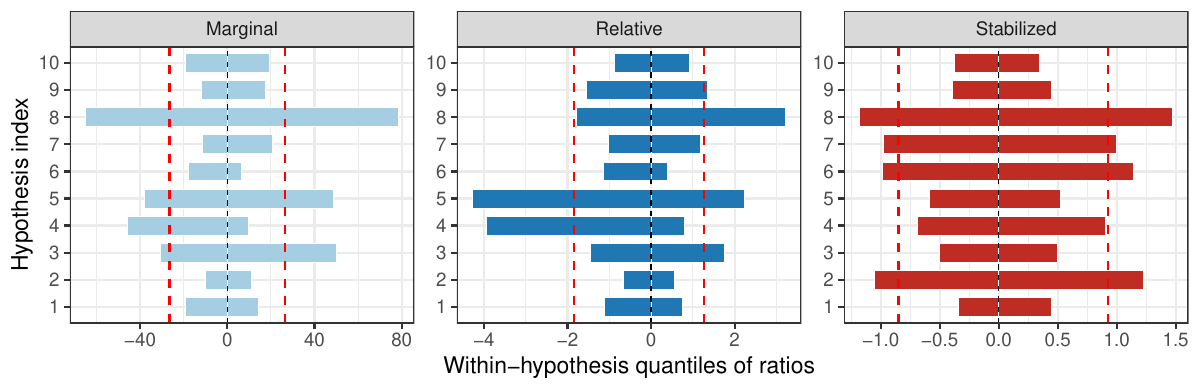}
    \caption{\textit{Lower and upper within-hypothesis quantiles of ratios which, once known, lead to exact $95\%$ empirical coverage of the prediction intervals for the Pipeline dataset. 
   The left ends of the bar plot are the lower quantiles; the right ends are the upper quantiles. The red dashed lines are the global quantiles over all studies.  
   Ideally, the quantiles should be invariant across studies for meaningful empirical calibration.
   \textnormal{\textbf{Left}:} quantiles of $\hat{r}^\mgn$ (marginal); 
    \textnormal{\textbf{Middle}:} quantiles of  $\hat{r}^{\textnormal{rel}}$ (relative);
    \textnormal{\textbf{Right}:} quantiles of $\hat{r}^\stab$ (stabilized). }}
    \label{fig:ratio_qt_eb}
\end{figure}

\subsection{The importance of stability for generalization}

We evaluate generalization tasks similar to Section~\ref{sec:generalize} with the three sets of distribution shift measures. Again, similar to the ideas of~\eqref{eq:PI_high_level}, we construct prediction intervals for the target site estimator $\hat\theta_j^{(k)}$ by calibrating lower and upper bounds for the ratios between each suite of distribution shifts. The detailed estimation procedures follow Appendix~\ref{app:subsec_est_generalize}.


Specifically, we aim to find lower and upper bounds for the ratios, such that (approximately)
\@\label{eq:LU_other}
&\PP\Big( L^{\mgn} \leq \hat{r}_{i\to j}^{\mgn,(k)} \leq U^{\mgn} \Big) \geq 1-\alpha, \quad \hat{r}_{i\to j}^{\mgn,(k)} = \hat\Delta_{Y|X}^{i\to j,(k)} / \,\hat\Delta_{X}^{i\to j,(k)},\\
&\PP\Big( L^{\relat} \leq \hat{r}_{i\to j}^{\relat,(k)} \leq U^{\relat} \Big) \geq 1-\alpha,\qquad \hat{r}_{i\to j}^{\relat,(k)} = \hat{t}_{Y|X}^{i\to j,(k)} / \,\hat\Delta_{\relat,X}^{i\to j,(k)}, \notag 
\@
and the bounds for the ratio between stabilized measures in the main text follow the idea of~\eqref{eq:LU_high_level}. 

Inverting the events in~\eqref{eq:LU_other} and by the definition of the measures, we set the prediction intervals
\$
\hat{C}_{i\to j}^{\mgn,(k)}:= \big[~ & \hat\theta_{i\to j}^{(k)} + \hat\Delta_{X}^{i\to j,(k)} \cdot L^{\mgn},~\hat\theta_{i\to j}^{(k)} + \hat\Delta_{X}^{i\to j,(k)} \cdot U^{\mgn} ~\big] \\ 
\hat{C}_{i\to j}^{\relat,(k)}:= \big[~ & \hat\theta_{i\to j}^{(k)} + \hat\Delta_{\relat,X}^{i\to j,(k)} \cdot \hat{s}_{Y|X}^{i\to j, (k)} \cdot  L^{\relat} ,  
~ \hat\theta_{i\to j}^{(k)} + \hat\Delta_{\relat,X}^{i\to j,(k)} \cdot \hat{s}_{Y|X}^{i\to j, (k)} \cdot  U^{\relat} ~\big] , 
\$
and recall that $\hat{C}_{i\to j}^{(k)}$ is the prediction interval~\eqref{eq:PI_high_level} based on our shift measures in the main text. 
We then evaluate the empirical coverage and average length of these prediction intervals.

\paragraph{Oracle calibration.} For reference, we evaluate the \texttt{Oracle} method in the main text for the three sets of shift measures, in order to show their performance in the most ideal case where the distribution of their ratios is perfectly known. Here, the $L$ and $U$ values in~\eqref{eq:LU_other} are the empirical quantiles of the shift ratios between all site pairs within each hypothesis. The empirical coverage and average length of prediction intervals within each hypothesis are in Figure~\ref{fig:instudy_PI_eb}. All three sets of measures lead to perfect $0.95$ coverage as expected. However, the prediction intervals by the stabilized measures several folds shorter (the $y$-axis is log-scaled for easier visualization), showing the importance of estimation stability.

\begin{figure}[htbp]
    \centering
    \includegraphics[width=0.85\linewidth]{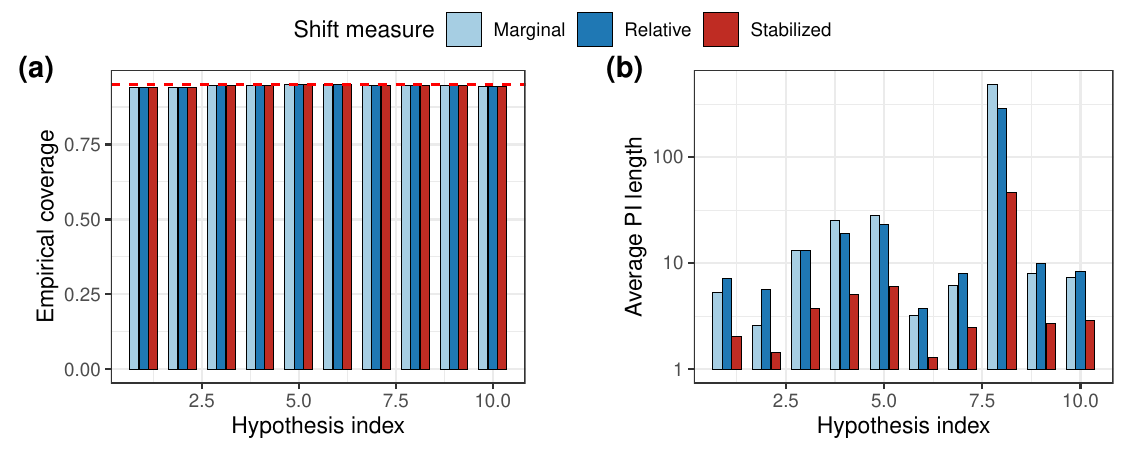}
    \caption{\textbf{Left}: \textit{Empirical coverage of oracle calibrated  prediction intervals at nominal level $1-\alpha=0.95$. The coverage is ensured to be $95\%$ since full observations are used.}
    \textbf{Right}: \textit{Average length of prediction intervals for in-study calibrated prediction intervals at nominal level $1-\alpha=0.95$ based on three measures. The $y$-axis on the right is log-scaled for visualization.}}
    \label{fig:instudy_PI_eb}
\end{figure}

\paragraph{Constant calibration.} 
Similar to Section~\ref{sec:generalize}, here we simply take all three lower quantiles to be $-1$, and all three upper quantiles to be $1$, with the belief that the covariate shift measure upper bounds the conditional shift measure. 
The hypothesis-wise coverage and average length of constant-calibrated prediction intervals are in Figure~\ref{fig:const_PI_eb}. It is not surprising that assuming $|\hat\Delta_{Y|X}|\leq |\hat\Delta_X|$ leads to poor coverage (marginal). 
Assuming that the conditional shift measure is bounded by the covariate shift measure leads to satisfying coverage for both the relative and stabilized measure. 
However, the stabilized measures lead to much shorter prediction intervals and slightly better coverage again due to stability. 

\begin{figure}[htbp]
    \centering
    \includegraphics[width=0.85\linewidth]{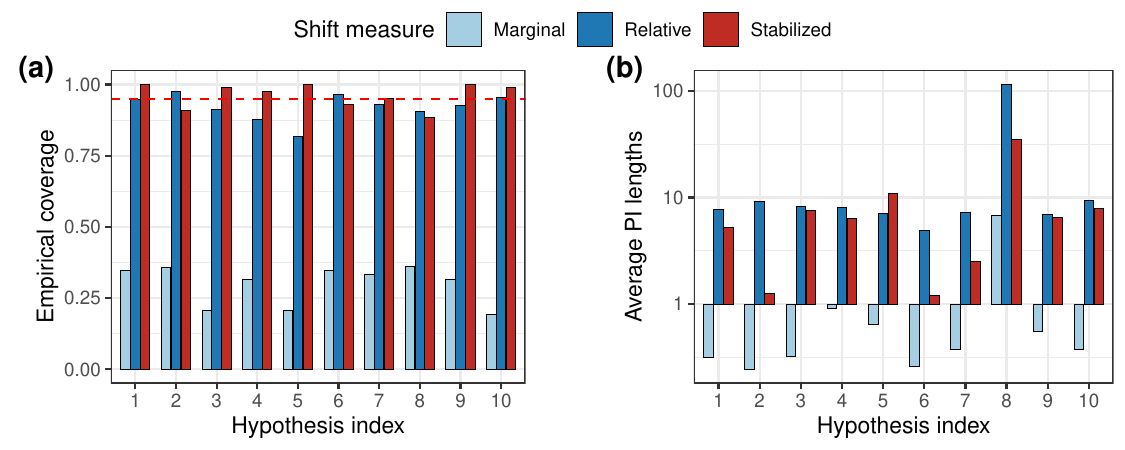}
    \caption{\textbf{Left}: \textit{Empirical coverage of constant calibrated  prediction intervals at nominal level $1-\alpha=0.95$.}
    \textbf{Right}: \textit{Average length of prediction intervals for constant calibrated prediction intervals at nominal level $1-\alpha=0.95$ based on three measures. The $y$-axis on the right is log-scaled for visualization.}}
    \label{fig:const_PI_eb}
\end{figure}

\paragraph{Data-adaptive calibration.} 
Finally, we consider the data-adaptive calibration scenario where full observations from other sites/hypotheses are available, which are used to compute the quantiles for a new generalization task. 
The specific methods are the same as Section~\ref{sec:generalize} and Appendix~\ref{sec:appendix-calibration}, with detailed procedures following Appendix~\ref{app:subsec_est_generalize}.

The first scenario is the same as Section~\ref{sec:generalize} in the main text, where data for some other hypotheses in all sites are available, and the new generalization task involves new hypotheses.  Figure~\ref{fig:gen_PI_study} illustrates the order of data collection, as well as the coverage and length of prediction intervals, averaged over 10 random draws of hypothesis ordering. 
We see that all three measures lead to satisfactory coverage, meaning that \emph{the distribution shift measures tend to be ``generalizable'' across hypotheses/estimators/outcomes}. Yet, the stabilized measures still yield the shortest prediction intervals. 

\begin{figure}[htbp]
    \centering
    \includegraphics[width=\linewidth]{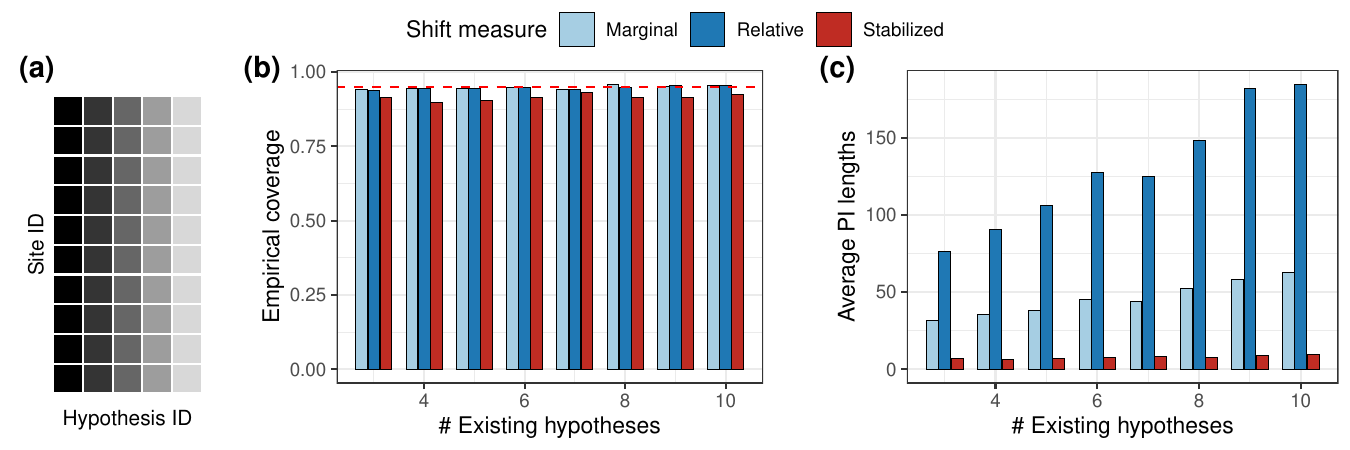}
    \caption{\textit{Generalization  based on distribution shift measures calibrated with data for other hypotheses in the same sites.} \textbf{Left}: \textit{Illustration of data collection order, where dark color means earlier.} 
    \textbf{Middle}: \textit{Average coverage (bars) of prediction intervals over 10 random draws of study ordering. The red dashed line is the nominal level $0.95$}. 
    \textbf{Right}: \textit{Average length of prediction intervals based on three sets of shift measures}.}
    \label{fig:gen_PI_study}
\end{figure}

The second scenario is to calibrate the measures with existing sites involving all hypotheses for new sites, same as ``calibration with other sites'' in Appendix~\ref{sec:appendix-calibration}. 
Figure~\ref{fig:gen_PI_site} presents the corresponding results. We see that all measures lead to satisfactory coverage, while the stabilized measures lead to much shorter prediction intervals.  

\begin{figure}[htbp]
    \centering
    \includegraphics[width=\linewidth]{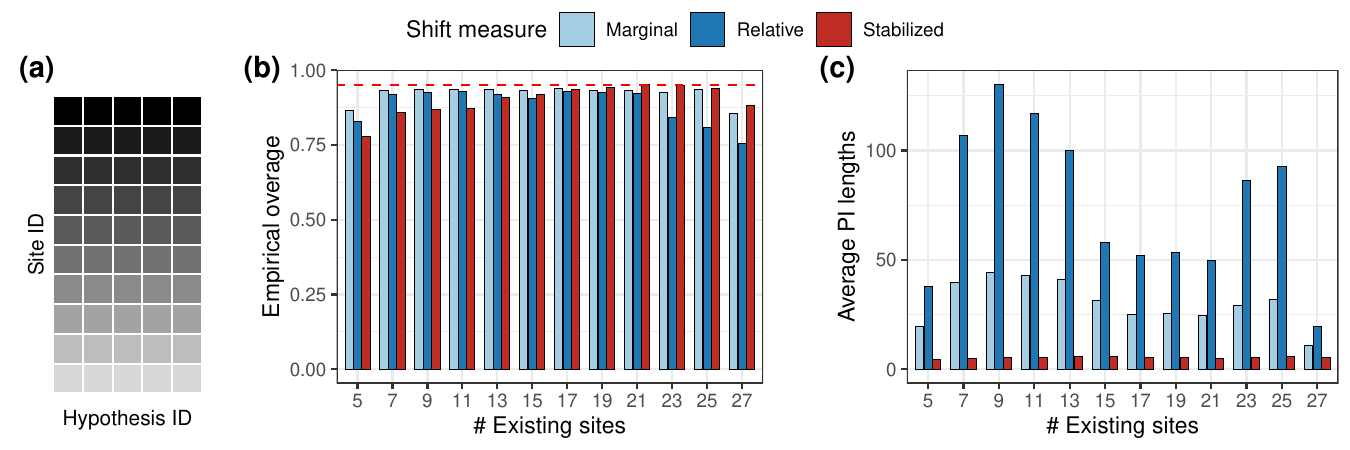}
    \caption{\textit{Data collection order, average coverage and length of prediction intervals for generalization in new sites based on data from the same studies in other sites. Details are otherwise the same as in Figure~\ref{fig:gen_PI_study}.}}
    \label{fig:gen_PI_site}
\end{figure}

Finally, we use data from other sites for other hypotheses to calibrate the upper and lower bounds for the distribution shift measures, which is the same as ``Calibration with other sites and other hypothesis'' in Appendix~\ref{sec:appendix-calibration}. Figure~\ref{fig:gen_PI_both} presents the results for the third scenario. Due to the limited samples, we observe slight undercoverage when only two hypotheses and sites are available. Similar to other scenarios, the stabilized measure leads to much shorter prediction intervals than the other two. 

\begin{figure}[htbp]
    \centering
    \includegraphics[width=\linewidth]{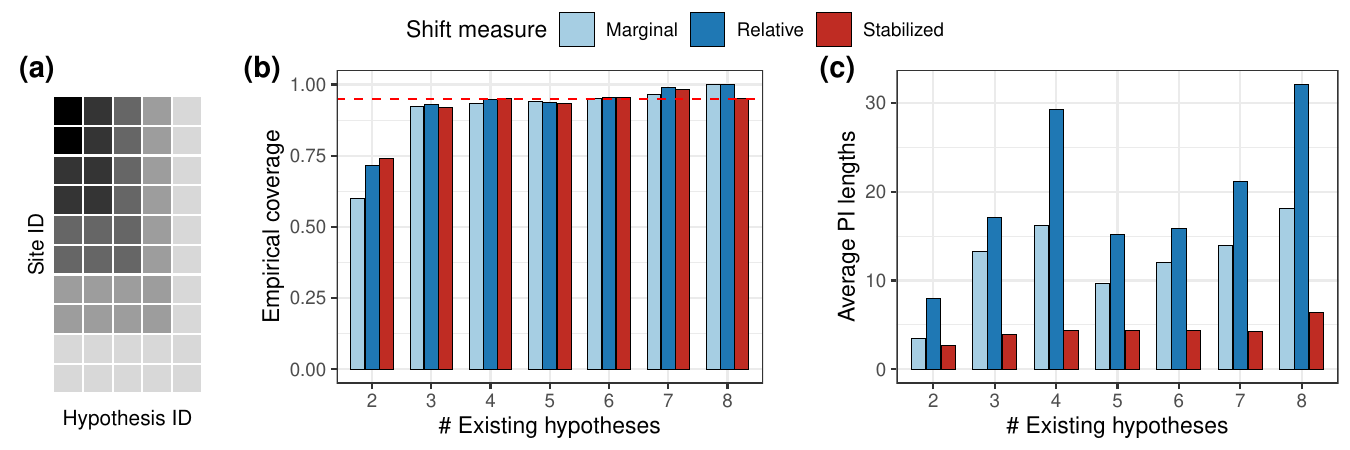}
    \caption{\textit{Data collection order, average coverage and length of prediction intervals for generalization between new sites in new studies based on data for other studies from other sites. Details are otherwise  as in Figure~\ref{fig:gen_PI_study}.}}
    \label{fig:gen_PI_both}
\end{figure}

\section{Technical details and proofs}

\subsection{Proof of distributional CLT}

\begin{proof}
Let $n_Q(M)$ and $n_P(M)$ be sequences of natural numbers such that $n_Q(M)/M$ and $n_P(M)/M$ converge to positive real numbers. In the following, for simplicity, we will suppress the dependence of $n_Q$ and $n_P$ on $M$. We will first show the result for bounded, one-dimensional $\phi$ with $\EE_P[\phi]=0$.
Define $\sigma_M^2 = \frac{1}{M} \sum_{m=1}^M \mathbb{E}_P[\phi| C_m^{(M)}]^2 - \EE_P[\phi]^2 $. As $M \rightarrow \infty$, by assumption we have $\sigma_M^2 \rightarrow \text{Var}_P(\EE_P[\phi|X,U]) =: \sigma^2$. If $\text{Var}_P(\EE_P[\phi|X,U]) = 0$, then $\EE_P[\phi] = \EE_Q[\phi]$, thus only uncertainty due to i.i.d.\ sampling remains and the statement of the theorem is trivial. Thus, in the following we will assume $\text{Var}_P(\EE_P[\phi|X,U]) > 0$. We can use $\EE_P[\phi] = \frac{1}{M} \sum_{m=1}^M \EE_P[\phi|C_m^{(M)}]$ to obtain
\begin{align*}
    &\frac{\sqrt{M}( \EE_Q[\phi] - \EE_P[\phi] )} {\sigma_M \text{sd}(W)/E[W]}  \\
    &= \frac{\frac{ M^{-1/2} \sum_{m=1}^M (W_m - 1/M \sum_{m'} W_{m'}) \EE_P[\phi|C_m^{(M)}] }{1/M \sum_m W_m}  }{\sigma_M \text{sd}(W)/E[W]} \\
   &= \frac{M^{-1/2} \frac{\sum_{m=1}^M (W_m - 1/M \sum_{m'} W_{m'}) \EE_P[\phi|C_m^{(M)}] }{E[W]}  }{\sigma_M \text{sd}(W)/E[W]}  + o_P(1/\sqrt{M})\\ 
   &=  \frac{M^{-1/2} \sum_{m=1}^M (W_m - \EE[W] ) \EE_P[\phi|C_m^{(M)}]    }{\sigma_M \text{sd}(W){}}  + o_P(1/\sqrt{M}). \\
\end{align*}
We will now check Lindeberg's condition with $v_M = \sum_{m=1}^M \mathbb{E}_P[\phi|C_m^{(M)}]^2 - \EE_P[\phi]^2$. By assumption $v_M/M \rightarrow \sigma^2 > 0$. Furthermore, by assumption $| \phi |_\infty \le B$ for some constant $B> 0$. 
Let $\epsilon > 0$. Then
\begin{align*}
   &\limsup_{M \rightarrow \infty} \frac{1}{v_M} \EE[ \sum_{m=1}^M (W_m - \EE[W] )^2 (\EE_P[\phi|C_m^{(M)}] - \EE_P[\phi])^2 1_{|W_m - \EE[W]| | \EE_P[\phi|C_m^{(M)}] - \EE_P[\phi]|\ge \epsilon \sqrt{v_M} } ] \\
&\le \limsup_{M \rightarrow \infty} \frac{1}{M \sigma^2} \EE[ \sum_{m=1}^M (W_m - \EE[W] )^2 4 B^2 1_{| W_m - \EE[W]| \ge  \sqrt{M} \sigma \epsilon/(4 B) } ] \\
&\le \limsup_{M \rightarrow \infty} \frac{4 B^2}{ \sigma^2} \EE[ (W - \EE[W])^2  1_{ |W - \EE[W]| \ge  \sqrt{M} \sigma \epsilon/(4B)}  ]
\end{align*}
By dominated convergence, this term is zero. Thus, by Lindeberg's CLT,
\begin{align*}
   \frac{\sqrt{M}( \EE_Q[\phi] - \EE_P[\phi] )} {\sigma_M\text{sd}(W)/E[W]} &\stackrel{d}{\rightarrow} \mathcal{N}(0,1).
\end{align*}
 By Slutsky, we get the distributional CLT
\begin{equation*}
    \frac{\sqrt{M}( \EE_P[\phi] - \EE_Q[\phi] )} {\text{sd}_P(\EE_P[\phi|X,U] )\text{sd}(W)/E[W]} \stackrel{d}{\rightarrow} \mathcal{N}(0,1).
\end{equation*}
We will now combine this result with uncertainty due to i.i.d.\ sampling. Recall that we consider the case where sampling uncertainty and distributional uncertainty are of the same order, i.e.\ $n_Q/M$ and $n_P/M$ converge to some positive constants.
\begin{align*}
    &\hat{\EE}_P[\phi(T,D)] - \hat{\EE}_Q[\phi(T,D)] \\
    & = \underbrace{\hat{\EE}_P[\phi(T,D)] -\mathbb{E}_P[\phi(T,D)]}_{\text{use standard CLT for i.i.d.\ data}} - \underbrace{(\hat{\EE}_Q[\phi(T,D)]  - \mathbb{E}_Q[\phi(T,D)] )}_{\text{use Berry-Esseen}} + \underbrace{\mathbb{E}_P[\phi(T,D)] - \mathbb{E}_Q[\phi(T,D)]}_{\text{use distributional CLT}}
\end{align*}
We can apply a standard CLT to the first term, since $P$ is fixed and the sample mean $\hat{\EE}_P[\phi(T,D)]$ is independent of the remaining terms. For the remaining terms, one complication is that the distribution $Q$ is not fixed, but shifts randomly.

Recall that for now we focus on bounded $\phi$, which implies bounded third moments. We will now use Berry–Esseen. For any $x \in \mathbb{R}$, conditionally on the random shift $(W_m)_{m=1,\ldots,M}$, 
\begin{equation}
  \sup_x \left| P \left( \frac{\sqrt{n_Q}}{\text{sd}_Q(\phi) } (\hat{\EE}_Q[\phi(T,D)]  - \mathbb{E}_Q[\phi(T,D)] )  \le x 
 \Bigg| (W_m)_{m=1,\ldots,M} \right) - \Phi(x) \right| \le \frac{7.59 \EE_Q[|\phi|^3]}{ \text{sd}_Q(\phi)^3 \sqrt{n_Q}}
\end{equation}
By assumption $\phi$ is bounded, and by the distributional CLT above $\text{sd}_Q(\phi) \stackrel{P}{\rightarrow} \text{sd}_P(\phi) > 0$ for $M \rightarrow \infty$. Thus, conditionally on the random shift $(W_m)_{m=1,\ldots,M}$,
\begin{equation*}
    \sup_x \left| P \left( \frac{\sqrt{n_Q}}{ \text{sd}_P(\phi)} (\hat{\EE}_Q[\phi(T,D)]  - \mathbb{E}_Q[\phi(T,D)] )  \le x \Bigg| (W_m)_{m=1,\ldots,M} \right) - \Phi(x)  \right| \rightarrow 0
\end{equation*}
Define
\begin{align*}
    Z := \mathbb{E}_P[\phi(T,D)] - \mathbb{E}_Q[\phi(T,D)], \\
    Z' := \hat{\EE}_Q[\phi(T,D)] - \mathbb{E}_Q[\phi(T,D)].
\end{align*}
Let $\delta_M^2 = \frac{1}{M} \frac{\text{Var}(W)}{E[W]^2}$. By assumption, $n_Q/M \rightarrow \rho$ for some constant $\rho > 0$. Thus $n_Q \delta_K^2 \rightarrow \rho \text{Var}(W)/E[W]^2$. As $M \rightarrow \infty$, for any $x \in \mathbb{R}$
\begin{align*}
   & P \left( \left( \frac{\text{Var}_P(\phi)}{n_Q} + \sigma^2\delta_M^2 \right)^{-1/2} (Z + Z') \ge x \right)  \\
   &= E \left[P \left( \frac{\sqrt{n_Q}}{\text{sd}_P(\phi)} Z' \ge \frac{\sqrt{n_Q}}{\text{sd}_P(\phi)} \left( \frac{\text{Var}_P(\phi)}{n_Q} + \sigma^2 \delta_M^2 \right)^{1/2} x - \frac{\sqrt{n_Q}}{\text{sd}_P(\phi)} Z \Bigg| (W_m)_{m=1,\ldots,M} \right) \right] \\
   &\stackrel{\text{Berry-Esseen}}{=} E \left[ 1- \Phi \left( \left(1 + \rho \frac{\sigma^2 \text{Var}(W)}{\text{Var}_P(\phi) E[W]^2} \right)^{1/2} x - \frac{\sqrt{n_Q}}{\text{sd}_P(\phi)} Z \right) \right] + o(1)
\end{align*}
In the third line, we used that $\text{sd}_Q(\phi) \stackrel{P}{\rightarrow} \text{sd}_P(\phi)$. Here, $\Phi$ is the cdf of a standard Gaussian random variable. Using weak convergence of $\sqrt{M} Z \stackrel{d}{\rightarrow} \mathcal{N}(0, \sigma^2 \text{Var}(W)/E[W]^2)$, and using that $\Phi$ is a continuous bounded function we get
\begin{align*}
    & E \left[ 1- \Phi \left( \left(1 + \rho \frac{\sigma^2 \text{Var}(W)}{\text{Var}_P(\phi) E[W]^2} \right)^{1/2} x - \frac{\sqrt{n_Q}}{\text{sd}_P(\phi)} Z \right) \right] \\
    \rightarrow & E \left[1- \Phi \left( \left(1 + \rho \frac{\sigma^2 \text{Var}(W)}{\text{Var}_P(\phi) E[W]^2} \right)^{1/2} x - \sqrt{\rho} \cdot \frac{\sigma \text{sd}(W)}{\text{sd}_P(\phi) E[W]} G \right) \right], 
\end{align*}
where $G$ is a standard Gaussian random variable. With constant $L =\sqrt{\rho} \sigma \text{sd}(W) / (\text{sd}_P(\phi) E[W])$ we can re-write this as
\begin{equation*}
    E[1-\Phi(\sqrt{1+L^2} x - L G)] = P(G' \ge \sqrt{1+L^2} x - L G) = P( \frac{G' + LG }{\sqrt{1+L^2}} \ge x) = 1- \Phi(x),
\end{equation*}
where $G'$ is a standard Gaussian random variable, independent of $G$. To summarize,
\begin{equation*}
    P \left( \left( \frac{\text{Var}_P(\phi)}{n_Q} + \sigma^2\delta_M^2 \right)^{-1/2} (Z + Z') \ge x \right) \rightarrow 1- \Phi(x).
\end{equation*}
Since the data from $P$ is independent of the perturbation and the data from $Q$,
\begin{equation}\label{eq:clt-bounded-phi}
  \left(  \text{Var}_P(\phi) \left( \frac{1}{n_P} + \frac{1}{n_Q} \right) + \sigma^2\delta_M^2  \right)^{-1/2}   \left(\hat{\EE}_P[\phi(T,D)] - \hat{\EE}_Q[\phi(T,D)] \right) \stackrel{d}{\rightarrow} \mathcal{N}(0,1).
\end{equation}
We will now extend this result from bounded functions $\phi$ to square-integrable functions $\phi \in L^2(P)$.
Define the bounded function
\begin{equation*}
    \phi_b = \phi 1_{|\phi| \le b} + 1_{|\phi| \ge b} \EE_P \left[\phi \Bigg| |\phi| \ge b \right].
\end{equation*}
We have $\EE_P[\phi] = \EE_P[\phi_b]$.
Applying Chebychev conditionally on the random perturbation,
\begin{equation*}
    P \left( |\hat{\EE}_Q[\phi_b-\phi] -\hat{\EE}_Q[\phi_b-\phi]| \ge \epsilon \Bigg| (W_m)_{m=1,\ldots,M} \right) \le \frac{\text{Var}_Q(\phi_b - \phi)}{ \epsilon^2 } \left( \frac{1}{n_P} + \frac{1}{n_Q} \right)
\end{equation*}
Take expectations over the random perturbation $(W_m)_{m=1,\ldots,M}$ we obtain that for all $\epsilon> 0$,
\begin{equation*}
    P (|\hat{\EE}_Q[\phi_b-\phi] - \EE_Q[\phi_b-\phi]| \ge \epsilon ) \le \frac{\EE_P[(\phi-\phi_b)^2]}{\epsilon^2 } \left( \frac{1}{n_P} + \frac{1}{n_Q} \right)
\end{equation*}
Similarly, since $W_m \ge w_0$,
\begin{equation*}
  | \EE_Q[\phi - \phi_b] |=  |\sum_{m=1}^M \frac{W_m}{\sum_{m'} W_{m'}} \EE_{P}[\phi - \phi_b|C_m] | \le \frac{1}{M w_0} |\sum_{m=1}^M W_m  \EE_{P}[\phi - \phi_b|C_m] |.
\end{equation*}
Applying Chebychev and using that $\EE_P[\phi - \phi_b] = 0$, we get
\begin{equation*}
    P( | \EE_Q[\phi - \phi_b]  | \ge \epsilon )  \le \frac{ \sum_{m=1}^M (\EE_P[\phi-\phi_b|C_m])^2}{ w_0^2 \epsilon^2 M^2} \le \frac{ \text{Var}_P(\phi-\phi_b)}{w_0^2 \epsilon^2 M}.
\end{equation*}
In the last equation, we used Jensen's inequality. Combining the two applications of Chebychev,
\begin{align*}
    P( | \hat{\EE}_Q[\phi - \phi_b]  | \ge \epsilon ) &\le P( | \EE_Q[\phi - \phi_b]  | \ge \epsilon/2 ) + P( | \hat{\EE}_Q[\phi - \phi_b] -\EE_Q[\phi - \phi_b]  | \ge \epsilon/2 )  \\
    &\le \frac{4}{\epsilon^2} \left( \frac{1}{ n_Q} + \frac{1}{n_P} + \frac{1}{w_0^2 M} \right) \text{Var}_P(\phi - \phi_b)
\end{align*}
Since by assumption $n_P(M) \sim n_Q(M) \sim M$, for any $\epsilon'>0$, $\epsilon> 0$, we can choose a bounded function $\phi_b$ such that $ P(\sqrt{M} |\hat{E}_Q[\phi-\phi_b] - E_P[\phi-\phi_b]| \ge \epsilon) \le \epsilon' $ for $M \rightarrow \infty$.  Let
\begin{equation*}
\sigma_{b,M}^2 = \left( \text{Var}_P(\phi_b) \left( \frac{1}{n_Q} + \frac{1}{n_P} \right) + \text{Var}_P(\EE_P[\phi_b|X,U]) \delta_M^2 \right). 
\end{equation*}
Then, for any $\epsilon' > 0$ there exists a $\epsilon > 0$ such that as $ M \rightarrow \infty$, we have
\begin{equation}\label{eq:small-residual}
 P \left(  \sigma_{b,M}^{-1} |\hat{E}_Q[\phi-\phi_b] - E_P[\phi-\phi_b]| \ge \epsilon   \right) \le \epsilon'.
\end{equation}
For any $\delta > 0$, for $b > 0$ large enough
\begin{equation}
   \text{Var}_P(\phi) \left( \frac{1}{n_P} + \frac{1}{n_Q} \right) + \text{Var}_P(\EE_P[\phi|X,U]) \delta_M^2    \le (1 + \delta)^2 \sigma_{b,M}^2.
\end{equation}
Then, for $ M \rightarrow \infty$,
\begin{align*}
    & \limsup_{M \rightarrow \infty} P \left( \left(  \text{Var}_P(\phi) \left( \frac{1}{n_P} + \frac{1}{n_Q} \right) + \text{Var}_P(\EE_P[\phi|X,U]) \delta_M^2 \right)^{-1/2} \left( \hat{\EE}_P[\phi] - \hat{\EE}_Q[\phi] \right) \le x   \right) \\
    &\le \limsup_{M \rightarrow \infty} P \left( \sigma_{b,M}^{-1} \left( \hat{\EE}_P[\phi] - \hat{\EE}_Q[\phi] \right) \le \max(x(1+\delta),x)   \right) \\
    &\le \limsup_{M \rightarrow \infty} P \left( \sigma_{b,M}^{-1} \left( \hat{\EE}_P[\phi_b] - \hat{\EE}_Q[\phi_b] \right) \le \max(x(1+\delta),x) + \epsilon \right) \\
    & + \limsup_{M \rightarrow \infty} P \left( \sigma_{b,M}^{-1} \left| \hat{\EE}_P[\phi - \phi_b] - \hat{\EE}_Q[\phi - \phi_b] \right|  \ge \epsilon \right) \\
    &\le \Phi(\max(x(1+\delta),x) + \epsilon)  +  \epsilon'.
\end{align*}
In the last line, we used equation~\eqref{eq:clt-bounded-phi} and equation~\eqref{eq:small-residual}. Since $\delta >0$, $\epsilon'>0$ and $\epsilon > 0$ can be chosen arbitrary small, 
\begin{equation*}
    \limsup_{ M \rightarrow \infty} P \left( \left(  \text{Var}_P(\phi) \left( \frac{1}{n_P} + \frac{1}{n_Q} \right) + \text{Var}_P(\EE_P[\phi|X,U]) \delta_M^2 \right)^{-1/2} \left( \hat{\EE}_P[\phi] - \hat{\EE}_Q[\phi] \right) \le x   \right) \le \Phi(x).
\end{equation*}
With an analogous argument,
\begin{equation*}
    \liminf_{ M \rightarrow \infty} P \left( \left(  \text{Var}_P(\phi) \left( \frac{1}{n_P} + \frac{1}{n_Q} \right) + \text{Var}_P(\EE_P[\phi|X,U]) \delta_M^2 \right)^{-1/2} \left( \hat{\EE}_P[\phi] - \hat{\EE}_Q[\phi] \right) \le x   \right) \ge \Phi(x).
\end{equation*}
Thus, as $M \rightarrow \infty$,
\begin{equation*}
    \left( \text{Var}_P(\phi) \left( \frac{1}{n_P} + \frac{1}{n_Q} \right) + \text{Var}_P(\EE_P[\phi|X,U]) \delta_M^2 \right)^{-1/2} \left( \hat{\EE}_P[\phi] - \hat{\EE}_Q[\phi] \right)  \stackrel{d}{\rightarrow} \mathcal{N}(0,1).
\end{equation*}
This completes the proof for one-dimensional $\phi$. The result for a vector of functions $\phi$ follows by applying the Cram\'er-Wold device.

\end{proof}

\subsection{Estimation for conditional variances}
\label{app:subsec_est_variance}

In this part, we detail the estimation of the conditional variances $\Var_P(\phi_P(X))$ and $\Var_P(\phi(X,Y,T)-\EE_P[\phi(X,Y,T)\given X])$ in the construction of our distribution shift measures, which is omitted from Appendix~\ref{app:subsec_est_bound}.

Recall that $P$ is the underlying distribution of the ``source'' site, and $Q$ is that of the ``target'' site. We write the influence function in the most general from $\phi(X,Y,T)$, though in the datasets it is a function of only $(Y,T)$. 
We will use cross-fitting~\citep{chernozhukov2018double} with machine learning models for fitting the conditional mean functions. 
   
The variance estimation only needs data from the source site,  $\cD_1\iid P$. We randomly split $\cD_1$ into two folds  $\cD_1^{(1)}\cup \cD_1^{(2)}$. 
For $k=1,2$, we use $\cD_1^{(k)} $ to fit the conditional mean function $\hat\phi^{(k)}(\cdot)$ for $\phi_P(\cdot):=\EE_P[\phi(X,Y,T)\given X=\cdot]$. 
Then, writing   $\hat\varphi(X_i)=\hat\varphi^{(k)}(X_i)$ for $i\in \cD_1\backslash\cD_1^{(k)}$, and $\hat\varphi(X_j) = \hat\varphi^{(k)}(X_j)$ for  $j\in \cD_2\backslash \cD_2^{(k)}$, we let 
\$
\hat{s}_{Y|X}^2 &= \frac{1}{n_1}\sum_{i\in \cD_1} \big(\phi(X_i,Y_i,T_i)-\hat\varphi(X_i)\big)^2,\\ 
\hat{s}_X^2 &= \frac{1}{n_1}\sum_{i\in \cD_1} \hat\varphi(X_i) \cdot \big(2\phi(X_i,Y_i,T_i) - \hat\varphi(X_i) \big) - \bigg(\frac{1}{n_1}\sum_{i\in \cD_1} \phi(X_i,Y_i,T_i)\bigg)^2.
\$ 
We define $\sigma(x) = \Var_P(\phi(X,Y,T)\given X=x)$. 
Our estimators converge in $n^{-1/2}$ rate under standard slow convergence rates of nuisance components. 
\begin{theorem}\label{thm:normal}
Suppose  $\|  \hat\varphi^{(k)}-\varphi \|_{L_2(P)} = o_P(n^{-1/4})$, 
 and $\|\sigma\cdot  (\hat{\varphi}^{(k)}-\varphi) \|_{L_2(P)} = o_P(1)$ for $k=1,2$. Then,
\$
\begin{pmatrix}
    \hat{s}_{Y|X} \\ 
    \hat{s}_{X} 
\end{pmatrix} = \begin{pmatrix}
    {s}_{Y|X} \\ 
     {s}_{X} 
\end{pmatrix}  +\frac{1}{n_1}\sum_{i\in \cD_1} \psi_1(X_i,Y_i,T_i)  + o_P(1/\sqrt{n_1 }  )
\$
for some fixed function  $\psi $ with mean zero. As a result, each element is consistent and asymptotically $\sqrt{n}$-normal, and the asymptotic variances can all be consistently estimated.  
\end{theorem}

\begin{proof}[Proof of Theorem~\ref{thm:normal}]

For simplicity, we write $D_i=(X_i,Y_i,T_i)$. By definition,
\$ 
& \hat{s}_{Y|X}^2 = \frac{1}{n_1}\sum_{i\in \cD_1}\big\{ \phi(D_i) - \hat\varphi(X_i)  \big\}^2 \\ 
&= \frac{1}{n_1}\sum_{i\in \cD_1} \big\{ \phi(D_i) - \varphi(X_i))\big\}^2
+ \frac{1}{n_1}\sum_{i\in \cD_1}  (\hat\varphi(X_i)- \varphi(X_i) )^2   - \frac{2}{n_1}\sum_{i\in \cD_1}  (\hat\varphi(X_i)- \varphi(X_i) )\cdot ( \phi(D_i) - \varphi(X_i) ).
\$
Since $\|\hat\varphi^{(k)}-\varphi\|_{L_2(P)} = o_P(n^{-1/4})$,  and by Markov's inequality, we know that for any fixed $\epsilon>0$, 
\$
&\PP\Bigg[ \bigg| \frac{1}{n_1/2}   \sum_{i\notin \cD_1^{(k)}}  (\hat\varphi(X_i)- \varphi(X_i) )^2  \bigg| >\epsilon \Bigggiven \cD_1^{(k)}\Bigg] \\
&\leq  \frac{2}{ \epsilon^2} \EE\Big[(\hat\varphi(X_i)- \varphi(X_i) )^2   \Bigggiven \cD_1^{(k)}\Big]  = 2\|\hat\varphi^{(k)}-\varphi\|_{L_2(P)} ^2 / \epsilon^2 =o_P(n^{-1/2}).
\$
This implies 
\$
\frac{1}{n_1}\sum_{i\in \cD_1}  (\hat\varphi(X_i)- \varphi(X_i) )^2 = o_P(1/\sqrt{n}).
\$
Also, conditional on $\cD_1^{(k)}$, note that  $ (\hat\varphi(X_i)- \varphi(X_i) )\cdot ( \phi(D_i) - \varphi(X_i) )$ is i.i.d.~with mean zero for all $i\notin \cD_1^{(k)}$. Thus, by Markov's inequality, we have 
\$
&\PP\Bigg[ \bigg| \frac{1}{n_1/2}   \sum_{i\notin \cD_1^{(k)}}   (\hat\varphi(X_i)- \varphi(X_i) )\cdot ( \phi(D_i) - \varphi(X_i) ) \bigg| >\epsilon \Bigggiven \cD_1^{(k)}\Bigg] \\
&\leq  \frac{1}{ \epsilon^2} \EE\Bigg[\bigg( \frac{1}{n_1/2}   \sum_{i\notin \cD_1^{(k)}}   (\hat\varphi(X_i)- \varphi(X_i) )\cdot ( \phi(D_i) - \varphi(X_i) ) \bigg)^2   \Bigggiven \cD_1^{(k)}\Bigg]  \\
&= \frac{4}{\epsilon^2 n_1^2}\sum_{i\notin \cD_i^{(k)}} \EE\Big[ (\hat\varphi(X_i)- \varphi(X_i) )\cdot ( \phi(D_i) - \varphi(X_i) ) ^2\Biggiven \cD_1^{(k)}\Big] \\ 
&= \frac{4}{\epsilon^2 n_1} \|\sigma\cdot(\hat\varphi^{(k)}-\varphi)\|_{L_2(P)}^2.
\$
Given that $\|\sigma\cdot(\hat\varphi^{(k)}-\varphi)\|_{L_2(P)} = o_P(1)$, we know  
\$
&\frac{1}{n_1}\sum_{i\in \cD_1}  (\hat\varphi(X_i)- \varphi(X_i) )\cdot ( \phi(D_i) - \varphi(X_i) )  = o_P(1/\sqrt{n}).
\$
This means 
\$
\hat{s}_{Y|X}^2 - s_{Y|X}^2 = \frac{1}{n_1}\sum_{i\in \cD_1} \big\{ (\phi(D_i) - \varphi(X_i))^2 - s_{Y|X}^2\big\}  + o_P(1/\sqrt{n}).
\$
Similarly, by definition, and due to the fact that $\frac{1}{n_1}\sum_{i\in \cD_1} \phi(D_i) -\EE_P[\phi] = O_P(1/\sqrt{n})$, 
\$
\hat{s}_X^2 &= \frac{1}{n_1}\sum_{i\in \cD_1} \hat\varphi(X_i) (2\phi(D_i) - \hat\varphi(X_i)) - \Big(\frac{1}{n_1}\sum_{i\in \cD_1} \phi(D_i)\Big)^2 \\ 
&= \frac{1}{n_1}\sum_{i\in \cD_1}  \varphi(X_i) (2\phi(D_i) -  \varphi(X_i)) + \frac{2}{n_1}\sum_{i\in \cD_1}  (\hat\varphi(X_i) - \varphi(X_i) ) ( \phi(D_i) - \varphi(X_i)) \\
&\qquad - \frac{1}{n_1}\sum_{i\in \cD_1}  (\hat\varphi(X_i) - \varphi(X_i) )^2 - \Big(\frac{1}{n_1}\sum_{i\in \cD_1} \phi(D_i)\Big)^2 \\ 
&= \frac{1}{n_1}\sum_{i\in \cD_1}  \varphi(X_i) (2\phi(D_i) -  \varphi(X_i))- \Big(\frac{1}{n_1}\sum_{i\in \cD_1} \phi(D_i)\Big)^2 +o_P(1/\sqrt{n}) \\
&= \frac{1}{n_1}\sum_{i\in \cD_1}  \varphi(X_i) (2\phi(D_i) -  \varphi(X_i))- (\EE_P[\phi])^2 - 2\EE_P[\phi]\Big(\frac{1}{n_1}\sum_{i\in \cD_1} \phi(D_i) - \EE_P[\phi]\Big)  +o_P(1/\sqrt{n}) \\ 
&= \frac{1}{n_1}\sum_{i\in \cD_1} \big\{ \varphi(X_i) (2\phi(D_i) -  \varphi(X_i)) -2\EE_P[\phi] \phi(D_i) + (\EE_P[\phi])^2 \big\}  +o_P(1/\sqrt{n}) ,
\$
which further implies 
\$
\hat{s}_X^2 - s_X^2 = \frac{1}{n_1}\sum_{i\in \cD_1} \big\{ \varphi(X_i) (2\phi(D_i) -  \varphi(X_i)) - \EE_P[\phi^2] -2\EE_P[\phi] (\phi(D_i) - \EE_P[\phi])  \big\}  +o_P(1/\sqrt{n}).
\$
By Delta method, the above two results imply 
\$
\hat{s}_{Y|X} - s_{Y|X} &= \frac{1}{n_1}\sum_{i\in \cD_1} \frac{1}{2s_{Y|X}}\big\{ (\phi(D_i) - \varphi(X_i))^2 - s_{Y|X}^2\big\}  + o_P(1/\sqrt{n}),\\
\hat{s}_X - s_X &= \frac{1}{n_1}\sum_{i\in \cD_1} \frac{1}{2s_X}\big\{ \varphi(X_i) (2\phi(D_i) -  \varphi(X_i)) - \EE_P[\phi^2] -2\EE_P[\phi] (\phi(D_i) - \EE_P[\phi])  \big\}  +o_P(1/\sqrt{n}).
\$
These give us the desired asymptotic expansion of the resulting estimators, with 
\$
\psi(X_i,Y_i,T_i) &= \begin{pmatrix}
    \frac{1}{2s_{Y|X}}\big\{ (\phi(D_i) - \varphi(X_i))^2 - s_{Y|X}^2\big\} \\ 
    \frac{1}{2s_X}\big\{ \varphi(X_i) (2\phi(D_i) -  \varphi(X_i)) - \EE_P[\phi^2] -2\EE_P[\phi] (\phi(D_i) - \EE_P[\phi])  \big\} 
\end{pmatrix}.
\$
We thus conclude the proof of Theorem~\ref{thm:normal}.
\end{proof}

%% file: main_new.bbl
\begin{thebibliography}{}

\bibitem[Bansak et~al., 2024]{bansak2024learning}
Bansak, K.~C., Paulson, E., and Rothenh{\"a}usler, D. (2024).
\newblock Learning under random distributional shifts.
\newblock In {\em International Conference on Artificial Intelligence and Statistics}, pages 3943--3951. PMLR.

\bibitem[Bareinboim and Pearl, 2016]{bareinboim2016fusion}
Bareinboim, E. and Pearl, J. (2016).
\newblock {Causal Inference and the Data-Fusion Problem}.
\newblock {\em Proceedings of the National Academy of Sciences}, 113(27):7345--7352.

\bibitem[Bickel et~al., 2007]{bickel2007discriminative}
Bickel, S., Br{\"u}ckner, M., and Scheffer, T. (2007).
\newblock Discriminative learning for differing training and test distributions.
\newblock In {\em Proceedings of the 24th international conference on Machine learning}, pages 81--88.

\bibitem[Buchanan et~al., 2018]{buchanan2018ipsw}
Buchanan, A.~L., Hudgens, M.~G., Cole, S.~R., Mollan, K.~R., Sax, P.~E., Daar, E.~S., Adimora, A.~A., Eron, J.~J., and Mugavero, M.~J. (2018).
\newblock {Generalizing Evidence From Randomized Trials Using Inverse Probability Of Sampling Weights}.
\newblock {\em Journal of the Royal Statistical Society: Series A (Statistics in Society)}, 181(4):1193--1209.

\bibitem[Cai et~al., 2023]{cai2023diagnosing}
Cai, T.~T., Namkoong, H., and Yadlowsky, S. (2023).
\newblock Diagnosing model performance under distribution shift.
\newblock {\em arXiv preprint arXiv:2303.02011}.

\bibitem[Chernozhukov et~al., 2018]{chernozhukov2018double}
Chernozhukov, V., Chetverikov, D., Demirer, M., Duflo, E., Hansen, C., Newey, W., and Robins, J. (2018).
\newblock Double/debiased machine learning for treatment and structural parameters.

\bibitem[Cole and Stuart, 2010]{cole2010generalizing}
Cole, S.~R. and Stuart, E.~A. (2010).
\newblock Generalizing evidence from randomized clinical trials to target populations: the actg 320 trial.
\newblock {\em American journal of epidemiology}, 172(1):107--115.

\bibitem[Colnet et~al., 2024]{colnet2024causal}
Colnet, B., Mayer, I., Chen, G., Dieng, A., Li, R., Varoquaux, G., Vert, J.-P., Josse, J., and Yang, S. (2024).
\newblock {Causal Inference Methods for Combining Randomized Trials and Observational Studies: A Review}.
\newblock {\em Statistical science}, 39(1):165--191.

\bibitem[Coppock et~al., 2018]{coppock2018generalizability}
Coppock, A., Leeper, T.~J., and Mullinix, K.~J. (2018).
\newblock Generalizability of heterogeneous treatment effect estimates across samples.
\newblock {\em Proceedings of the National Academy of Sciences}, 115(49):12441--12446.

\bibitem[Dahabreh et~al., 2020]{dahabreh2020extending}
Dahabreh, I.~J., Robertson, S.~E., Steingrimsson, J.~A., Stuart, E.~A., and Hernan, M.~A. (2020).
\newblock {Extending Inferences from A Randomized Trial to A New Target Population}.
\newblock {\em Statistics in medicine}, 39(14):1999--2014.

\bibitem[Dahabreh et~al., 2019]{dahabreh2019generalizing}
Dahabreh, I.~J., Robertson, S.~E., Tchetgen, E.~J., Stuart, E.~A., and Hern{\'a}n, M.~A. (2019).
\newblock Generalizing causal inferences from individuals in randomized trials to all trial-eligible individuals.
\newblock {\em Biometrics}, 75(2):685--694.

\bibitem[Deaton and Cartwright, 2018]{Deaton:2017kg}
Deaton, A. and Cartwright, N. (2018).
\newblock {Understanding and Misunderstanding Randomized Controlled Trials}.
\newblock {\em Social Science {\&} Medicine}.

\bibitem[Degtiar and Rose, 2023]{degtiar2023review}
Degtiar, I. and Rose, S. (2023).
\newblock {A Review of Generalizability and Transportability}.
\newblock {\em Annual Review of Statistics and Its Application}, 10(1):501--524.

\bibitem[Delios et~al., 2022]{delios2022examining}
Delios, A., Clemente, E.~G., Wu, T., Tan, H., Wang, Y., Gordon, M., Viganola, D., Chen, Z., Dreber, A., Johannesson, M., et~al. (2022).
\newblock Examining the generalizability of research findings from archival data.
\newblock {\em Proceedings of the National Academy of Sciences}, 119(30):e2120377119.

\bibitem[Deville and S{\"a}rndal, 1992]{deville1992}
Deville, J.-C. and S{\"a}rndal, C.-E. (1992).
\newblock {Calibration Estimators in Survey Sampling}.
\newblock {\em Journal of the American Statistical Association}, 87(418):376--382.

\bibitem[Egami and Hartman, 2021]{egami2021covariate}
Egami, N. and Hartman, E. (2021).
\newblock {Covariate Selection for Generalizing Experimental Results: Application to A Large-scale Development Program in Uganda}.
\newblock {\em Journal of the Royal Statistical Society Series A: Statistics in Society}, 184(4):1524--1548.

\bibitem[Egami and Hartman, 2023]{egami2023elements}
Egami, N. and Hartman, E. (2023).
\newblock Elements of external validity: Framework, design, and analysis.
\newblock {\em American Political Science Review}, 117(3):1070--1088.

\bibitem[Gama et~al., 2014]{gama2014survey}
Gama, J., {\v{Z}}liobait{\.e}, I., Bifet, A., Pechenizkiy, M., and Bouchachia, A. (2014).
\newblock A survey on concept drift adaptation.
\newblock {\em ACM computing surveys (CSUR)}, 46(4):1--37.

\bibitem[Hainmueller, 2012]{hainmueller2012entropy}
Hainmueller, J. (2012).
\newblock Entropy balancing for causal effects: A multivariate reweighting method to produce balanced samples in observational studies.
\newblock {\em Political analysis}, 20(1):25--46.

\bibitem[Hartman et~al., 2015]{hartman2015sate}
Hartman, E., Grieve, R., Ramsahai, R., and Sekhon, J.~S. (2015).
\newblock {From Sample Average Treatment Effect to Population Average Treatment Effect on the Treated}.
\newblock {\em Journal of the Royal Statistical Society. Series A (Statistics in Society)}, 178(3):757--778.

\bibitem[Holzmeister et~al., 2024]{holzmeister2024heterogeneity}
Holzmeister, F., Johannesson, M., B{\"o}hm, R., Dreber, A., Huber, J., and Kirchler, M. (2024).
\newblock Heterogeneity in effect size estimates.
\newblock {\em Proceedings of the National Academy of Sciences}, 121(32):e2403490121.

\bibitem[Horvitz and Thompson, 1952]{horvitz1952generalization}
Horvitz, D.~G. and Thompson, D.~J. (1952).
\newblock A generalization of sampling without replacement from a finite universe.
\newblock {\em Journal of the American statistical Association}, 47(260):663--685.

\bibitem[Hotz et~al., 2005]{hotz2005predicting}
Hotz, V.~J., Imbens, G.~W., and Mortimer, J.~H. (2005).
\newblock {Predicting the Efficacy of Future Training Programs Using Past Experiences at Other Locations}.
\newblock {\em Journal of Econometrics}, 125(1-2):241--270.

\bibitem[Hu and Hong, 2013]{hu2013kullback}
Hu, Z. and Hong, L.~J. (2013).
\newblock Kullback-leibler divergence constrained distributionally robust optimization.
\newblock {\em Available at Optimization Online}, 1(2):9.

\bibitem[Hudson, 2023]{hudson2023explicating}
Hudson, R. (2023).
\newblock Explicating exact versus conceptual replication.
\newblock {\em Erkenntnis}, 88(6):2493--2514.

\bibitem[Imai et~al., 2008]{imai2008misunderstandings}
Imai, K., King, G., and Stuart, E.~A. (2008).
\newblock {Misunderstandings Between Experimentalists and Observationalists About Causal Inference}.
\newblock {\em Journal of the Royal Statistical Society: Series A (Statistics in Society)}, 171(2):481--502.

\bibitem[Jeong and Rothenh{\"a}usler, 2022]{jeong2022calibrated}
Jeong, Y. and Rothenh{\"a}usler, D. (2022).
\newblock Calibrated inference: statistical inference that accounts for both sampling uncertainty and distributional uncertainty.
\newblock {\em arXiv preprint arXiv:2202.11886}.

\bibitem[Jeong and Rothenh{\"a}usler, 2024]{jeong2024out}
Jeong, Y. and Rothenh{\"a}usler, D. (2024).
\newblock Out-of-distribution generalization under random, dense distributional shifts.
\newblock {\em arXiv preprint arXiv:2404.18370}.

\bibitem[Jin et~al., 2023]{jin2023diagnosing}
Jin, Y., Guo, K., and Rothenh{\"a}usler, D. (2023).
\newblock Diagnosing the role of observable distribution shift in scientific replications.
\newblock {\em arXiv preprint arXiv:2309.01056}.

\bibitem[Jin and Rothenh{\"a}usler, 2024]{jin2024tailored}
Jin, Y. and Rothenh{\"a}usler, D. (2024).
\newblock Tailored inference for finite populations: conditional validity and transfer across distributions.
\newblock {\em Biometrika}, 111(1):215--233.

\bibitem[Kern et~al., 2016]{kern2016assessing}
Kern, H.~L., Stuart, E.~A., Hill, J., and Green, D.~P. (2016).
\newblock Assessing methods for generalizing experimental impact estimates to target populations.
\newblock {\em Journal of research on educational effectiveness}, 9(1):103--127.

\bibitem[Klein et~al., 2014]{klein2014investigating}
Klein, R.~A., Ratliff, K.~A., Vianello, M., Adams~Jr, R.~B., Bahn{\'\i}k, {\v{S}}., Bernstein, M.~J., Bocian, K., Brandt, M.~J., Brooks, B., Brumbaugh, C.~C., et~al. (2014).
\newblock {Investigating Variation in Replicability}.
\newblock {\em Social psychology}.

\bibitem[Klein et~al., 2018]{klein2018many}
Klein, R.~A., Vianello, M., Hasselman, F., Adams, B.~G., Adams~Jr, R.~B., Alper, S., Aveyard, M., Axt, J.~R., Babalola, M.~T., Bahn{\'\i}k, {\v{S}}., et~al. (2018).
\newblock Many labs 2: Investigating variation in replicability across samples and settings.
\newblock {\em Advances in Methods and Practices in Psychological Science}, 1(4):443--490.

\bibitem[Krefeld-Schwalb et~al., 2024]{krefeld2024exposing}
Krefeld-Schwalb, A., Sugerman, E.~R., and Johnson, E.~J. (2024).
\newblock Exposing omitted moderators: Explaining why effect sizes differ in the social sciences.
\newblock {\em Proceedings of the National Academy of Sciences}, 121(12):e2306281121.

\bibitem[Lu et~al., 2023]{lu2023you}
Lu, B., Ben-Michael, E., Feller, A., and Miratrix, L. (2023).
\newblock {Is It Who You Are or Where You Are? Accounting for Compositional Differences in Cross-Site Treatment Effect Variation}.
\newblock {\em Journal of Educational and Behavioral Statistics}, 48(4):420--453.

\bibitem[Lu et~al., 2018]{lu2018learning}
Lu, J., Liu, A., Dong, F., Gu, F., Gama, J., and Zhang, G. (2018).
\newblock Learning under concept drift: A review.
\newblock {\em IEEE transactions on knowledge and data engineering}, 31(12):2346--2363.

\bibitem[Madan et~al., 2016]{Madan_Uhlmann_Schweinsberg_Tierney_2016}
Madan, N., Uhlmann, E.~L., Schweinsberg, M., and Tierney, W. (2016).
\newblock The pipeline project.

\bibitem[McShane et~al., 2022]{mcshane2022modeling}
McShane, B.~B., B{\"o}ckenholt, U., and Hansen, K.~T. (2022).
\newblock Modeling and learning from variation and covariation.
\newblock {\em Journal of the American Statistical Association}, 117(540):1627--1630.

\bibitem[Miratrix et~al., 2018]{miratrix2018worth}
Miratrix, L.~W., Sekhon, J.~S., Theodoridis, A.~G., and Campos, L.~F. (2018).
\newblock {Worth Weighting? How to Think About and Use Weights in Survey Experiments}.
\newblock {\em Political Analysis}, 26(3):275--291.

\bibitem[Nosek et~al., 2002]{nosek2002math}
Nosek, B.~A., Banaji, M.~R., and Greenwald, A.~G. (2002).
\newblock {Math = Male, Me= Female, Therefore Math $\ne$ Me.}
\newblock {\em Journal of personality and social psychology}, 83(1):44.

\bibitem[Pan and Yang, 2009]{pan2009survey}
Pan, S.~J. and Yang, Q. (2009).
\newblock A survey on transfer learning.
\newblock {\em IEEE Transactions on knowledge and data engineering}, 22(10):1345--1359.

\bibitem[Quinonero-Candela et~al., 2008]{quinonero2008dataset}
Quinonero-Candela, J., Sugiyama, M., Schwaighofer, A., and Lawrence, N.~D. (2008).
\newblock {\em Dataset Shift in Machine Learning}.
\newblock MIT Press.

\bibitem[Robins et~al., 1994]{robins1994estimation}
Robins, J.~M., Rotnitzky, A., and Zhao, L.~P. (1994).
\newblock Estimation of regression coefficients when some regressors are not always observed.
\newblock {\em Journal of the American statistical Association}, 89(427):846--866.

\bibitem[S{\"a}rndal et~al., 2003]{sarndal2003model}
S{\"a}rndal, C.-E., Swensson, B., and Wretman, J. (2003).
\newblock {\em {Model Assisted Survey Sampling}}.
\newblock Springer Science \& Business Media.

\bibitem[Schweinsberg et~al., 2016]{schweinsberg2016pipeline}
Schweinsberg, M., Madan, N., Vianello, M., Sommer, S.~A., Jordan, J., Tierney, W., Awtrey, E., Zhu, L.~L., Diermeier, D., Heinze, J.~E., et~al. (2016).
\newblock The pipeline project: Pre-publication independent replications of a single laboratory's research pipeline.
\newblock {\em Journal of Experimental Social Psychology}, 66:55--67.

\bibitem[Shadish et~al., 2002]{shadish2002}
Shadish, W.~R., Cook, T.~D., and Campbell, D.~T. (2002).
\newblock {\em {Experimental and Quasi-Experimental Designs for Generalized Causal Inference}}.
\newblock Boston: Houghton Mifflin.

\bibitem[Shimodaira, 2000]{shimodaira2000improving}
Shimodaira, H. (2000).
\newblock Improving predictive inference under covariate shift by weighting the log-likelihood function.
\newblock {\em Journal of statistical planning and inference}, 90(2):227--244.

\bibitem[Stroebe and Strack, 2014]{stroebe2014alleged}
Stroebe, W. and Strack, F. (2014).
\newblock The alleged crisis and the illusion of exact replication.
\newblock {\em Perspectives on Psychological Science}, 9(1):59--71.

\bibitem[Stuart et~al., 2011]{stuart2011use}
Stuart, E.~A., Cole, S.~R., Bradshaw, C.~P., and Leaf, P.~J. (2011).
\newblock The use of propensity scores to assess the generalizability of results from randomized trials.
\newblock {\em Journal of the Royal Statistical Society Series A: Statistics in Society}, 174(2):369--386.

\bibitem[Tipton, 2013]{Tipton:2013ew}
Tipton, E. (2013).
\newblock {Improving Generalizations From Experiments Using Propensity Score Subclassification: Assumptions, Properties, and Contexts}.
\newblock {\em Journal of Educational and Behavioral Statistics}, 38(3):239--266.

\bibitem[Tipton et~al., 2014]{tipton2014sample}
Tipton, E., Hedges, L., Vaden-Kiernan, M., Borman, G., Sullivan, K., and Caverly, S. (2014).
\newblock {Sample Selection in Randomized Experiments: A New Method Using Propensity Score Stratified Sampling}.
\newblock {\em Journal of Research on Educational Effectiveness}, 7(1):114--135.

\bibitem[Tversky and Kahneman, 1981]{tversky1981framing}
Tversky, A. and Kahneman, D. (1981).
\newblock {The Framing of Decisions and the Psychology of Choice}.
\newblock {\em science}, 211(4481):453--458.

\end{thebibliography}
